\documentclass[aps,pra,superscriptaddress,footinbib,twocolumn]{revtex4-2}
\usepackage{amsmath,amsthm,amsfonts,amssymb,amscd} 
\usepackage[utf8]{inputenc}
\usepackage{indentfirst}
\usepackage{graphicx}
\usepackage{color}
\usepackage[T1]{fontenc}
\usepackage{lmodern}
\usepackage[american,ngerman,greek,british]{babel}
\usepackage{float}
\usepackage[dvipsnames]{xcolor}
\usepackage[colorlinks=true,citecolor=Mulberry,linkcolor=prettygreen,urlcolor=MidnightBlue,hyperindex,breaklinks]{hyperref}
\usepackage{cleveref}
\usepackage{mathtools}
\usepackage{bm}
\usepackage{ulem}
\usepackage{enumerate}
\usepackage{sidecap}
\usepackage{microtype}
\usepackage{enumitem}
\usepackage{bbold}
\usepackage{listings}
\usepackage{algorithm}
\usepackage{algpseudocode}
\usepackage{multirow}
\usepackage{setspace}
\usepackage{soul}

\definecolor{prettygreen}{RGB}{5,125,143}

\newcommand{\id}{\mathbb{1}}

\newcommand{\rd}{\text{\normalfont r\,}}

\newcommand{\Ical}{\mathcal{I}}

\newcommand{\Xcal}{\mathcal{X}}
\newcommand{\Ycal}{\mathcal{Y}}

\newcommand{\Aoo}{\rm I_{\rm A }}

\newcommand{\Aot}{\rm {O'}_{\rm A}}

\newcommand{\Ia}{{\rm I_{\rm A}}}
\newcommand{\Iap}{{\rm {I'}_{\rm A}}}
\newcommand{\Oa}{{\rm O_{\rm A}}}
\newcommand{\Oap}{{\rm {O'}_{\rm A}}}

\newcommand{\Ib}{{\rm I_{\rm B}}}
\newcommand{\Ibp}{{\rm {I'}_{\rm B}}}
\newcommand{\Ob}{{\rm O_{\rm B}}}
\newcommand{\Obp}{{\rm {O'}_{\rm B}}}

\renewcommand{\Iat}{{\td{\rm I}_{\rm A}}}

\newcommand{\Oat}{{\td{\rm O}_{\rm A}}}

\newcommand{\Ibt}{{\td{\rm I}_{\rm B}}}

\newcommand{\Obt}{{\td{\rm O}_{\rm B}}}

\newcommand{\tIa}{{\td{\rm I}_{\rm A}}}

\newcommand{\tOa}{{\td{\rm O}_{\rm A}}}
\newcommand{\tOap}{{\td{\rm {O'}}_{\rm A}}}

\newcommand{\tIb}{{\td{\rm I}_{\rm B}}}
\newcommand{\tIbp}{{\td{\rm {I'}}_{\rm B}}}
\newcommand{\tOb}{{\td{\rm O}_{\rm B}}}
\newcommand{\tObp}{{\td{\rm {O'}}_{\rm B}}}

\newcommand{\I}{{\rm I}}
\newcommand{\Ip}{{\rm I'}}
\renewcommand{\O}{{\rm O}}
\newcommand{\Op}{{\rm O'}}
\newcommand{\A}{{\rm A}}
\newcommand{\B}{{\rm B}}
\newcommand{\C}{{\rm C}}
\newcommand{\D}{{\rm D}}
\newcommand{\X}{{\rm X}}
\newcommand{\Y}{{\rm Y}}
\newcommand{\Yp}{{\rm Y'}}

\newcommand{\W}{W}

\newcommand{\PP}{\mathcal P}

\newcommand{\td}{\widetilde }

\newtheorem{theorem}{Theorem}%[chapter]
\newtheorem*{theorem*}{Theorem}
\newtheorem{lemma}{Lemma}%[chapter]
\newtheorem{fact}{Fact}% [chapter]
\begin{document}

\author{Jessica Bavaresco} 
\affiliation{Department of Applied Physics, University of Geneva, 1205 Geneva, Switzerland}

\author{Ämin Baumeler} 
\affiliation{Facoltà di scienze informatiche, Università della Svizzera italiana, 6900 Lugano, Switzerland}
\affiliation{Facoltà indipendente di Gandria, 6978 Gandria, Switzerland}

\author{Yelena Guryanova} 
\affiliation{QuantumBasel, Schorenweg 44b, 4144 Arlesheim, Switzerland}
\affiliation{Center for Quantum Computing and Quantum Coherence (QC2), University of Basel, Petersplatz 1, Basel, 4001, Switzerland}

\author{Costantino Budroni} 
\affiliation{Department of Physics ``E.~Fermi'', University of Pisa, Largo B.~Pontecorvo 3, 56127 Pisa, Italy}

\date{\today}

\title{Indefinite causal order in boxworld theories} 

\begin{abstract}
An astonishing feature of higher-order quantum theory is that it can accommodate indefinite causal order. In the simplest bipartite setting, there exist signaling correlations for which it is fundamentally impossible to ascribe a definite causal order for the parties’ actions. Moreover, the assumptions required to arrive at such a statement (local quantum transformations and well-behaved probabilities) result in a nontrivial set of correlations, whose boundary is, to date, uncharacterized. In this work, we investigate indefinite causal order in boxworld theories. We construct a higher-order theory whose descriptor is the generalized bit (gbit)---the natural successor of the classical bit and quantum qubit. By fixing the local transformations in boxworld and asking about the global causal structure, we find that we trivially recover the full set of two-way signaling correlations. In light of this, we motivate and propose two physical principles in order to limit the set of achievable correlations: nonsignaling preservation and no signaling without system exchange. We find that a higher-order boxworld theory that respects these physical principles leads to (i) a nontrivial set of achievable correlations and (ii) a violation of some causal inequalities that is higher than what can be achieved in higher-order quantum theory. These results lead us to conjecture that the set of correlations of our higher-order boxworld theory is an outer approximation to the set of correlations produced by higher-order quantum theory.
\end{abstract}

\maketitle

%%%%%%%%%%%%%%%%%%%%%%%%%%%%%%%%%%%%%%%%%%%%
\section{Introduction}\label{sec::intro}
%%%%%%%%%%%%%%%%%%%%%%%%%%%%%%%%%%%%%%%%%%%%

Statistical data of physical experiments involving independent parties, often referred to as \textit{correlations}, can be used to infer physical properties of the underlying systems involved in the experiment. Of particular interest to quantum information science is the set of \textit{nonsignaling} correlations, those that do not allow for the transmission of information between the parties. In an experiment where two parties, Alice and Bob, make independent measurements on their share of a bipartite quantum state, a so-called Bell scenario~\cite{bell1964onthe,clauser1969proposed,brunner2014bell}, the collected statistics always correspond to nonsignaling correlations. When violating a Bell inequality, these correlations certify the presence of entanglement in the shared quantum state as well as the incompatibility of the applied local measurements~\cite{werner1989quantum,quintino2014joint,uola2014joint}. 

It has long been known that not all nonsignaling correlations can be created in a Bell scenario~\cite{tsirelson1980quantum}. Several physical and information-theoretic principles---e.g., nontrivial communication complexity~\cite{vandam2005implausible,brassard2006limit}, information causality~\cite{pawlowski2009information}, macroscopic locality~\cite{navascues2009glance}, local orthogonality~\cite{fritz2013local}, and E-principle~\cite{cabello2013simple}, to name a few---have been proposed and studied as candidates for outlining the set of nonsignaling correlations that can be achieved by quantum systems.

%%%%%%%%%%%%%%%%%%%%%%%%%%%%%%%%%%%%%%%%%%%
\begin{figure}[h!]
\begin{center}
	\includegraphics[width=0.8\columnwidth]{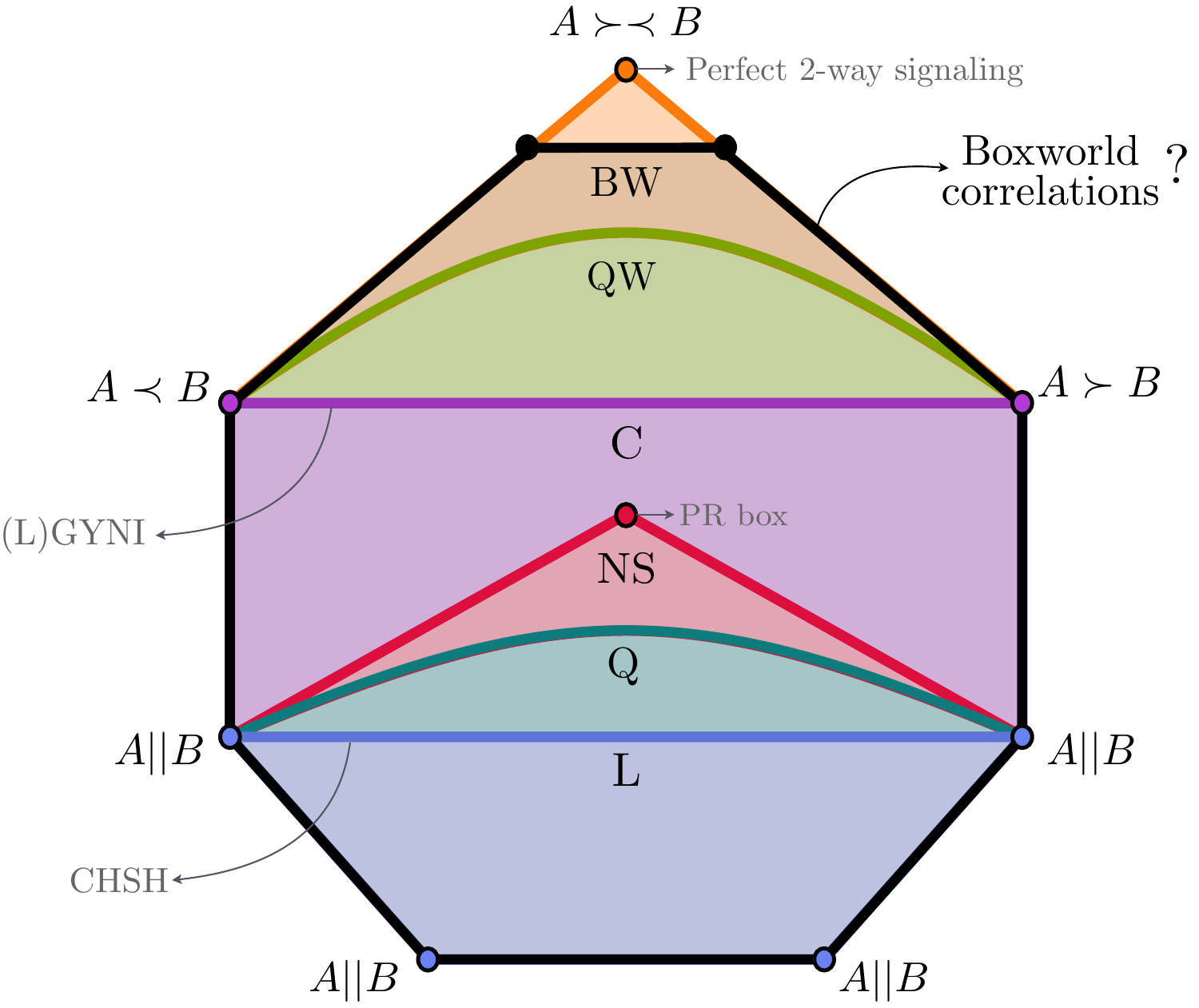}
	\caption{\textbf{A schematic of the polytope of all correlations.} Depiction of the polytope of all correlations $P_{\A\B|\X\Y}$ in a scenario with 2 parties, 2 inputs, and 2 outputs. The extremal points are deterministic probability distributions, which can be nonsignaling ($A||B$), one-way signaling ($A\prec B$ or $A\succ B$), or two-way signaling ($A\succ\prec B$). The convex hull of the deterministic nonsignaling correlations forms the set L of local correlations, with the nontrivial facet given by the CHSH inequality. Q is the convex set of correlations that can be achieved by quantum mechanics in a Bell scenario. NS is the polytope of nonsignaling correlations, whose extremal points are deterministic nonsignaling correlations and PR-boxes. The convex hull of deterministic nonsignaling and one-way signaling correlations forms the set C of causal correlations, with the nontrivial facets given by the GYNI and LGYNI inequalities. QW is the convex set of correlations that can be achieved by process matrices. Finally, we expect that BW, the polytope of boxworld correlations, is an outer approximation of the set QW.}
\label{fig::polytopes}
\end{center}
\end{figure}
%%%%%%%%%%%%%%%%%%%%%%%%%%%%%%%%%%%%%%%%%%%

Arguably less effort has been dedicated to date to the study of correlations that can be achieved by quantum systems in \textit{signaling} scenarios, in which communication, both classical and quantum, is allowed~\cite{ oreshkov2012quantum,branciard2016simplest,feix2016causally,baumeler2016space,araujo2017purification,vanderlugt2023device,wechs2023existence,kunjwal2023nonclassicality,liu2024tsirelson}. Such scenarios have been studied in the context of correlations generated by quantum higher-order operations with indefinite causal order~\cite{chiribella2008quantum,chiribella2008transforming,chiribella2009theoretical,chiribella2013quantum,bisio2019theoretical,milz2024characterising}.

In this scenario, Alice and Bob are each considered to receive in their respective labs an input quantum system, and to apply onto them the most general form of quantum operations: quantum instruments. These operations locally produce a classical outcome, as in a Bell scenario, as well as a quantum output system, which is sent out of their labs. All resources that bridge Alice's and Bob's laboratories are characterized by a \textit{process matrix}~\cite{oreshkov2012quantum}. Since a process matrix may in general include channels connecting the parties, it allows for communication, hence, correlations in this scenario may be signaling. The statistics over the local classical inputs and outputs of Alice and Bob constitute the correlations of the experiment. In this sense, a process matrix can therefore be viewed as a higher-order transformation that maps local quantum instruments into sets of probability distributions---the process matrix correlations.

Although, in this scenario, Alice and Bob can only interact with their input systems and prepare output systems once per round of the experiment, process matrices are known to generate correlations which display \textit{genuine two-way signaling}. That is, correlations that cannot be described by either exclusively Alice sending information to Bob or exclusively Bob sending information to Alice in each round of the experiment, nor any convex combination thereof. Such correlations can violate causal inequalities~\cite{oreshkov2012quantum,branciard2016simplest}, the analogs of Bell inequalities, and therefore certify the presence of indefinite causal order in the underlying process matrix. Little is known about the set of genuine two-way signaling correlations (also called \textit{noncausal correlations}) that can be generated by process matrices. It has been shown that perfect two-way signaling cannot be achieved by any finite or infinite dimensional process matrix~\cite{bavaresco2019semi, kunjwal2023nonclassicality} and, more recently, upper bounds on the violation of causal inequalities by process matrices have been established in Ref.~\cite{liu2024tsirelson}. Despite these efforts, the boundary of the set of process matrix correlations remains largely unknown.

In this work, we propose physical principles to bound the set of process matrix correlations within the set of all probability distributions. The approach we take is based on a \textit{boxworld theory}, a form of general probabilistic theory which considers generalized nonsignaling systems, called gbits, as its basic states~\cite{BarrettPRA2007}. We develop a fully fledged higher-order boxworld theory, which allows processes that can produce noncausal correlations. We investigate and define the local transformations and processes of the theory, exploring their different causal constraints, and characterizing the polytope of correlations they can generate.

Our initial higher-order boxworld theory, based mostly on mathematical considerations of admissibility, is more general than quantum (and classical) higher-order theories, and we find that the correlations it allows not only include the set of process matrix correlations, but in fact all sets of probability distributions, which are depicted in Fig.~\ref{fig::polytopes}. Determined to arrive at a theory whose correlations are a nontrivial outer approximation to the set of process matrix correlations, we impose different constraints on the objects of our theory, which we motivate as physical principles. The two physical principles we explore are based on \textit{nonsignaling preservation} (NSP) and \textit{no signaling without system exchange} (NSWSE)~\cite{chiribella2010probabilistic}. 
We find that a higher-order boxworld theory that respects both of these physical principles cannot produce all sets of probability distributions as correlations. In particular, it cannot produce perfect two-way signaling. Nevertheless, it allows  correlations that can attain a violation of some causal inequalities which is higher than what can be attained by process matrices. Therefore, a higher-order boxworld theory that respect NSP and NSWSE proves itself to be a good candidate for an outer approximation of the set of process matrix correlations with indefinite causal order, and hence, of the set of genuinely two-way signaling correlations that can be attained in higher-order quantum theory.

%%%%%%%%%%%%%%%%%%%%%%%%%%%%%%%%%%%%%%%%%%%%
\section{Results}\label{sec::results}
%%%%%%%%%%%%%%%%%%%%%%%%%%%%%%%%%%%%%%%%%%%%

%%%%%%%%%%%%%%
\subsection{Basic elements: states and operations}

Boxworld \cite{BarrettPRA2007} refers to a generalized probability theory where the basic systems are probability distributions, i.e., the state of a system is described by a tensor $P_{\O|\I}$ with each entry $P_{\O|\I}(o|i)$ being a nonnegative real number that represents the probability of a classical outcome $o\in\O$ given a classical input $i\in\I$, where $\O,\I$ are random variables. See more details about the notation and tensor operations in Sec.~\ref{sec::methods}. In the bipartite case, the systems of the theory are described by states corresponding to nonsignaling boxes. These are tensors of the form $P^\text{NS}_{\Oa\Ob|\Ia\Ib}$, that is, sets of conditional, joint probability distributions over outputs $\Oa,\Ob$ given inputs $\Ia,\Ib$, which satisfy
\begin{align}
    \sum_{\Ob} P^\text{NS}_{\Oa\Ob|\Ia\Ib} &= P_{\Oa|\Ia}, \label{eq::probNS_AprecB} \\
    \sum_{\Oa} P^\text{NS}_{\Oa\Ob|\Ia\Ib} &= P_{\Ob|\Ib}, \label{eq::probNS_BprecA}
\end{align}
where $P_{\Oa|\Ia}$ and $P_{\Ob|\Ib}$ are marginal probability distributions that are independent of $\Ib$ and $\Ia$, respectively.
In this notation, a summation over a random variable $\X$ represents a summation over all values $x\in\X$ that the random variable can assume, and we leave the entry $P_{\O|\I}(o|i)$ of the tensor $P_{\O|\I}$ implicit for sake of simplicity. Hence, throughout the paper, we will denote $\sum_{o\in\O}P_\O(o)$ as simply $\sum_\O P_\O$. We identify each tensor $P_{\Oa\Ob|\Ia\Ib}$ with the hypothetical device generating it, so we refer to $P_{\Oa\Ob|\Ia\Ib}$ as a \textit{box}.

Local operations $T$ are those that map states into states, $T: P_{\O|\I} \mapsto P_{\Op|\Ip}$. That is, they correspond to transformations between probability distributions. When applied to part of a bipartite box, we require that local operations do not introduce signaling, hence they must map nonsignaling boxes to nonsignaling boxes. As shown in Ref.~\cite{rosset2020algebraic}, all local operations satisfy this property of nonsignaling preservation when acting on a part of bipartite box. Moreover, it was shown that local operations can be expressed as tensors $T=T_{\I\Op|\Ip\O}$ that also correspond to probability distributions, and that can further be decomposed, without loss of generality, into local pre- and post-processing operations~\cite{rosset2020algebraic}. In turn, these can be decomposed into a convex combination of deterministic sets of probability distributions, according to
\begin{equation}\label{eq::op_deterministic}
    T_{\I\Op|\Ip\O} = P_{\I|\Ip}\,P_{\Op|\Ip\I\O} = \sum_\lambda \pi_\lambda\, D^\lambda_{\I|\Ip}\,D^\lambda_{\Op|\Ip\O},
\end{equation}
where $\pi_\lambda\geq0$, $\sum_\lambda \pi_\lambda=1$ are convex weights and $D^\lambda_{\td{\O}|\td{\I}}$ is, for each $\lambda$, a set of deterministic probability distributions from inputs $\td{\I}$ to outputs $\td{\O}$ (with entries either $0$ or $1$). These local operations are the analog of completely-positive, trace-preserving maps in quantum theory, (i.e., quantum channels) which deterministically map quantum states into quantum states. Hence, we will refer to these local operations in boxworld as deterministic operations, even though they cannot always be described by sets of deterministic probability distributions (but rather convex combinations thereof). A graphical representation of deterministic operations is given in Fig.~\ref{fig::operations}.

%%%%%%%%%%%%%%%%%%%%%%%%%%%%%%%%%%%%%%%%%%%
\begin{figure}%[h!]
\begin{center}
	\includegraphics[width=\columnwidth]{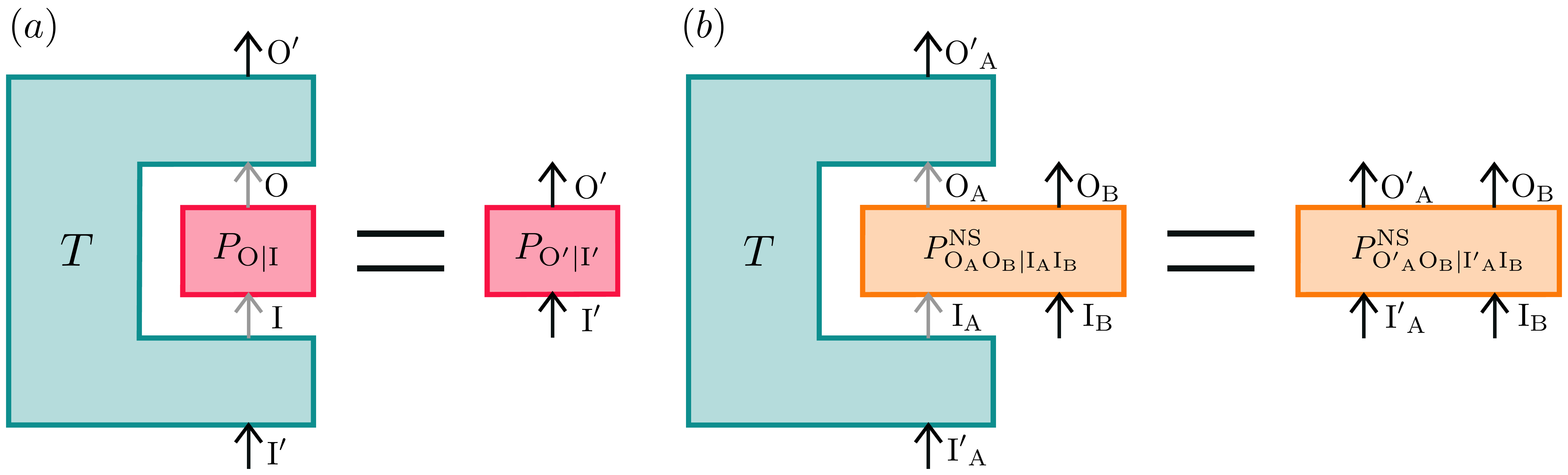}
	\caption{\textbf{Deterministic operations.} Deterministic operations $T$ are defined as transformations that map boxes into boxes. As shown in Ref.~\cite{rosset2020algebraic}, this is equivalent to requiring that local operations map bipartite nonsignaling boxes into bipartite nonsignaling, even when acting on only a part of the box.}
\label{fig::operations}
\end{center}
\end{figure}
%%%%%%%%%%%%%%%%%%%%%%%%%%%%%%%%%%%%%%%%%%%

For the probabilistic version of local operations, which are the analog of quantum instruments in quantum theory, we take any nonnegative resolution of a deterministic operation. Therefore, a set of probabilistic operations $T^{\A|\X}$ is a set of tensors with nonnegative entries that satisfy $\sum_\A T^{\A|\X}=T^\X$, where each $T^\X$ satisfies Eq.~\ref{eq::op_deterministic}. Elements of probabilistic operations can be decomposed as 
\begin{equation}\label{eq::op_probabilistic}
    T^{\A|\X}_{\I\Op|\Ip\O} = \sum_\lambda \pi_\lambda\, D^\lambda_{\I|\Ip \X}\,D^\lambda_{\Op|\Ip\O \X}\,D^\lambda_{\A|\Ip\O \X}.
\end{equation}
This definition of probabilistic operations is the most general one that is compatible with the definition of deterministic operations, which are the most general transformations that preserve completely nonsignaling input boxes. Alternatively, one may formalize probabilistic operations within the framework of generalized or operational probabilistic theories~\cite{BarrettPRA2007,chiribella2010probabilistic} (see Refs.~\cite{gptreview,optreview} for reviews). Transiting from there to the higher-order theories, as also proposed in Ref.~\cite{apadula2024}, poses its own challenges.
Here, we sidetrack these challenges by employing the most general probabilistic operations compatible with deterministic ones%
~\footnote{
At the same time, this definition allows for some counterintuitive instruments. For simplicity, let the random variables~$\X$,~$\O$, and~$\I$ be trivial. In that case,~$T^A$ corresponds to a state preparation with a classical output~$\A$, and according to Eq.~\eqref{eq::op_probabilistic},~$T^A$ decomposes as~$T^{\A}_{\Op|\Ip} = \sum_\lambda \pi_\lambda\, D^\lambda_{\Op|\Ip}\,D^\lambda_{\A|\Ip}$. An example of such a probabilistic operation is~$T^{\A}_{\Op|\Ip} = P_{\Op|\Ip}D^\text{id}_{\A|\Ip}$, where~$P_{\Op|\Ip}$ is an arbitrary gbit, and~$D^\text{id}_{\A|\Ip}$ the ``identity distribution'':~$D^\text{id}_{\A|\Ip}(a|a)=1$. This operation is counterintuitive because the output~$\A$ is correlated to~$\Ip$, and the random variable~$\Ip$ may be specified at a later time. One could imagine that~$P_{\Op|\Ip}$ is sent to a distant party who specifies~$\Ip$. This effectively introduces a communication channel in the reverse direction, and as we will see later, allows for perfect two-way signaling.}.

%%%%%%%%%%%%%%
\subsection{Process tensors}

%%%%%%%%%%%%%%%%%%%%%%%%%%%%%%%%%%%%%%%%%%%
\begin{figure}%[h!]
\begin{center}
	\includegraphics[width=\columnwidth]{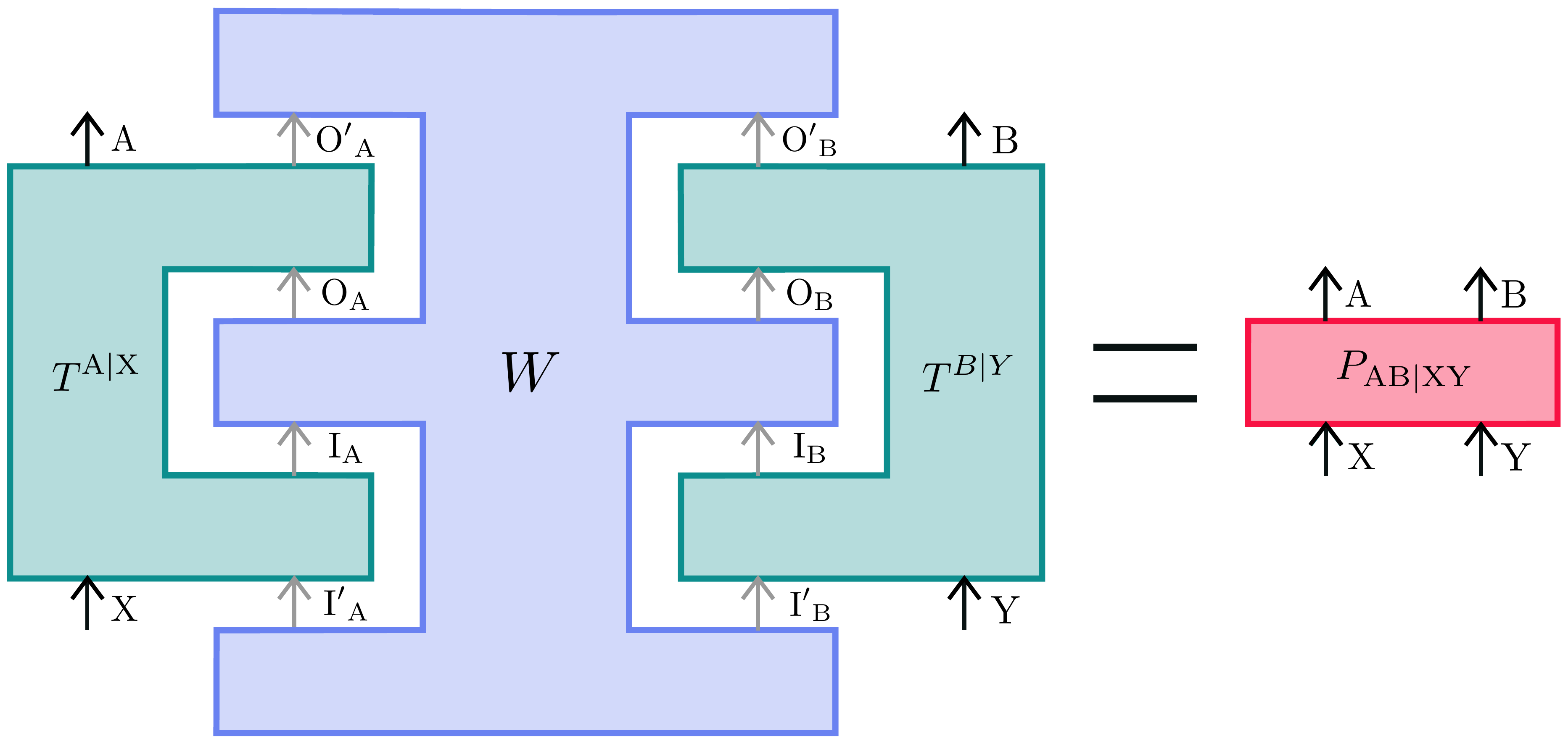}
	\caption{\textbf{Process tensors.} Process tensors $W$ are defined as the most general tensors that transform independent sets of local probabilistic operations $T^{\A|\X}$ and $T^{\B|\Y}$ into correlations $P_{\A\B|\X\Y}$.}
\label{fig::processtensor}
\end{center}
\end{figure}
%%%%%%%%%%%%%%%%%%%%%%%%%%%%%%%%%%%%%%%%%%%

Having characterized the local probabilistic operations of our boxworld theory, we now define a higher-order transformation that acts on such local transformations, giving it the name of \textit{process tensor}. A process tensor is the most general bilinear transformation that maps local probabilistic operations into probability distributions, the correlations. That is, a bipartite process tensor $W=W_{\Iap\Oa\Ibp\Ob|\Ia\Oap\Ib\Obp}$ is the most general tensor that satisfies
\begin{equation}\label{eq::bornrule}
     W * T^{\A|\X} * T^{\B|\Y} = P_{\A\B|\X\Y},
\end{equation}
for any set of probabilistic local operations $T^{\A|\X}=T^{\A|\X}_{\Ia\Oap|\Iap\Oa}$ and $T^{\B|\Y}=T^{\B|\Y}_{\Ib\Obp|\Ibp\Ob}$, where $P_{\A\B|\X\Y}$ are valid sets of probability distributions, and where the symbol~$*$ is used to represent a tensor contraction over random variables that are common to two tensors (see Sec.~\ref{sec::methods}). This definition is illustrated in Fig.~\ref{fig::processtensor}, while a representation of process tensors in terms of their random variables and in terms of their local box systems is given in Fig.~\ref{fig::processtensor_representations}.  In Sec.~\ref{sec::methods}, we provide a characterization of process tensors in terms of positivity and linear constraints, which is proven in App.~\ref{app::processtensors} (Theorem~\ref{thm::processtensor}), where we also show that a process tensor is in itself a set of probability distributions. The analogous object in quantum theory would be a process matrix, which is a representation of the most general transformation that maps local sets of quantum instruments to probability distributions~\cite{oreshkov2012quantum,araujo2015witnessing}. 

Conversely, a set of correlations $P_{\A\B|\X\Y}$ are called \textit{process tensor correlations} when there exist local operations $T^{\A|\X}$ and $T^{\B|\Y}$, and a process tensor $W$, such that Eq.~\eqref{eq::bornrule} holds. Hence, process tensor correlations are those that can be generated by local operations acting on process tensors. The parallel in the process matrix formalism is the set of correlations that can be generated by acting with local quantum instruments on a process matrix, here called the set of \textit{process matrix correlations}.

%%%%%%%%%%%%%%%%%%%%%%%%%%%%%%%%%%%%%%%%%%%
\begin{figure}%[h!]
\begin{center}
	\includegraphics[width=\columnwidth]{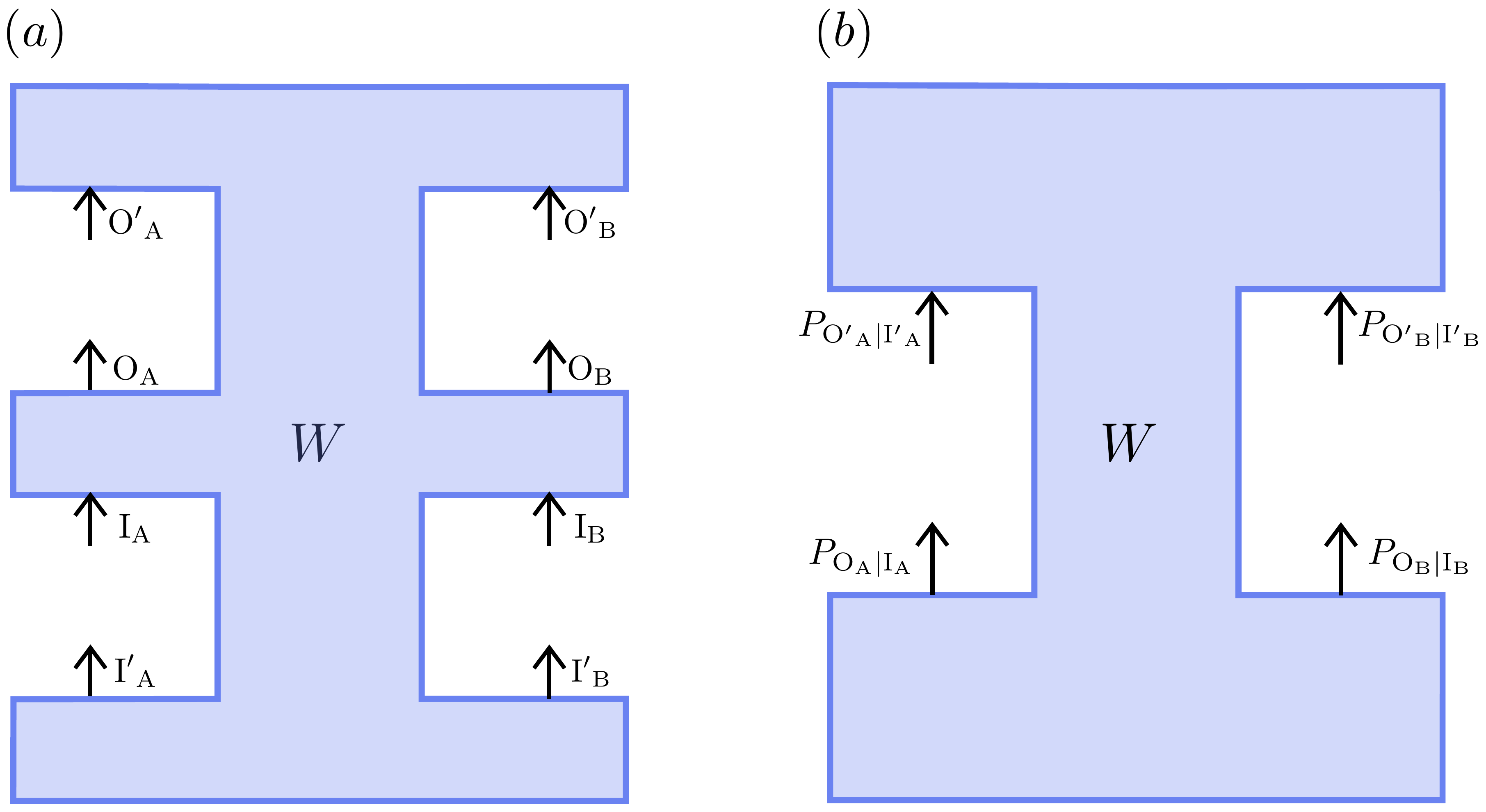}
	\caption{\textbf{Representation of process tensors.} A process tensor $W$ can be equivalently depicted in two ways: (a) with its wires representing the random variables they act on, or (b) with its wires representing local box systems.} 
\label{fig::processtensor_representations}
\end{center}
\end{figure}
%%%%%%%%%%%%%%%%%%%%%%%%%%%%%%%%%%%%%%%%%%%

The first question then is: what sets of bipartite joint conditional probability distributions correspond to process tensor correlations? We show that, in fact, any valid set of probability distributions correspond to process tensor correlations. Take for example the particular process tensor and local operations given by
\begin{align}
    W_{\diamond} &= \delta_{\Iap,\phi}\,\delta_{\Ibp,\phi}\,P_{\Oa\Ob|\Ia\Ib} \label{eq::ansatzW1} \\
    T^{\A|\X}_{\diamond} &= \delta_{\Ia,\X}\,\delta_{\Oap,\phi}\,\delta_{\A,\Oa} \label{eq::ansatzTax1} \\
    T^{\B|\Y}_{\diamond} &= \delta_{\Ib,\Y}\,\delta_{\Obp,\phi}\,\delta_{\B,\Ob}, \label{eq::ansatzTby1}
\end{align}
where $\phi$ is a constant, $\delta_{a,b}$ is the Kronecker delta function, and $P_{\Oa\Ob|\Ia\Ib}$ is any set of probability distributions. By comparison with Eq.~\eqref{eq::op_probabilistic}, one can check that $T^{\A|\X}_{\diamond}$ in Eq.~\eqref{eq::ansatzTax1} and $T^{\B|\Y}_{\diamond}$ in Eq.~\eqref{eq::ansatzTby1} are valid operations. In App.~\ref{app::processtensors} we show that $W_{\diamond}$ in Eq.~\eqref{eq::ansatzW1} is a valid process tensor. This process tensor and local operations are represented in Fig.~\ref{fig::W_diamond}. It is straightforward to see that the correlations resulting from the action of these operations on this process tensor will be $P_{\A\B|\X\Y}(a,b|x,y)=P_{\Oa\Ob|\Ia\Ib}(a,b|x,y)$ for all $a\in\A,b\in\B,x\in\X$, and $y\in\Y$. Hence, all sets of probability distributions are correlations that can be produced by process tensors. 

A few problems arise from this fact. First, although process tensor correlations form an outer approximation to the set of process matrix correlations, it is a trivial one---the set of all correlations. Second, in this theory, we find that the operations in Eqs.~\eqref{eq::ansatzTax1} and~\eqref{eq::ansatzTby1}, which are constant operations that correspond to a simple relabeling of random variables, and a process tensor such as the one in Eq.~\eqref{eq::ansatzW1}, which corresponds to a simple joint state preparation and does not allow for any communication between Alice and Bob, can be used to generate all sets of correlations, including perfect two-way signaling correlations
\begin{equation}\label{eq::2-sigcorr}
    P^\text{2-sig}_{\A\B|\X\Y} \coloneqq \delta_{\A,\Y}\,\delta_{\B,\X}. 
\end{equation}

This extremal point cannot be obtained in quantum theory, in other words, it is not in the set of process matrix correlations~\cite{bavaresco2019semi, kunjwal2023nonclassicality}.
It is worth noting that to generate \textit{any} correlations when in possession of the process tensor $W_\triangle$, Alice and Bob are not required to apply any meaningful operations nor exchange any systems or information; they can simply read off the desired correlations from $W_\triangle$. 

%%%%%%%%%%%%%%%%%%%%%%%%%%%%%%%%%%%%%%%%%%%
\begin{figure}%[h!]
\begin{center}
	\includegraphics[width=\columnwidth]{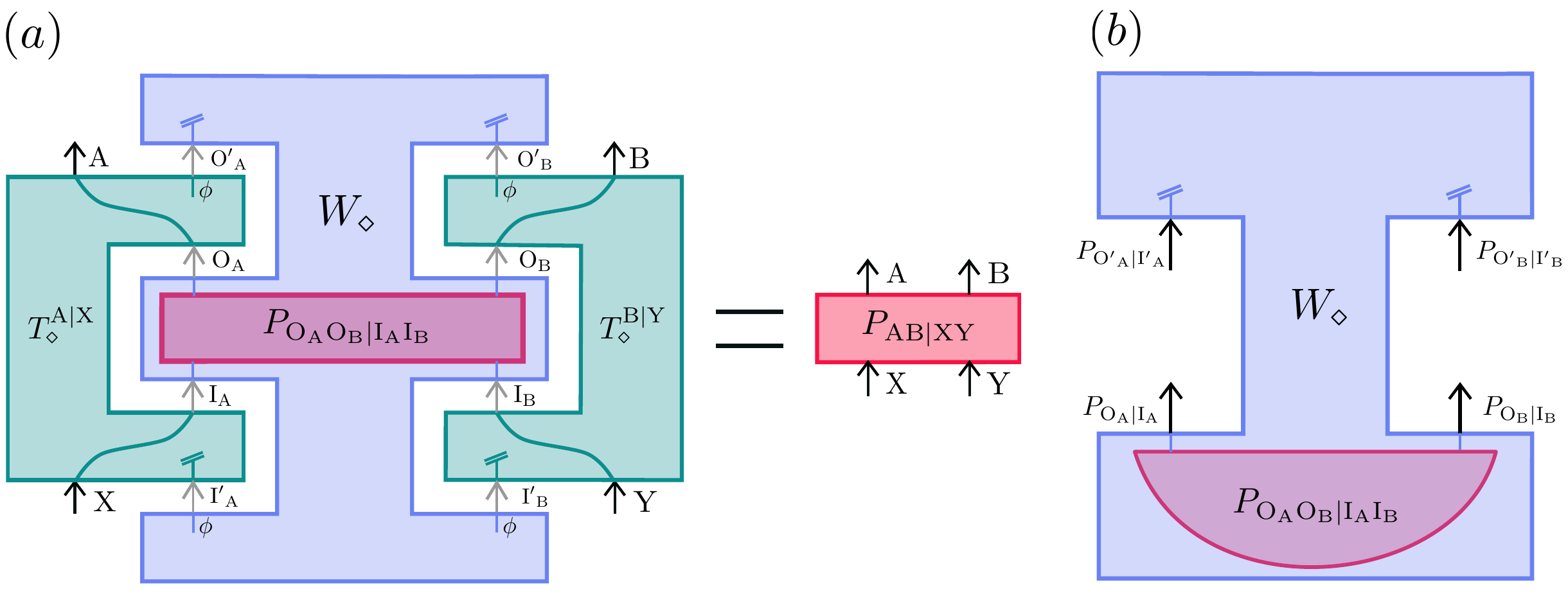}
	\caption{\textbf{Representation of $W_\diamond$.} $(a)$ The process tensor $W_\diamond$ and local operations $T^{\A|\X}_\diamond$ and $T^{\B|\Y}_\diamond$ that can generate any correlations $P_{\A\B|\X\Y}$. $(b)$ The same process tensor $W_\diamond$ depicted with its wires representing box systems, showing how it can be seen as a simple joint state preparation.} 
\label{fig::W_diamond}
\end{center}
\end{figure}
%%%%%%%%%%%%%%%%%%%%%%%%%%%%%%%%%%%%%%%%%%%

This phenomenon violates what we believe should be a tenet of \textit{any} generalized probabilistic theory whose fundamental elements are nonsignaling probability distributions. That is, a shared state preparation between two parties should not be able to generate signaling correlations under the local operations of the theory.

We posit this as a requirement predominantly on physical grounds, thus proposing that process tensors should respect the principle of \textit{nonsignaling preservation}. In the following, we investigate what are the consequences of imposing our boxworld theory to respect this physical principle, both to the set of processes as well as to the set of correlations that are allowed by the theory.

%%%%%%%%%%%%%%
\subsection{Process tensors and nonsignaling preservation}

%%%%%%%%%%%%%%%%%%%%%%%%%%%%%%%%%%%%%%%%%%%
\begin{figure}%[h!]
\begin{center}
	\includegraphics[width=0.8\columnwidth]{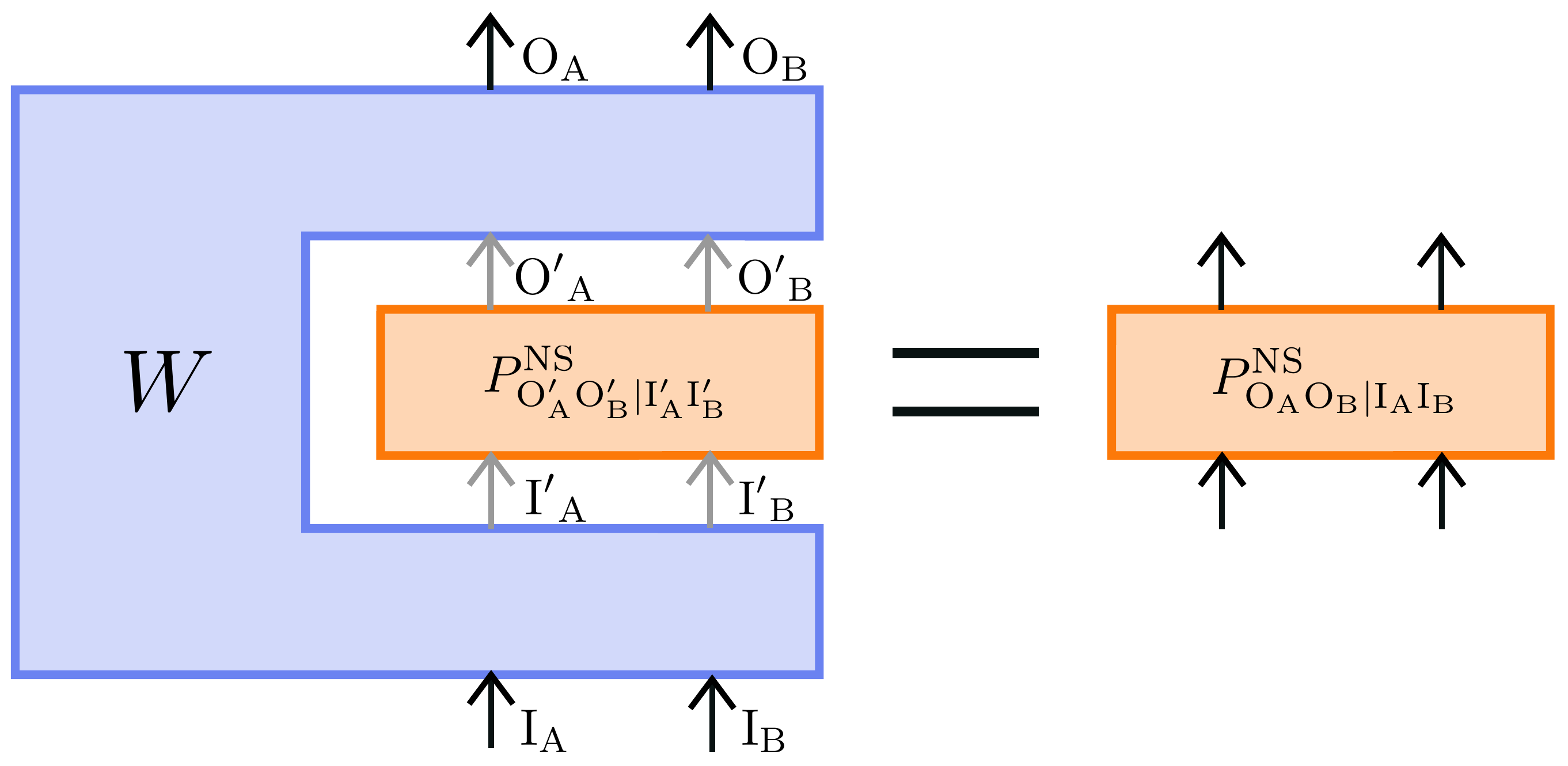}
	\caption{\textbf{Nonsignaling preservation.} A process tensor $W$ that is nonsignaling preserving must satisfy the condition of corresponding to a valid deterministic transformation that maps every nonsignaling box in the local parties' output systems $P^\text{NS}_{\Oap\Obp|\Iap\Ibp}$ to a nonsignaling box in the local parties input systems $P^\text{NS}_{\Oa\Ob|\Ia\Ib}$.}
\label{fig::NSP}
\end{center}
\end{figure}
%%%%%%%%%%%%%%%%%%%%%%%%%%%%%%%%%%%%%%%%%%%

In the process matrix formalism, a 
process matrix can be viewed as a channel, i.e., as a deterministic transformation that maps quantum states of the local parties output systems to quantum states in the local parties input systems~\cite{araujo2015witnessing}. The equivalent condition in boxworld is what we call nonsignaling preservation (NSP): A process tensor that satisfies NSP must act as a deterministic transformation that maps all nonsignaling boxes (the states of the theory) in the local parties' output spaces to nonsignaling boxes in the local parties' input spaces. Formally, a nonsignaling-preserving process tensor $W$ is one that, for all sets of nonsignaling probability distributions $P^\text{NS}_{\Oap\Obp|\Iap\Ibp}$, satisfies
\begin{equation}\label{eq::W_NSP}
    W * P^\text{NS}_{\Oap\Obp|\Iap\Ibp} = P^\text{NS}_{\Oa\Ob|\Ia\Ib} 
\end{equation}
where $P^\text{NS}_{\Oa\Ob|\Ia\Ib}$ is a set of nonsignaling probability distributions. This definition is illustrated in Fig.~\ref{fig::NSP}. As a consequence, a nonsignaling-preserving process tensor is one that cannot prepare a signaling system in its input space, ruling out, for example, the process tensor in Eq.~\eqref{eq::ansatzW1}. A complete characterization of nonsignaling-preserving processes in terms of positivity and linear constraints is presented in Sec.~\ref{sec::methods}, with the proof given in App.~\ref{app::processtensors}  (Theorem~\ref{thm::processtensorsNSP}). There, we also show that this condition is equivalent to the requirement of \textit{complete nonsignaling preservation}, which is the demand that process tensors must correspond to deterministic transformations that map nonsignaling boxes to nonsignaling boxes, even when they act on only a part of the set of probability distributions. 

From the example of the process tensor in Eq.~\eqref{eq::ansatzW1}, which is a valid process tensor that does not respect the principle NSP, one can already see that NSP is a principle that imposes nontrivial constraints in the set of process tensors. However, this is \textit{not} the case for the set of correlations they can create. In fact, just like in the previous case, all correlations can be created by nonsignaling-preserving process tensors and local operations. In particular, the perfect two-way signaling correlations in Eq.~\eqref{eq::2-sigcorr} can be generated by the process tensor and local probabilistic operations given by
\begin{align}
    W_{\triangle} &= \delta_{\Iap,\Ib}\,\delta_{\Oa,\phi}\,\delta_{\Ibp,\phi}\,\delta_{\Ob,\Oap} \label{eq::ansatzW2} \\
    T^{\A|\X}_{\triangle} &= \delta_{\Ia,\phi}\,\delta_{\Oap,\X}\,\delta_{\A,\Iap}, \label{eq::ansatzTax2} \\
    T^{\B|\Y}_{\triangle} &= \delta_{\Ib,\Y}\,\delta_{\Obp,\phi}\,\delta_{\B,\Ob}. \label{eq::ansatzTby2}
\end{align}
Again by comparison with Eq.~\eqref{eq::op_probabilistic}, one can check that $T^{\A|\X}_{\triangle}$ in Eq.~\eqref{eq::ansatzTax2} and $T^{\B|\Y}_{\triangle}$ in Eq.~\eqref{eq::ansatzTby2} are valid probabilistic operations. In App.~\ref{app::NSP}, we show that the $W_{\triangle}$ in Eq.~\eqref{eq::ansatzW2} is a valid process tensor that is nonsignaling preserving. This process tensor and local operations are represented in Fig.~\ref{fig::W_triangle}. It can be easily checked that the correlations produced by this set of operations and process tensor amount to $P_{\A\B|\X\Y}=T^{\A|\X}_{\triangle}*T^{\B|\Y}_{\triangle}*W_{\triangle}=\delta_{\A,\Y}\,\delta_{\B,\X}=P^\text{2-sig}_{\A\B|\X\Y}$.

%%%%%%%%%%%%%%%%%%%%%%%%%%%%%%%%%%%%%%%%%%%
\begin{figure}%[h!]
\begin{center}
	\includegraphics[width=\columnwidth]{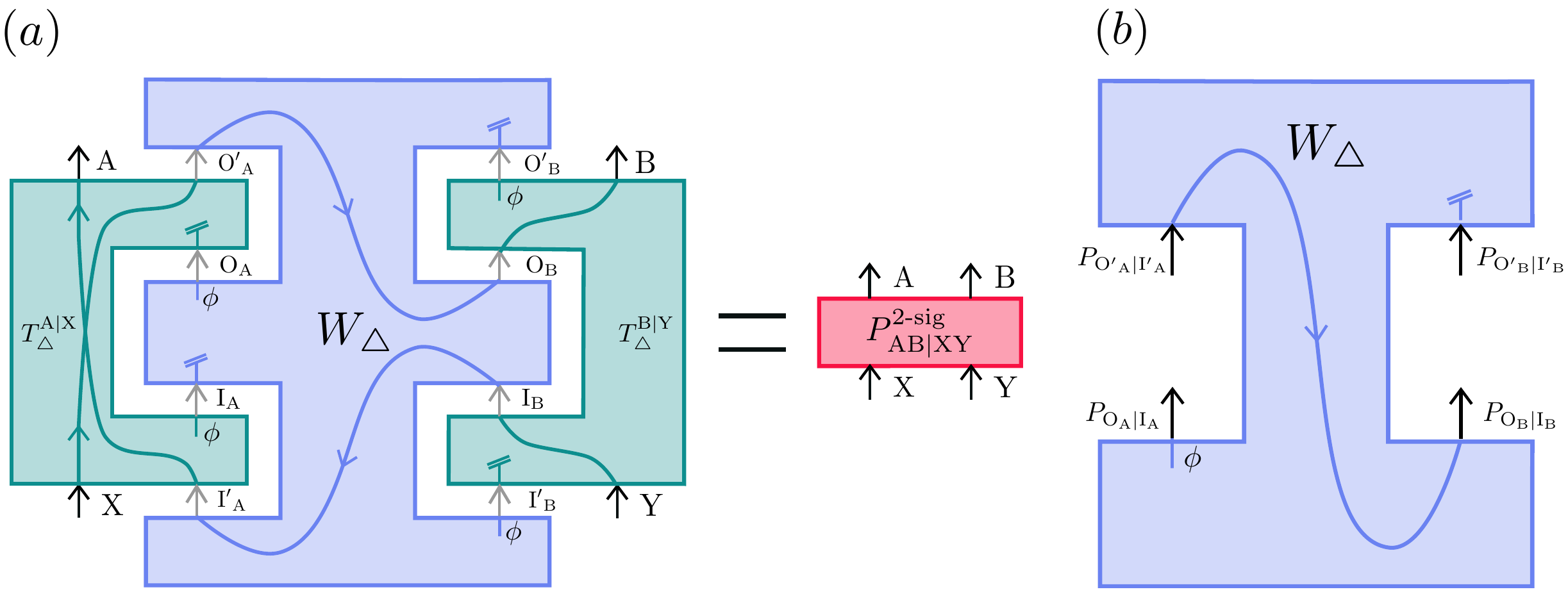}
	\caption{\textbf{Representation of $W_\triangle$.} $(a)$ The process tensor $W_\triangle$ and local operations $T^{\A|\X}_\triangle$ and $T^{\B|\Y}_\triangle$ that can generate perfect two-way signaling correlations $P^\text{2-sig}_{\A\B|\X\Y}$. $(b)$ The same process tensor $W_\triangle$ depicted with its wires representing box systems, showing how it can be seen as a simple identity channel that maps a box system from Alice's output space to Bob's input space.}
\label{fig::W_triangle}
\end{center}
\end{figure}
%%%%%%%%%%%%%%%%%%%%%%%%%%%%%%%%%%%%%%%%%%%

The process in Eq.~\eqref{eq::ansatzW2} corresponds to a one-way identity channel that maps sets of probability distributions $P_{\Oap|\Iap}$ in Alice's output space to sets of probability distributions $P_{\Ob|\Ib}$ in Bob's input space. Yet, when in possession of this one-way perfect communication channel, we see that Alice and Bob are able to create perfect two-way signaling correlations. Hence, via local classical post-processing, they are also able to generate any sets of probability distributions as correlations. Therefore, we prove that process tensors that respect NSP are able to generate any set of probability distributions under local probabilistic operations.

The constraint that process tensors should be nonsignaling preserving results in a class of processes that cannot generate signaling correlations from simple state preparations, as was the case of $W_\diamond$ in Eq.~\eqref{eq::ansatzW1}. However, we find that this requirement is not strong enough to reduce the set of correlations whatsoever: for this new class, the full set of two-way signaling correlations is still achievable under local transformations. What's more is that the process required to produce any distribution need itself be only one-way signaling, as is the case of $W_\triangle$ in Eq.~\eqref{eq::ansatzW2}. We argue that this fact---that all two-way signaling correlations can be generated by one-way signaling processes---posts a radical departure from common sense, and assert that signaling correlations (in a given direction) should only originate from processes that facilitate {communication} in that direction. In short: one should not be able to acquire external information unless someone else physically sends it to them. This prompts us to propose a second physical principle for our higher-order boxworld theory: \textit{no signaling without system exchange}.

%%%%%%%%%%%%%%
\subsection{Process tensors and no signaling without system exchange}

Having learned that the NSP principle does not restrict the set of correlations and also allows for some processes with undesirable properties, we set out to formulate a different principle to impose on our boxworld theory. To do so, we take careful note of the concrete examples in Eqs.~\eqref{eq::ansatzW1}--\eqref{eq::ansatzTby1} and Eqs.~\eqref{eq::ansatzW2}--\eqref{eq::ansatzTby2}. There, we see that regardless of whether the process is constrained by the NSP principle or not, both examples create perfect two-way signaling distributions under identical operations on Bob's side. This operation is a simple relabeling of random variables, which in both cases allows for Bob to signal to Alice. We reason that these types of operations are trivial, since they do not encode any information into the transformed systems. Therefore, it is reasonable to impose that a process tensor under such operations should not create signaling correlations.

Thus, our proposed restriction: the principle of no signaling without system exchange (NSWSE) solicits a modification in the definition of how a process should behave under a subclass of operations that we call {\it nonsignaling operations}. Our explicit requirement then is that process tensors which respect NSWSE cannot produce signaling correlations when acted upon by nonsignaling operations.

Our starting point is to identify the subclass of operations in Eqs.~\eqref{eq::op_deterministic} and~\eqref{eq::op_probabilistic} which we term nonsignaling. A set of deterministic operations $\overline{T}^{\X}$ is termed nonsignaling if, when acting locally on a box $P_{\O|\I}$, the resulting box $P_{\Op|\Ip}=\overline{T}^{\X}*P_{\O|\I}$ is independent of $\X$. A set of probabilistic operations $\overline{T}^{\A|\X}$ is nonsignaling if $\sum_\A\overline{T}^{\A|\X}=\overline{T}^{\X}$ is nonsignaling as well.

%%%%%%%%%%%%%%%%%%%%%%%%%%%%%%%%%%%%%%%%%%%
\begin{figure}%[h!]
\begin{center}
	\includegraphics[width=0.8\columnwidth]{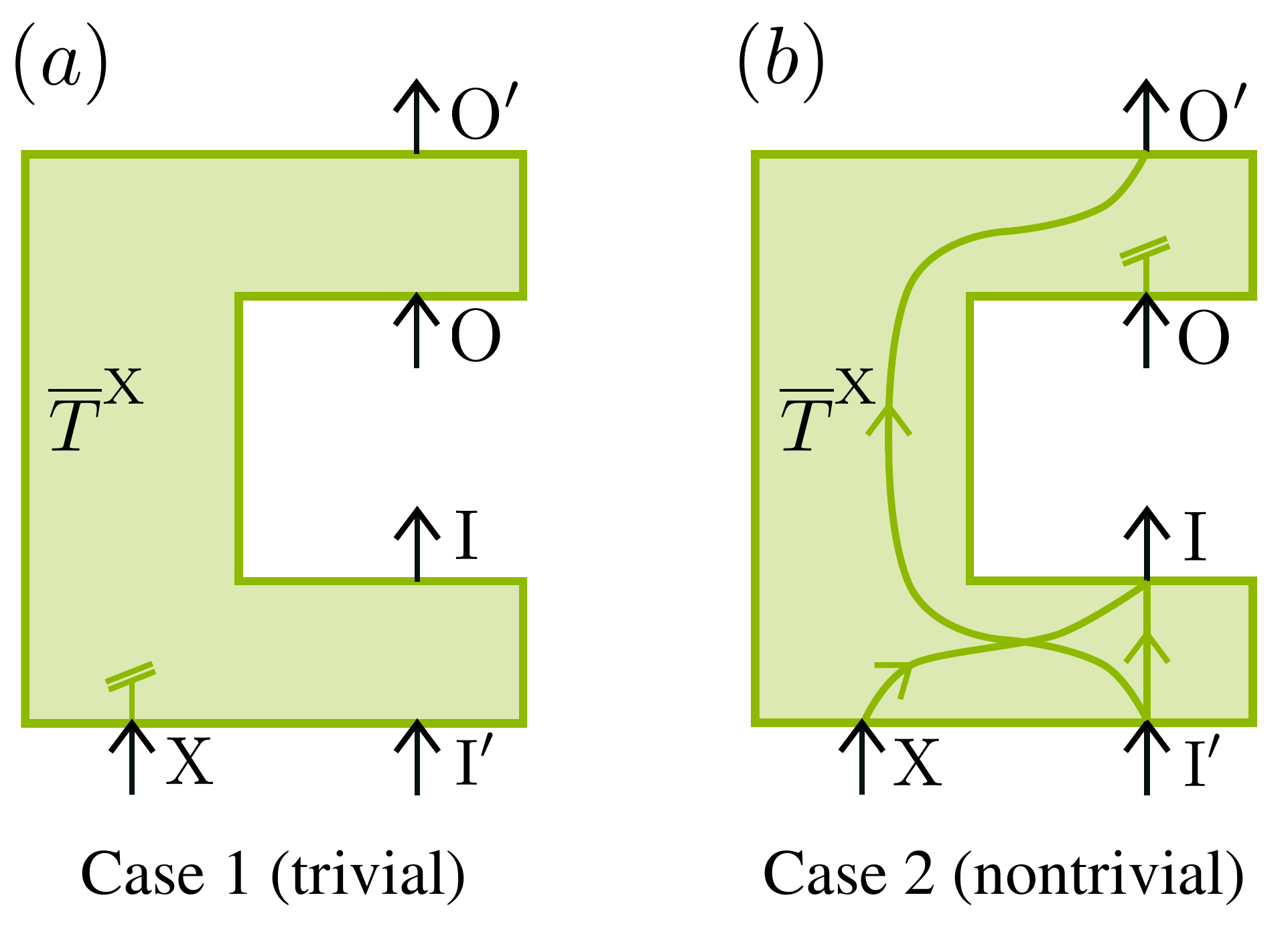}
	\caption{\textbf{Nonsignaling deterministic operations.} Sets of nonsignaling deterministic operations $\overline{T}^{\X}$ are either of $(a)$ case 1 (trivial), where $\overline{T}^{\X}$ does not depend on $\X$, or $(b)$ case 2 (nontrivial), where in $\overline{T}^{\X}$, only the random variable $\I$ may depend on $\X$ but $\Op$ is independent of $\O$.}
\label{fig::operations_NS}
\end{center}
\end{figure}
%%%%%%%%%%%%%%%%%%%%%%%%%%%%%%%%%%%%%%%%%%%

At first glance, we might assume that a set of operations is nonsignaling if all operations in the set are identical, i.e., if the set  $\overline{T}^{\X}$ itself is independent of $\X$, implying they can be decomposed as $\overline{T}^{\X}= P_{\I|\Ip}\,P_{\Op|\Ip\O} = \sum_\lambda \pi_\lambda\, D^\lambda_{\I|\Ip}\,D^\lambda_{\Op|\Ip\O}$ for all values of $\X$. While this is true, there is yet another class of operations that also satisfy the nonsignaling constraint, which might depend on $\X$ and yet not encode this information on its transformed systems. As we prove in App.~\ref{app::NSoperations}  (Theorem~\ref{thm::nonsignalingoperations}), these operations are necessarily of the form
\begin{equation}\label{eq::op_nonsignaling}
    \overline{T}^{\X} = P_{\I|\Ip \X}\,P_{\Op|\Ip} = \sum_\lambda \pi_\lambda\, D^\lambda_{\I|\Ip X}\,D^\lambda_{\Op|\Ip},
\end{equation}
where the pre-processing functions $P_{\I|\Ip \X}$ may depend on $\X$, however, the post-processing functions $P_{\Op|\Ip}$ must be independent of both $\O$ and $\X$. Together, these two classes of operations correspond to the entire set of nonsignaling operations. A characterization of these operations is provided in Sec.~\ref{sec::methods}. Note that nonsignaling is a nontrivial constraint for local operations in boxworld, which defines a strict subset of operations, depicted in Fig.~\ref{fig::operations_NS}.

We are now equipped to define process tensors that satisfy the principle of NSWSE. A process tensor $W$ satisfies NSWSE if, when nonsignaling operations are applied to it, it can only lead to nonsignaling correlations. That is, if for all nonsignaling operations $\overline{T}^{\A|\X}$ and $\overline{T}^{\B|\Y}$,
\begin{equation}\label{eq::W_NSWSE}
    W * \overline{T}^{\A|\X} * \overline{T}^{\B|\Y} = P^\text{NS}_{\A\B|\X\Y}
\end{equation}
where $P^\text{NS}_{\A\B|\X\Y}$  satisfies Eqs.~\eqref{eq::probNS_AprecB} and~\eqref{eq::probNS_BprecA}. This definition is illustrated in Fig.~\ref{fig::NSWSE}. In Sec.~\ref{sec::methods}, we provide a characterization of process tensors that satisfy NSWSE in terms of positivity and linear constraints. This characterization is proven in App.~\ref{app::NSWSE}  (Theorem~\ref{thm::processtensorsNSWSE}). We furthermore show that NSWSE is a stronger condition than NSP, and that all process tensors that satisfy NSWSE also satisfy NSP.

In fact, NSWSE is a \textit{strictly} stronger condition than NSP. Take for example the set of operations in Eqs.~\eqref{eq::ansatzTax2}--\eqref{eq::ansatzTby2}, as well as those previously presented in Eqs.~\eqref{eq::ansatzTax1}--\eqref{eq::ansatzTby1}. These sets of operations are nonsignaling, i.e., they satisfy Eq.~\eqref{eq::op_nonsignaling}. Yet, there exist process tensors, namely the ones in Eqs.~\eqref{eq::ansatzW2} and~\eqref{eq::ansatzW1}, such that, when these operations are applied onto these processes, the generated correlations are signaling. Hence, neither the process tensor in Eq.~\eqref{eq::ansatzW1} nor the one in Eq.~\eqref{eq::ansatzW2} satisfy NSWSE, and would be forbidden by this principle. In that respect, this principle has stark consequences: it rules out the possibility of one-way gbit channels, as the one of Eq.~\eqref{eq::ansatzW2}.

We duly arrive at the conclusion that the NSWSE principle imposes nontrivial constraints on the set of allowed process tensors. Additionally, contrarily to the NSP principle, we find that NSWSE also imposes nontrivial constraints on the set of allowed correlations. In particular, in App.~\ref{app::no2waysig}  (Theorem~\ref{thm::no2waysig}), we show that correlations of the form $P^\text{2-sig}_{\A\B|\X\Y}$ are excluded: under the most general operations, a process tensor that satisfies NSWSE cannot generate perfect two-way signaling. Since these correlations are also not in the set of process matrix correlations~\cite{bavaresco2019semi,kunjwal2023nonclassicality}, NSWSE presents itself as a candidate for a principle that could lead to an outer approximation of the set of signaling correlations generated in higher-order quantum theory.

%%%%%%%%%%%%%%%%%%%%%%%%%%%%%%%%%%%%%%%%%%%
\begin{figure}%[h!]
\begin{center}
	\includegraphics[width=\columnwidth]{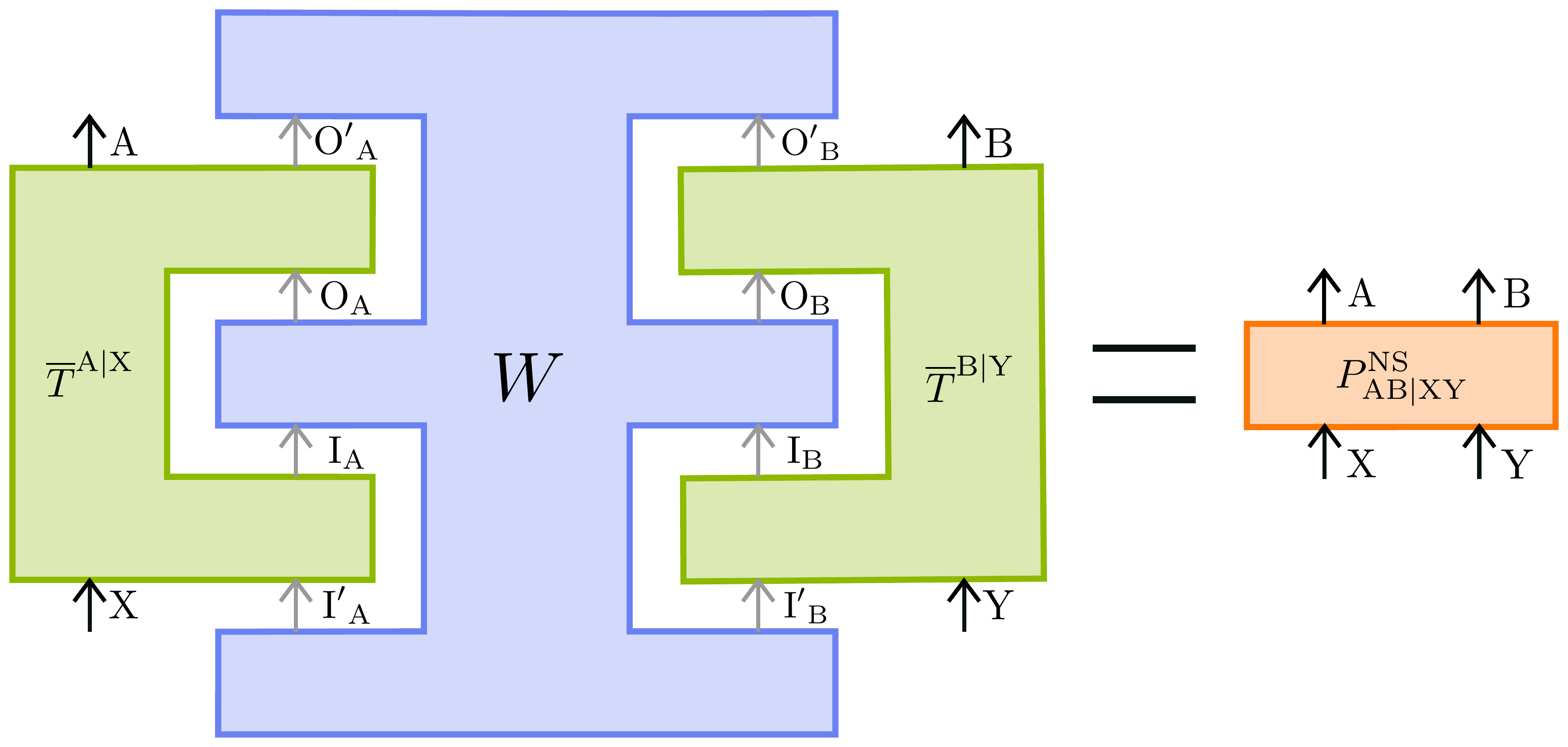}
	\caption{\textbf{No signaling without system exchange.} A process tensor $W$ respects the principle of no signaling without system exchange if, for all nonsignaling local operations $\overline{T}^{\A|\X}$ and $\overline{T}^{\B|\Y}$, the resulting correlations $P^\text{NS}_{\A\B|\X\Y}$ are nonsignaling.}
\label{fig::NSWSE}
\end{center}
\end{figure}
%%%%%%%%%%%%%%%%%%%%%%%%%%%%%%%%%%%%%%%%%%%

If this is the case, then process matrix correlations must also satisfy NSWSE; intuitively one would expect them to, and we now argue that this is indeed the case. To make the connection, one must first identify the class of nonsignaling quantum instruments. In quantum theory, a set of quantum instruments can only be nonsignaling in a manner analogous to case 1 in boxworld: when the set of transformations trivially does not depend on the classical input $\X$. These are the transformations whose output system, once summed over all values of the classical output, does not depend on the value of the classical input. In technical vocabulary these are completely positive, trace non-increasing linear maps $\{\Ical^{a|x}\}$, which, for any input state $\rho$, satisfy $\sum_a \Ical^{a|x}(\rho) = \rho'$, where $\rho'$ is quantum state independent of $x$. This type of set of quantum instruments has been called an instrument assemblage~\cite{piani2015channel}. Under the action of these nonsignaling sets of quantum instruments, it is straightforward to see that processes matrices can only produce nonsignaling correlations. Hence, in higher-order quantum theory, process matrices automatically satisfy NSWSE.

However, since this principle is imposed on the set of processes, both in quantum theory and in boxworld, it is not straightforward to conclude its implications to the set of correlations that can be attained by process tensors versus process matrices.

We have now reached a vantage point where we have two sets of correlations at hand: correlations generated by process matrices that respect NSWSE `for free'; and correlations generated by process tensors that we require to respect the NSWSE principle. However, since this principle is formulated at the level of the processes, both in quantum theory and in boxworld, its implications to the set of correlations that these processes can create is not straightforward. Namely, while the set of general process tensor correlations contains the set of process matrix correlations, albeit in a trivial manner, is it still the case that the set of correlations from process tensors that respect NSWSE contains the set of process matrix correlations?
In the next section we construct a fully fledged theory of processes tensors that respect NSWSE in order to further investigate this question.

%%%%%%%%%%%%%%
\subsection{A theory of boxworld processes}

Motivated by the physical principles of NSP and NSWSE, and by the goal of constructing a nontrivial outer approximation to the set of process matrix correlations, we develop a theory of boxworld processes. We begin by defining \textit{boxworld processes} to be process tensors that respect the NSWSE principle, and consequently also the NSP principle. We furthermore define \textit{boxworld correlations} to be the correlations that can be produced by boxworld processes under the action of general local probabilistic operations. Boxworld processes will still be denoted as ${W}$, while boxworld correlations will be denoted as $\overline{P}_{\A\B|\X\Y}$. 

As previously mentioned, the first result we prove about boxworld correlations is that they do not contain sets of perfect two-way signaling probability distributions. In particular, in the two input, two output scenario (i.e., $|\A|=|\B|=|\X|=|\Y|=2$, where $|\O|$ denotes the cardinality of the random variable $\O$) we show that that if $\overline{P}_{\A\B|\X\Y}$ are boxworld correlations, then they satisfy
\begin{equation}\label{eq::bwcorr_bound}
    \frac{1}{4}\sum_{\A,\B,\X,\Y} \delta_{\A,\Y}\,\delta_{\B,\X} \overline{P}_{\A\B|\X\Y} \leq 1 - \frac{1}{2d},
\end{equation}
where $d\coloneqq\min\{|\Oap|,|\Obp|\}>0$. This statement is proven in App.~\ref{app::no2waysig}  (Theorem~\ref{thm::no2waysig}). This inequality is violated by $P^\text{2-sig}_{\A\B|\X\Y}$ in Eq.~\eqref{eq::2-sigcorr}, since they would attain a value of $1$ for the left-hand side of Eq.~\eqref{eq::bwcorr_bound}. Hence, perfect two-way signaling cannot be achieved by any boxworld process over random variables of finite, albeit arbitrarily high, cardinality.

The second result is that the set of all boxworld correlations forms a polytope. This follows from the fact that both boxworld processes and local operations are themselves described by sets of probability distributions that form polytopes, as shown in App.~\ref{app::polytope}  (Theorem~\ref{thm::polytope}). Thus, the polytope of boxworld correlations corresponds to a strict subset of the polytope of all probability distributions. Since the set of boxworld correlations forms a polytope, it is clearly different from the set of process matrix correlations, which is also a convex set however not a polytope~\cite{branciard2016simplest}. It remains then to find out whether the set of process matrix correlations is strictly contained in the set of boxworld correlations.

To further understand the properties of boxworld correlations, we investigate the causal properties of boxworld processes. Much like process matrices, boxworld processes can display different causal properties. Boxworld processes that, under the action of any sets of local probabilistic transformations for Alice and Bob, can only generate correlations $P^{A\prec B}_{\A\B|\X\Y}$ that are nonsignaling from Bob to Alice [i.e., correlations that satisfy Eq.~\eqref{eq::probNS_AprecB}] are termed causally ordered from Alice to Bob and denoted as $W^{A\prec B}$. Analogously, we define boxworld processes $W^{B\prec A}$ to be causally ordered from Bob to Alice. Boxworld processes that can only generate nonsignaling correlations $P^\text{NS}_{\A\B|\X\Y}$ [i.e., correlations that satisfy both Eq.~\eqref{eq::probNS_AprecB} and~\eqref{eq::probNS_BprecA}] regardless of the sets of probabilistic operations applied to it, are called nonsignaling boxworld processes and denoted as $W^{A||B}$. Finally, causally separable boxworld processes are those that can be decomposed as a convex combination of causally ordered boxworld processes, i.e., $W^\text{sep}=qW^{A\prec B}+(1-q)W^{B\prec A}$. Boxworld processes that are not causally separable are considered to display indefinite causal order. We provide characterizations of all these different sets of boxworld processes in App.~\ref{app::boxworldprocesses} (Theorem~\ref{thm::causordproc}). This same logic also holds for the definitions of casual order for process matrices in higher-order quantum theory~\cite{chiribella2009theoretical}. 

In terms of the correlations, those that display two-way signaling but that can nonetheless be expressed as convex combinations of one-way nonsignaling correlations, i.e.,
\begin{equation}\label{eq::probcausal}
    P^\text{causal}_{\A\B|\X\Y} = q\,P^{A\prec B}_{\A\B|\X\Y} + (1-q)\,P^{B\prec A}_{\A\B|\X\Y},
\end{equation}
where $0\leq q\leq 1$, have been term \textit{causal correlations}. In turn, noncausal correlations, which are those that cannot be decomposed in such a way, are correlations that display genuine two-way signaling, that is, correlations with a form of two-way signaling that cannot be explained as classical mixing of signaling in different one-way directions. Correlations that cannot be decomposed in such a way are called \textit{noncausal correlations}, and they can violate causal inequalities~\cite{oreshkov2012quantum,branciard2016simplest}---a certification of indefinite causal order.

We show that all causal correlations can be attained by boxworld processes. First, we show that all nonsignaling correlations can be attained by local transformations acting on \textit{nonsignaling} boxworld processes. Second, we show that all one-way nonsignaling correlations can be obtained by local transformations acting on \textit{causally ordered} boxworld processes. Finally, we show that all causal correlations can be attained by local transformations acting on \textit{causally separable} boxworld processes. The proofs can be found in App.~\ref{app::causalcorrelations}  (Theorem~\ref{thm::causalcorrelations}). This is also true for the equivalent statements in the formalism of process matrices~\cite{oreshkov2012quantum,bavaresco2019semi}. Therefore, in terms of causal correlations, there is no difference between process matrices and boxworld processes. 

Our final effort is dedicated to comparing the genuine two-way signaling correlations generated by boxworld processes with those generated by process matrices. That is, to further investigating the potential differences of indefinite causal order at the level of correlations between process matrices and boxworld processes. In order to investigate whether there is a difference between the two sets, we test boxworld correlations against several causal inequalities and compare their performance against known results from the literature on process matrix correlations. The three inequalities we study are: the ``guess your neighbor's input'' (GYNI) inequality~\cite{almeida2010guess}; the ``lazy guess your neighbors input'' (LGYNI) inequality~\cite{branciard2016simplest}; and the ``Oreshkov-Costa-Brukner'' (OCB) inequality~\cite{oreshkov2012quantum}. GYNI and LGYNI are inequalities in the scenario where all inputs and outputs are bits, i.e., where $|\A|=|\B|=|\X|=|\Y|=2$; they form the nontrivial facets of the causal polytope in this scenario. The OCB inequality involves a scenario where Bob has an additional input, $|\Yp|=2$. This inequality is not a facet of the causal polytope. 

The GYNI inequality~\cite{almeida2010guess} is defined as 
\begin{equation}\label{eq::gyni}
    \text{GYNI}(P)\coloneqq\frac{1}{4}\sum_{\A,\B,\X,\Y} \delta_{\A,\Y}\,\delta_{\B,\X} P_{\A\B|\X\Y} \leq \frac{1}{2},
\end{equation}
while the LGYNI inequality~\cite{branciard2016simplest} is defined as 
\begin{equation}\label{eq::lgyni}
    \text{LGYNI}(P)\coloneqq\frac{1}{4}\sum_{\A,\B,\X,\Y} \delta_{\X(\A\oplus\Y),0}\,\delta_{\Y(\B\oplus\X),0} P_{\A\B|\X\Y} \leq \frac{3}{4},
\end{equation}
and the OCB inequality~\cite{oreshkov2012quantum} is defined as
\begin{align}\label{eq::ocb}
%\begin{split}
    \text{OCB}(P) &\coloneqq \frac{1}{8}\sum_{\substack{\A,\B, \\ \X,\Y,\Yp}} \delta_{(\Yp\oplus 1)(\A\oplus\Y),0}\,\delta_{\Yp(\B\oplus\X),0}\, P_{\A\B|\X\Y\Yp} \nonumber \\
    &\leq \frac{3}{4}.
%\end{split}
\end{align}
In all cases, the bounds are respected by all causal correlations, defined in Eq.~\eqref{eq::probcausal}. These inequalities are known to be violated by some process matrix correlations~\cite{oreshkov2012quantum,branciard2016simplest}. 

We show that these inequalities are also violated by boxworld correlations. As detailed in App.~\ref{app::inequalities}, we explicitly construct boxworld processes and local operations that generate boxworld correlations $\overline{P}_{\A\B|\X\Y}$ that yield the following violations for the GYNI, LGYNI, and OCB inequalities, implying the following lower bounds:
\begin{align}
    \max_{\mathcal{P}_\text{BW}}\, \text{GYNI}(\overline{P}_{\A\B|\X\Y}) &\geq \frac{3}{4} = 0.75 \\
    \max_{\mathcal{P}_\text{BW}}\, \text{LGYNI}(\overline{P}_{\A\B|\X\Y}) &\geq \frac{11}{12} \approx 0.9167 \\
    \max_{\mathcal{P}_\text{BW}}\, \text{OCB}(\overline{P}_{\A\B|\X\Y\Yp}) &= 1,
\end{align}
where $\mathcal{P}_\text{BW}$ is the polytope of boxworld correlations. The fact that all these causal inequalities are violated by boxworld correlations certifies that indefinite causal order is a property that is present in our theory of boxworld processes. The above value for the GYNI expression was achieved by boxworld correlations generated by boxworld processes that act on random variables of cardinality $3$, i.e., where $|\Iap|=|\Ia|=|\Oa|=|\Oap|=|\Ibp|=|\Ib|=|\Ob|=|\Obp|=3$. For the LGYNI and OCB expressions, the values above where achieved by boxworld correlations from boxworld processes that act on random variables of cardinality 2, i.e., where $|\Iap|=|\Ia|=|\Oa|=|\Oap|=|\Ibp|=|\Ib|=|\Ob|=|\Obp|=2$. The boxworld correlations, local probabilistic operations, and boxworld processes that achieve these violations are presented explicitly in App.~\ref{app::inequalities}.

Notice that, for the OCB expression, boxworld correlations can attain the algebraic maximum. However, this does not imply perfect two-way signaling correlations, which cannot be generated by boxworld processes. Instead, there are nondeterministic boxworld correlations that attain the algebraic maximum of the OCB expression, showing that the polytope of boxworld correlations has nondeterministic extremal points. This is similar to the polytope of nonsignaling correlations, which also has nondeterministic extremal points known as PR-boxes~\cite{vandam2005implausible,rastall1985locality,popescu1994quantum}, and which forms a nontrivial outer approximation to the set of nonsignaling quantum correlations in Bell scenarios.

Ref.~\cite{liu2024tsirelson} has recently proven upper bounds on the maximal violations of these inequalities by process matrix correlations, in arbitrary (finite or infinite) dimensions. This is the analogue of upper bounds of the Tsirelson bound~\cite{tsirelson1980quantum}, or quantum bound, of Bell inequalities, which is the highest violation of such inequalities by correlations produced by quantum systems and local quantum measurements in Bell scenarios. Ref.~\cite{liu2024tsirelson} showed that 
\begin{align}
    \max_{\mathcal{S}_\text{QW}}\, \text{GYNI}(P_{\A\B|\X\Y}) &\leq 0.7592 \\ 
    \max_{\mathcal{S}_\text{QW}}\, \text{LGYNI}(P_{\A\B|\X\Y}) &\approx 0.8194 \\ 
    \max_{\mathcal{S}_\text{QW}}\, \text{OCB}(P_{\A\B|\X\Y\Yp}) &= \frac{2+\sqrt{2}}{4} \approx 0.8536,
\end{align}
where $\mathcal{S}_\text{QW}$ is the set of process matrix correlations. It is unknown whether the upper bound for GYNI is tight~\cite{liu2024tsirelson}.

Therefore, we conclude that boxworld correlations can attain a higher violation of both the LGYNI and the OCB inequality than process matrix correlations. For GYNI, it is not clear if boxworld correlations can attain a higher value than process matrix correlations. However, since the bound for process matrix correlations from Ref.~\cite{liu2024tsirelson} may not be tight, and the value we found for boxworld correlations, which was attained could increase for higher-dimensional systems, it is possible that boxworld correlations could also attain a higher violation of GYNI than process matrix correlations.

Based on these results, we conjecture that the polytope of boxworld correlations is a nontrivial outer approximation of the set of process matrix correlations.

%%%%%%%%%%%%%%%%%%%%%%%%%%%%%%%%%%%%%%%%%%%%%%%%%%%%
\begin{table}%[h!]
\begin{center}
{\renewcommand{\arraystretch}{2}
\begin{tabular}{| c | c | c | c |}
	\hline
	Expression & Causal & Process matrix & Boxworld \\
    \hline
    \hline
    GYNI & $\leq\frac{1}{2}=0.5$  & $\leq 0.7592$ & $\geq \frac{3}{4}=0.75$ \\
    %\hline
    LGYNI & $\leq\frac{3}{4}=0.75$  & $\approx 0.8194$ &  $\geq\frac{11}{12}\approx0.9167$\\
    %\hline        
    OCB & $\leq\frac{3}{4}=0.75$ & $=\frac{2+\sqrt{2}}{4}\approx0.8536$ & $=1$ \\
    \hline
\end{tabular}
}
\end{center}
\caption{Bounds on the maximum values for the GYNI [Eq.~\eqref{eq::gyni}], LGYNI [Eq.~\eqref{eq::lgyni}], and OCB [Eq.~\eqref{eq::ocb}] expressions that can be attained by causal correlations, process matrices correlations, and boxworld correlations.}
\label{tb::bounds}
\end{table}
%%%%%%%%%%%%%%%%%%%%%%%%%%%%%%%%%%%%%%%%%%%%%%%%%%%%

%%%%%%%%%%%%%%%%%%%%%%%%%%%%%%%%%%%%%%%%%%%%
\section{Discussion}\label{sec::discussion}
%%%%%%%%%%%%%%%%%%%%%%%%%%%%%%%%%%%%%%%%%%%%

Physical and information-theoretic principles have long been proposed as ways to understand the power and limitations of correlations in quantum theory; in particular, for the nonsignaling quantum correlations, produced by quantum systems in Bell scenarios. At the same time, the development of higher-order quantum theory has provided a framework to study other more general, although less well understood, correlations that can be achieved in signaling scenarios. Beyond simply capturing signaling that can be achieved by exchange of quantum systems through quantum communication in a fixed order, the quantum higher-order formalism predicted forms of genuine two-way signaling correlations associated with a lack of a definite global order of events between the communicating parties. Such a property, called indefinite causal order, has been extensively explored, and has raised foundational questions about the nature of causality and found several applications in quantum information processing.
Despite this, the largest set of genuine two-way signaling correlations allowed in higher-order quantum theory---namely, the set of process matrix correlations---has remained
poorly understood. 

In this work we investigated the boundary of the set of process matrix correlations by first establishing a probabilistic higher-order theory more general than quantum theory---one which is based on boxworld---and second, by imposing physical principles on this theory. 

The starting point was choosing our basic states to be nonsignaling sets of probability distributions, a generalization of quantum states. We then adopted the characterization of the most general local transformations that map nonsignaling probability distributions onto themselves, from Ref.~\cite{baumeler2016space}, as the local operations of our theory. We subsequently characterized the most general set of higher-order transformations that map local operations onto correlations, which are the processes of our theory, called process tensors. After concluding that the correlations generated by general process tensors do not satisfy any nontrivial constraints, i.e., that they can correspond to any set of probability distributions, we motivated and imposed two physical principles to limit the correlations that can be achieved in our boxworld theory. The first principle, nonsignaling preservation, restricted the set of allowed processes in the theory but not the set of allowed correlations. The second principle, no signaling without system exchange (which is strictly stronger than the first), further restricted the set of processes in the theory as well as the set of correlations. Combining our higher-order boxworld theory with the principle of NSWSE, we constructed a theory of so-called boxworld processes and correlations. We then exploited this theory to investigate properties of boxworld correlations and how they relate to process matrix correlations. We found that boxworld correlations can violate some causal inequalities, in particular the LGYNI and OCB inequalities, beyond what is achievable by process matrix correlations. Combined with the fact that process matrices satisfy what we argue to be the analog of NSWSE in quantum theory, we find evidence that the polytope of boxworld correlations is a nontrivial outer approximation for the set of process matrix correlations. We therefore conjecture this to be the case. 

Establishing whether this conjecture holds would help advancing our understanding of the relationship between boxworld correlations and process matrix correlations. Confirming this would provide a clearer, physically motivated limitation for the set of process matrix correlations while also improving our understanding of higher-order quantum theory. Beyond this, it would be interesting to explore additional physical principles that could further constrain or characterize the correlations within this framework. Such investigations could offer deeper insights into the nature of indefinite causal order and its implications for quantum theory.

%%%%%%%%%%%%%%%%%%%%%%%%%%%%%%%%%%%%%%%%%%%%
\section{Methods}\label{sec::methods}
%%%%%%%%%%%%%%%%%%%%%%%%%%%%%%%%%%%%%%%%%%%%

\textit{\textbf{Tensor operations and notation.}} The basic elements of our theory are sets of probability distributions that are represented by tensors in finite-dimensional real vector spaces. Let us establish the convention of denoting a probability distribution $P$ from inputs $\I=\{i\}_i$ to outputs $\O=\{o\}_o$ as a tensor $P_{\O|\I}$, of which each element is denoted by $P_{\O|\I}(o|i)$. Each index of a tensor corresponds to a certain random variable.

To express the constraints on tensors that correspond to sets of probability distributions, we denote elementwise nonnegativity $P_{\A|\X}(a|x)\geq0 \ \forall \ a\in\A,x\in\X$ with the shorthand notation
\begin{equation}
P_{\A|\X} \geq 0,
\end{equation}
and, for normalization, one denote the expression $\sum_{a}P_{\A|\X}(a|x)=1 \ \forall \ x\in\X$ as simply
\begin{equation}
\sum_{\A} P_{\A|\X} = 1. 
\end{equation}

We define tensor composition by a contraction over repeated indices. Tensor composition is denoted by the symbol~$*$, and is analogous to the link product for matrices~\cite{chiribella2008quantum}. For instance, take the tensors $R_{\A\B}$ and $S_{\B\C}$. Their contraction $Q_{\A\C} = R_{\A\B}*S_{\B\C}$ is given by $Q_{\A\C}(a,c) = \sum_{b}R_{\A\B}(a,b)S_{\B\C}(b,c) \ \forall \ a\in\A,c\in\C$, or, in a more compact form,
\begin{equation}
    Q_{\A\C} = R_{\A\B} * S_{\B\C} \coloneqq \sum_{\B}R_{\A\B}S_{\B\C}.
\end{equation}
When there is no risk of ambiguity, a contraction can be simply denoted as $Q=R*S$, and it is implied to be given over all repeated indices and only over repeated indices. When all the indices are contracted, the operation~$*$ defines a scalar product on a real vector space. Similarly, we can compose tensors with no overlapping indices with the tensor analog of a matrix tensor product, denoted by~$\otimes$. For instance, given two tensors $R_{\A\B}$ and $S_{\C\D}$ we define their composition $Q_{\A\B\C\D}\coloneqq R_{\A\B}\otimes S_{\C\D}$ as $Q_{\A\B\C\D}(a,b,c,d) = R_{\A\B}(a,b) S_{\C\D}(c,d),\ \forall \ a\in\A,b\in\B,c\in\C,d\in\D$, which can be simply written as
\begin{equation}\label{eq::notation_tensorprod}
    Q_{\A\B\C\D} = R_{\A\B} \otimes S_{\C\D} \coloneqq R_{\A\B}S_{\C\D}.
\end{equation}
When there is no risk of ambiguity, the above expression can also be written as simply $Q=R\otimes S$.

We also define the identity tensor $\id_\A$ over a random variable $\A=\{a\}_a$ as a tensor of all ones, i.e., $\id_\A(a)=1$ for all $a\in\A$. With the identity tensor, we may define the equivalent of a matrix (partial) trace, the (partial) \textit{reduction} of a tensor $\rd$, given by the contraction with the identity tensor. The partial reduction of the tensor $T_{\A\B}$ over the index $A$ is given by $\rd_{\A} [T_{\A\B}] \coloneqq T_{\A\B} * \id_{\A} = \sum_a T_{\A\B}(a,b)\ \forall \ a\in\A$, or, in a more compact form,
\begin{equation}
    \rd_{\A} [T_{\A\B}] \coloneqq T_{\A\B} * \id_{\A} = \sum_\A T_{\A\B}.
\end{equation}

We are now equipped to define the tensor analog of a trace-and-replace operation for matrices, which we introduce here in the form of a {\it reduce-and-replace} operation acting on the random variable $\A$ of $T_{\A\B}$,
\begin{equation}
    {}_{\A} T_{\A\B} \coloneqq \rd_\A[T_{\A\B}] \otimes \frac{\id_\A}{d_\A},
\end{equation}
where $d_\A\coloneqq|A|=|\{a\}_a|$ denotes the dimension of the variable $\A$, which is given by its cardinality. Explicitly, ${}_{\A} T_{\A\B}(a,b) \coloneqq\sum_{a'}T_{\A\B}(a',b) \frac{\id_\A(a)}{d_\A} \ \forall \ a\in\A,b\in\B$.
It is straightforward to show that the reduce-and-replace operation is idempotent [i.e., ${}_{\A}({}_{\A} T_{\A\B}) ={}_{\A} T$], commutative [i.e., ${}_{\A}({}_{\B} T)= {}_{\B}({}_{\A} T)$], preserves the elementwise nonnegativity of the tensor, and is self-adjoint w.r.t.\ the tensor composition operation (i.e., $W * {}_\A T = {}_\A W * T$).
\\

\textit{\textbf{Characterization of operations and processes.}}
Here, we present an equivalent characterization to the one presented in the main text of local operations, nonsignaling operations, process tensors, process tensors that respect NSP, and finally process tensors that respect NSWSE, called boxworld processes. The proofs of the equivalence between these definitions are provided in the Appendices. The following characterizations are given in terms of positivity constraints and linear constraints, the latter expressed in terms of reduce-and-replace operations over tensors.

\textit{Probabilistic local operations} $T^{A|X}_{\I\Op|\Ip\O}$ are characterized as
\begin{subequations}
\begin{align}
    T^{\A|\X} &\geq 0 \label{eq::op_positivity} \\
    \rd (T^{\X}) &= d_{\Ip}d_{\O} \label{eq::op_totalred} \\
    _{\I\Op} T^{\X} &= \, _{\Ip\I\O\Op} T^{\X} \label{eq::op_proj1} \\
    _{\Op}T^{\X} &= \, _{\O\Op}T^{\X}, \label{eq::op_proj2}
\end{align}
\end{subequations}
where $T^{\X}_{\I\Op|\Ip\O}\coloneqq\sum_\A T^{A|X}_{\I\Op|\Ip\O}$ is a deterministic local operation. This formulation follows directly from their definition in Eq.~\eqref{eq::op_probabilistic}.

\textit{Probabilistic nonsignaling operations} $\overline{T}^{A|X}_{\I\Op|\Ip\O}$ are valid local operations, which satisfy Eqs.~\eqref{eq::op_positivity}--\eqref{eq::op_proj2}, that additionally respect at least one of the following two sets of conditions, for the trivial and nontrivial case. In the trivial case, where all operations in the set $\overline{T}^{\X}$ are the same for every $\X$, one can write simply $\overline{T}^{\X}=\overline{T}^{\X'}$. In the nontrivial case, where the operations in the set $\overline{T}^{\X}$ depend on the value of $\X$, nonsignaling operations are characterized by
\begin{subequations}
\begin{align}
    _{\I}\overline{T}^{\X} &= \, _{\I}\overline{T}^{\X'} \label{eq::opNS_proj1} \\
    _{\O} \overline{T}^{\X} &= \overline{T}^{\X}, \label{eq::opNS_proj2}
\end{align}
\end{subequations}
where $\overline{T}^{\X}_{\I\Op|\Ip\O} = \sum_\A\overline{T}^{A|X}_{\I\Op|\Ip\O}$. This formulation follows directly from the definition in Eq.~\eqref{eq::op_nonsignaling}. The fact that nonsignaling operations necessarily fall into one of these two cases is proven in App.~\ref{app::NSoperations}  (Theorem~\ref{thm::nonsignalingoperations}).

\textit{Bipartite process tensors} $W_{\Iap\Oa\Ibp\Ob|\Ia\Oap\Ib\Obp}$ are characterized as
\begin{subequations}
\begin{align}
    & W \geq 0 \label{eq::W_positivity} \\
    & r(W) = d_{\Ia}d_{\Oap}d_{\Ib}d_{\Obp} \label{eq::W_totalred} \\
    & _{\mathbb{A}(1-\Obp+\Ob\Obp-\Ib\Ob\Obp)}W = 0 \label{eq::W_proj1} \\
    &_{(1-\Oap+\Oa\Oap-\Ia\Oa\Oap)\mathbb{B}}W = 0 \label{eq::W_proj2} \\
    &_{(1-\Oap+\Oa\Oap-\Ia\Oa\Oap)(1-\Obp+\Ob\Obp-\Ib\Ob\Obp)}W = 0, \label{eq::W_proj3}
\end{align}
\end{subequations}
where $\mathbb{A}\coloneqq\Iap\Ia\Oa\Oap$ and $\mathbb{B}\coloneqq\Ibp\Ib\Ob\Obp$ represent the collective tensor indices of each party. We prove that this formulation is equivalent to the one in Eq.~\eqref{eq::bornrule} in App.~\ref{app::processtensors}  (Theorem~\ref{thm::processtensor}).

\textit{Bipartite process tensors that satisfy the condition of nonsignaling preservation (NSP)} are valid process tensors $W_{\Iap\Oa\Ibp\Ob|\Ia\Oap\Ib\Obp}$ [i.e., they satisfy Eqs.~\eqref{eq::W_positivity}--\eqref{eq::W_proj3}] that additionally respect 
\begin{subequations}
\begin{align}
    _{\Oa(\Oap\Obp\Iap\Ibp)}W &= \, _{\Oa\Ia(\Oap\Obp\Iap\Ibp)}W \label{eq::W_NSP_A1} \\
    _{\Ob(\Oap\Obp\Iap\Ibp)}W &= \, _{\Ob\Ib(\Oap\Obp\Iap\Ibp)}W \label{eq::W_NSP_B1} \\
    _{\Oa(1-\Oap+\Oap\Iap}&_{-\Obp\Oap\Iap)}W \nonumber \\
    = \, _{\Oa\Ia}&_{(1-\Oap+\Oap\Iap-\Obp\Oap\Iap)}W \label{eq::W_NSP_A2} \\
    _{\Ob(1-\Obp+\Obp\Ibp}&_{-\Oap\Obp\Ibp)}W \nonumber \\
    = \, _{\Ob\Ib}&_{(1-\Obp+\Obp\Ibp-\Oap\Obp\Ibp)}W. \label{eq::W_NSP_B2}
\end{align}
\end{subequations}
We prove that this formulation is equivalent to the one in Eq.~\eqref{eq::W_NSP} in App.~\ref{app::NSP}  (Theorem~\ref{thm::processtensorsNSP}).

\textit{Bipartite process tensors that satisfy the condition of no signaling without system exchange (NSWSE)} are valid process tensors $W_{\Iap\Oa\Ibp\Ob|\Ia\Oap\Ib\Obp}$ [i.e., they satisfy Eqs.~\eqref{eq::W_positivity}--\eqref{eq::W_proj3}] that additionally respect 
\begin{subequations}
\begin{align}
    _{\Oa}W &= \, _{\Oa\Ia}W \label{eq::W_NSWSE_A} \\
    _{\Ob}W &= \, _{\Ob\Ib}W, \label{eq::W_NSWSE_B} 
\end{align}
\end{subequations}
a condition strictly stronger than NSP. We prove that this formulation is equivalent to the one in Eq.~\eqref{eq::W_NSWSE} in App.~\ref{app::NSWSE}  (Theorem~\ref{thm::processtensorsNSWSE}).

Finally, \textit{boxworld processes} are defined to be processes tensors that satisfy NSWSE. Applying Eqs.~\eqref{eq::W_NSWSE_A}--\eqref{eq::W_NSWSE_B} to Eqs.~\eqref{eq::W_proj1}--\eqref{eq::W_proj3}, the characterization of boxworld processes can be further simplified to 
\begin{subequations}
\begin{align}
    W &\geq 0  \label{eq::bwW_positivity} \\
    r(W) &= d_{\Ia}d_{\Oap}d_{\Ib}d_{\Obp} \label{eq::bwW_totalred} \\
    _{\mathbb{A}}W &= \, _{\mathbb{A}\Obp}W \label{eq::bwW_proj1} \\
    _{\mathbb{B}}W &= \, _{\Oap\mathbb{B}}W \label{eq::bwW_proj2} \\
    W &= \, _{\Oap}W + \, _{\Obp}W - \, _{\Oap\Obp}W \label{eq::bwW_proj3} \\
    _{\Oa}W &= \, _{\Ia\Oa}W \label{eq::bwW_proj4} \\
    _{\Ob}W &= \, _{\Ib\Ob}W, \label{eq::bwW_proj5}
\end{align}
\end{subequations}
where again $\mathbb{A}\coloneqq\Iap\Ia\Oa\Oap$ and $\mathbb{B}\coloneqq\Ibp\Ib\Ob\Obp$ represent the collective tensor indices of each party.
\\

%%%%%%%%%%%%%%
\noindent\textbf{Code availability.} 
All code developed for this work is available at our online repository in Ref.~\cite{code-repo}. 
\\

%%%%%%%%%%%%%%
\noindent\textbf{Acknowledgements.} 
This work was funded by the Austrian Science Fund (FWF) through Zukunftskolleg ZK03.
JB acknowledges funding from the Swiss National Science Foundation (SNSF) through the funding schemes SPF and NCCR SwissMAP.
\"{A}B acknowledges funding from the Swiss National Science Foundation (SNSF) through project~214808, and by the Hasler Foundation through project~24010.

%%%%%%%%%%%%%%%%%%%
%\bibliography{refs}
%%%%%%%%%%%%%%%%%%%


%%%%%%%%%%%%%%%%%%%%%%%%%%%%%%%%%%%%%%%%%%%%%%
\clearpage

\onecolumngrid
\appendix

\section*{APPENDIX}

The appendix is organized in the following way:

\begin{itemize}
    \item Appendix \ref{app::part1}: Operations and process tensors
    \begin{itemize}
        \item[--] \ref{app::operations}. Local operations
        \item[--] \ref{app::processtensors}. Process tensors
        \item[--] \ref{app::NSP}. Process tensors that respect nonsignaling preservation (NSP)
        \item[--] \ref{app::NSoperations}. Nonsignaling operations
        \item[--] \ref{app::NSWSE}. Process tensors that respect no signaling without system exchange (NSWSE)
        \item[--] \ref{app::boxworldprocesses}. Boxworld processes and causal order
         \item[--] \ref{sec::Sym}. Symmetry transformations of boxworld processes
    \end{itemize}
    \item Appendix \ref{app::part2}: Boxworld correlations
    \begin{itemize}
        \item[--] \ref{app::no2waysig}. No perfect two-way signaling in boxworld correlations
        \item[--] \ref{app::polytope}. The set of boxworld correlations is a polytope
        \item[--] \ref{app::causalcorrelations}. All causal correlations are achieved by boxworld processes
        \item[--] \ref{app::inequalities}. Violating causal inequalities with boxworld correlations
    \end{itemize}
\end{itemize}

%%%%%%%%%%%%%%%%%%%%%%%%%%%%%%%%%%%%%%%%%%%%
\section{Operations and process tensors}\label{app::part1}
%%%%%%%%%%%%%%%%%%%%%%%%%%%%%%%%%%%%%%%%%%%%

%%%%%%%%%%%%%%
\subsection{Local operations}\label{app::operations}

Operations that map sets of probability distributions into sets of probability distributions, both deterministically and probabilistically, have been characterized in Eq.~\eqref{eq::op_deterministic} and~\eqref{eq::op_probabilistic} in terms of pre- and post-processing operations. In Eqs.~\eqref{eq::op_positivity}--\eqref{eq::op_proj2}, an equivalent characterizations is given, expressed in terms of reduce-and-replace operations~\cite{rosset2020algebraic}.

Here we also provide a third equivalent characterization, expressed in terms of positivity and a parametrization of local operations by arbitrary tensors, which will be used in the proofs presented in the following sections. 

First, notice that Eqs.~\eqref{eq::op_proj1} and~\eqref{eq::op_proj2} can be expressed as individual projectors $\widetilde{P}_1(T^{\X}) = T^{\X} - \, _{\I\Op}T^{\X} + \, _{\Ip\I\O\Op}T^{\X}$ and $\widetilde{P}_2(T^{\X}) = T^{\X} - \, _{\Op}T^{\X} + \, _{\O\Op}T^{\X}$, and then combined into one projector $\widetilde{P}_T\coloneqq\widetilde{P}_1\circ\widetilde{P}_2$ that defines the subspace of deterministic operations, according to
\begin{equation}
    \widetilde{P}_T(T^{\X}) = T^{\X} - \, _{\Op}T^{\X} + \, _{\O\Op}T^{\X} - \, _{\I\O\Op}T^{\X} + \, _{\Ip\I\O\Op}T^{\X}. \label{eq::op_fullproj}
\end{equation}

Additionally, using the formalism of Ref.~\cite{milz2024characterising}, any tensors $T$ that have the same constant reduction $r(T)$ for every $T$ and are characterized by some projector $\widetilde{P}$ can be parametrized by an arbitrary tensor $\Xcal$ such that 
\begin{equation}\label{eq::tensor_parametrization}
    T = \rd(T)\frac{\id}{d_\mathbb{T}} + \widetilde{P}(\Xcal) - \, _{[\mathbb{T}]}\Xcal,
\end{equation}
where $\mathbb{T}$ are the collective random variables of the tensor $T$ and $d_\mathbb{T}\coloneqq|\mathbb{T}|$. 

Using this parametrization, all operations given by Eqs.~\eqref{eq::op_totalred} and~\eqref{eq::op_fullproj} can be rewritten as a function of a parameter tensor $\Xcal$, and therefore, every local operation $T^{\A|\X}$ can be expressed as
\begin{align}
    T^{\A|\X} &\geq 0 \label{eq::op_param_positivity} \\
    T^{\X} &= \frac{\id}{d_{\I}d_{\Op}} + \Xcal - \, _{\Op}\Xcal + \, _{\O\Op}\Xcal - \, _{\I\O\Op}\Xcal, \label{eq::op_parametrization}
\end{align}
for some tensor $\Xcal$ such that Eq.~\eqref{eq::op_param_positivity} is satisfied. For every tensor of $\Xcal$ such that the positivity condition is satisfy, the result $T^{\X}$ is a valid deterministic operation. In particular, the choice $\Xcal=0$ is valid choice, implying that $T^{\X} = \frac{\id}{d_{\I}d_{\Op}}$ is a valid set of local deterministic operations.

%%%%%%%%%%%%%%
\subsection{Process tensors}\label{app::processtensors}

In this section, we prove the characterization of bipartite process tensors provided in Eqs.~\eqref{eq::W_positivity}--\eqref{eq::W_proj3}, show that all process tensors correspond to valid sets of conditional probability distributions, and show that the process tensor in Eq.~\eqref{eq::ansatzW1} is valid. 
\\

\textbf{Characterization of bipartite process tensors.} We begin by recalling the definition of a process tensor, as the most general tensor $W=W_{\Iap\Oa\Ibp\Ob|\Ia\Oap\Ib\Obp}$ that satisfies
\begin{equation}\label{eq::bornrule_app}
    P_{\A\B|\X\Y} =  (T^{\A|\X} \otimes T^{\B|\Y}) * W,
\end{equation}
for any set of probabilistic local operations $T^{\A|\X}=T^{\A|\X}_{\Ia\Oap|\Iap\Oa}$ and $T^{\B|\Y}=T^{\B|\Y}_{\Ib\Obp|\Ibp\Ob}$, where $P_{\A\B|\X\Y}$ are valid sets of probability distributions referred to as correlations. 

Let us now prove the following characterization: 

\begin{theorem}[Characterization of bipartite process tensors]\label{thm::processtensor}
A tensor $W=W_{\Iap\Oa\Ibp\Ob|\Ia\Oap\Ib\Obp}$ is a bipartite process tensor iff it is of the form 
\begin{align}
    W &\geq 0 \label{eq::W_positivity_app} \\
    r(W) &= d_{\Ia}d_{\Oap}d_{\Ib}d_{\Obp} \label{eq::W_totalred_app} \\
    _{\mathbb{A}(1-\Obp+\Ob\Obp-\Ib\Ob\Obp)}W &= 0 \label{eq::W_proj1_app} \\
    _{(1-\Oap+\Oa\Oap-\Ia\Oa\Oap)\mathbb{B}}W &= 0 \label{eq::W_proj2_app} \\
    _{(1-\Oap+\Oa\Oap-\Ia\Oa\Oap)(1-\Obp+\Ob\Obp-\Ib\Ob\Obp)}W &= 0, \label{eq::W_proj3_app}
\end{align}
where $\mathbb{A}\coloneqq\Iap\Ia\Oa\Oap$ and $\mathbb{B}\coloneqq\Ibp\Ib\Ob\Obp$ represent the collective tensor indices of each party.
\end{theorem}

\begin{proof}
Let us start by rewriting Eq.~\eqref{eq::bornrule_app} explicitly, as
\begin{equation}
    P_{\A\B|\X\Y} = \sum_{\substack{\Iap\Ia\Oa\Oap\\\Ibp\Ib\Ob\Obp}} T^{\A|\X}_{\Ia\Oap|\Iap\Oa}\, T^{\B|\Y}_{\Ib\Obp|\Ibp\Ob}\, W_{\Iap\Oa\Ibp\Ob|\Ia\Oap\Ib\Obp}.
\end{equation}

With respect to positivity, the elementwise nonnegativity of $P_{\A\B|\X\Y}$, $T^{\A|\X}$ and $T^{\B|\Y}$ implies also that $W$ is elementwise nonnegative, denote as $W\geq 0$. In fact, consider the local operations $T^{\A|\X}_{\Ia\Oap|\Iap\Oa}$ and $T^{\B|\Y}_{\Ib\Obp|\Ibp\Ob}$, which are zero everywhere except for a fix set of indices ${\xi \alpha' \xi' \alpha}$ and ${\zeta \beta' \zeta' \beta}$, respectively. For these operations, given by $T^{\A|\X}_{\Ia\Oap|\Iap\Oa} =  T^{\A|\X}_{\xi\alpha'|\xi'\alpha} = \varepsilon \, \delta_{\xi,\Ia} \delta_{\alpha', \Oap} \delta_{\xi',\Iap} \delta_{\alpha,\Oa}$, for $0< \varepsilon< 1/(d_{\Ia}d_{\Oap})$ and $T^{\B|\Y}_{\Ib\Obp|\Ibp\Ob}=T^{\B|\Y}_{\zeta\beta'|\zeta'\beta} = \varepsilon \, \delta_{\zeta,\Ib} \delta_{\beta', \Obp} \delta_{\zeta',\Ibp} \delta_{\beta,\Ob}$, for $0< \varepsilon< 1/(d_{\Ib}d_{\Obp})$, one has that
\begin{equation}\label{eq:sing_W_pos}
    W * T^{\A|\X}_{\xi\alpha'|\xi'\alpha} * T^{\B|\Y}_{\zeta\beta'|\zeta'\beta}  = \varepsilon^2 W_{\xi'\alpha\zeta'\beta|\xi\alpha'\zeta\beta'} = P_{\A\B|\X\Y} \leq 0,
\end{equation}
where the positivity comes from the interpretation of the object $P_{\A\B|\X\Y}=W*T^{\A|\X}*T^{\B|\Y}$ as a probability distribution, thus implying  that the  $(\xi \alpha' \xi' \alpha)$-element of $W$ must be nonnegative. Since the choice of the indices $(\xi \alpha' \xi' \alpha)$ appearing in the construction of $ T_{\xi \alpha' \xi' \alpha}$ is arbitrary, we conclude that $W$ must be nonnegative for all choices and, hence, is elementwise nonnegative.

With respect to positivity, we write explicitly the constraint that 
\begin{equation}
    \sum_{\A\B} P_{\A\B|\X\Y} = \sum_{\substack{\Iap\Ia\Oa\Oap\\\Ibp\Ib\Ob\Obp}} T^{\X}_{\Ia\Oap|\Iap\Oa}\, T^{\Y}_{\Ib\Obp|\Ibp\Ob}\, W_{\Iap\Oa\Ibp\Ob|\Ia\Oap\Ib\Obp} = 1.
\end{equation}
We use the parametrization of local operations in Eq.~\eqref{eq::op_parametrization} to write $T^{\X}=\frac{\id^{\mathbb{A}}}{d_{\Ia}d_{\Oap}} + \, _{[1-\Oap+\Oa\Oap-\Ia\Oa\Oap]}\Xcal$ for some parameter tensor $\Xcal$ and $T^{\Y}=\frac{\id^{\mathbb{B}}}{d_{\Ib}d_{\Obp}} + \, _{[1-\Obp+\Ob\Obp-\Ib\Ob\Obp]}\Ycal$ for some parameter tensor $\Ycal$. Hence,
\begin{equation}
    1 = \sum_{\substack{\Iap\Ia\Oa\Oap\\\Ibp\Ib\Ob\Obp}} \left(\frac{\id^{\mathbb{A}}}{d_{\Ia}d_{\Oap}} + \, _{[1-\Oap+\Oa\Oap-\Ia\Oa\Oap]}\Xcal \right)\, \left(\frac{\id^{\mathbb{B}}}{d_{\Ib}d_{\Obp}} + \, _{[1-\Obp+\Ob\Obp-\Ib\Ob\Obp]}\Ycal \right)\, W.
\end{equation}

Starting with the case of $\Xcal=0$ and $\Ycal=0$ we have that
\begin{equation}\label{eq::e1}
    1 = \sum_{\substack{\Iap\Ia\Oa\Oap\\\Ibp\Ib\Ob\Obp}} \frac{\id^{\mathbb{A}}}{d_{\Ia}d_{\Oap}}\, \frac{\id^{\mathbb{B}}}{d_{\Ib}d_{\Obp}}\, W = \frac{1}{d_{\Ia}d_{\Oap}d_{\Ib}d_{\Obp}} r(W) \iff r(W) = d_{\Ia}d_{\Oap}d_{\Ib}d_{\Obp}.
\end{equation}

Now we consider the case where $\Xcal\neq0$ and $\Ycal=0$.
\begin{align}
    1 &= \sum_{\substack{\Iap\Ia\Oa\Oap\\\Ibp\Ib\Ob\Obp}} \left(\frac{\id^{\mathbb{A}}}{d_{\Ia}d_{\Oap}} + \, _{[1-\Oap+\Oa\Oap-\Ia\Oa\Oap]}\Xcal \right)\, \left(\frac{\id^{\mathbb{B}}}{d_{\Ib}d_{\Obp}} \right)\, W \\
    &= \sum_{\substack{\Iap\Ia\Oa\Oap\\\Ibp\Ib\Ob\Obp}} \left(\frac{r(W)}{d_{\Ia}d_{\Oap}d_{\Ib}d_{\Obp}} + \, _{[1-\Oap+\Oa\Oap-\Ia\Oa\Oap]}\Xcal\,W \right) \label{eq::e2} \\
    &= 1 + \sum_{\Iap\Ia\Oa\Oap} \Xcal \sum_{\Ibp\Ib\Ob\Obp} \, _{[1-\Oap+\Oa\Oap-\Ia\Oa\Oap]}W, \label{eq::e3}
\end{align}
where from Eq.~\eqref{eq::e2} to Eq.~\eqref{eq::e3} we used the total reduction of $W$ implied in Eq.~\eqref{eq::e1} and the self-adjointness of the reduce-and-replace operation (see Sec.~\ref{sec::methods}). Since this expression has to hold for any tensor $\Xcal$, it will hold iff 
\begin{equation}\label{eq::e4}
    _{(1-\Oap+\Oa\Oap-\Ia\Oa\Oap)\mathbb{B}}W = 0.
\end{equation}

\noindent For the case $\Xcal=0$ and $\Ycal\neq0$, we arrive at the equivalent expression 
\begin{equation}\label{eq::e5}
    _{\mathbb{A}(1-\Obp+\Ob\Obp-\Ib\Ob\Obp)}W = 0.
\end{equation}
For the final case where $\Xcal\neq0$ and $\Ycal\neq0$, we have that 
\begin{align}
    1 =& \sum_{\substack{\Iap\Ia\Oa\Oap\\\Ibp\Ib\Ob\Obp}} \left(\frac{\id^{\mathbb{A}}}{d_{\Ia}d_{\Oap}} + \, _{[1-\Oap+\Oa\Oap-\Ia\Oa\Oap]}\Xcal \right)\, \left(\frac{\id^{\mathbb{B}}}{d_{\Ib}d_{\Obp}} + \, _{[1-\Obp+\Ob\Obp-\Ib\Ob\Obp]}\Ycal \right)\, W \\
    =& 1 + 0 + 0 + \sum_{\substack{\Iap\Ia\Oa\Oap\\\Ibp\Ib\Ob\Obp}} \, _{[1-\Oap+\Oa\Oap-\Ia\Oa\Oap]}\Xcal\, _{[1-\Obp+\Ob\Obp-\Ib\Ob\Obp]}\Ycal\, W \label{eq::e6} \\
    =& 1 + \sum_{\substack{\Iap\Ia\Oa\Oap\\\Ibp\Ib\Ob\Obp}} \Xcal \, \Ycal \, _{(1-\Oap+\Oa\Oap-\Ia\Oa\Oap)(1-\Obp+\Ob\Obp-\Ib\Ob\Obp)} W, 
\end{align}
where to get the first 3 terms of Eq.~\eqref{eq::e6} we applied Eqs.~\eqref{eq::e1},~\eqref{eq::e4}, and~\eqref{eq::e5}. This expression must hold for all tensors $\Xcal$ and $\Ycal$, hence 
\begin{equation}
    _{(1-\Oap+\Oa\Oap-\Ia\Oa\Oap)(1-\Obp+\Ob\Obp-\Ib\Ob\Obp)} W = 0,
\end{equation}
which concludes the proof.
\end{proof}

From this characterization, we can check that every process tensor $W=W_{\Iap\Oa\Ibp\Ob|\Ia\Oap\Ib\Obp}$ corresponds indeed to a set of conditional probability distributions from inputs $\Ia\Oap\Ib\Obp$ to outputs $\Iap\Oa\Ibp\Ob$. The positivity condition of $W\geq0$ has already been proven. We must only check that $\sum_{\Iap\Oa\Ibp\Ob}W=1$. The normalization condition of a set of probability distributions $P = P_{\Iap\Oa\Ibp\Ob|\Ia\Oap\Ib\Obp}$ can be expressed in terms of reduce-and-replace operations as 
\begin{align}
    r(P) &= d_\Ia d_\Oap d_\Ib d_\Obp \\
    _{\Iap\Oa\Ibp\Ob}P &= \, _{\Ia\Oap\Ib\Obp\Iap\Oa\Ibp\Ob}P.
\end{align}
The first condition is automatically satisfied by any process tensor $W$. To check that the second is also satisfied by any $W$, it is convenient to write all three constraints in Eqs.~\eqref{eq::W_proj1_app}--\eqref{eq::W_proj3_app} in a single projector, as 
\begin{align}
    W = \widetilde{P}_W(W) &= \,_{\mathbb{A}(\Obp-\Ob\Obp+\Ib\Ob\Obp)}W -\,_{\mathbb{A}}W +\,_{(\Oap-\Oa\Oap+\Ia\Oa\Oap)\mathbb{B}}W -\,_{\mathbb{B}}W \label{eq::W_fullproj} \\
    & +\, _{(\Oap-\Oa\Oap+\Ia\Oa\Oap)}W +\, _{(\Obp-\Ob\Obp+\Ib\Ob\Obp)}W -\, _{(\Oap-\Oa\Oap+\Ia\Oa\Oap)(\Obp-\Ob\Obp+\Ib\Ob\Obp)}W.  \nonumber
\end{align}
From this expression, it is easy to check that $_{\Iap\Oa\Ibp\Ob}W = \, _{\Ia\Oap\Ib\Obp\Iap\Oa\Ibp\Ob}W.$
\\

\textbf{Validity of the $W_{\diamond}$ in Eq.~\eqref{eq::ansatzW1}.} Given this characterization, we can show that the process tensor in Eq.~\eqref{eq::ansatzW1} is a valid process tensor. We recall that the process tensor $W_{\diamond}$ in question is given by 
$ W_{\diamond} = \delta_{\Iap,\phi}\,\delta_{\Ibp,\phi}\,P_{\Oa\Ob|\Ia\Ib}$, where $\phi$ is a constant, $\delta_{a,b}$ is the Kronecker delta function, and $P_{\Oa\Ob|\Ia\Ib}$ is any set of probability distributions. Even more generally, we can show that any $W$ of the form of 
\begin{equation}
    W = P_{\Iap}P_{\Ibp}P_{\Oa\Ob|\Ia\Ib}
\end{equation}
is a valid process tensor. Elementwise nonnegativity [Eq.~\eqref{eq::W_positivity_app}] and total reduction [Eq.~\eqref{eq::W_totalred_app}] are straightforward to check. To check that Eqs.~\eqref{eq::W_proj1_app}--\eqref{eq::W_proj3_app} are also satisfied, it is sufficient to recognize that if $W=P_{\Iap}P_{\Ibp}P_{\Oa\Ob|\Ia\Ib}$, then it satisfies $W = \, _{\Oap}W= \, _{\Obp}W$ and $ _{\Oa\Ob}W= \, _{\Oa\Ob\Ia\Ib}W$. Then,
\begin{align}
    _{\mathbb{A}(1-\Obp+\Ob\Obp-\Ib\Ob\Obp)}W &= \, _{\mathbb{A}(1-1+\Ob-\Ib\Ob)}W \\ 
    &= \, _{\Iap\Ia\Oa\Oap(\Ob-\Ib\Ob)}W \\
    &= \, _{\Iap\Ia\Oa\Oap(\Ib\Ob-\Ib\Ob)}W \\
    &= 0,
\end{align}
\begin{align}
    _{(1-\Oap+\Oa\Oap-\Ia\Oa\Oap)\mathbb{B}}W &= \, _{(1-1+\Oa-\Ia\Oa)\mathbb{B}}W \\ 
    &= \, _{(\Oa-\Ia\Oa)\Ibp\Ib\Ob\Obp}W  \\
    &= \, _{(\Ia\Oa-\Ia\Oa)\Ibp\Ib\Ob\Obp}W  \\
    &= 0,
\end{align}
and
\begin{align}
    _{(1-\Oap+\Oa\Oap-\Ia\Oa\Oap)(1-\Obp+\Ob\Obp-\Ib\Ob\Obp)}W &= \, _{(1-1+\Oa-\Ia\Oa)(1-1+\Ob-\Ib\Ob)}W \\ 
    &= \, _{(\Oa\Ob-\Oa\Ib\Ob-\Ia\Oa\Ob+\Ia\Oa\Ib\Ob)}W \\
    &= \, _{(\Ia\Oa\Ib\Ob-\Ia\Oa\Ib\Ob-\Ia\Oa\Ib\Ob+\Ia\Oa\Ib\Ob)}W \\
    &= 0.
\end{align}

%%%%%%%%%%%%%%
\subsection{Process tensors that respect nonsignaling preservation (NSP)}\label{app::NSP}

In this section, we prove the characterization of process tensors that satisfy NSP provided in Eqs.~\eqref{eq::W_NSP_A1} and~\eqref{eq::W_NSP_B2}, prove that these conditions are equivalent to complete nonsignaling preservation, and show that the tensor in Eq.~\eqref{eq::ansatzW2} is a process tensor that respects NSP.
\\

\textbf{Characterization of process tensors that satisfy nonsignaling preservation.}
We begin by recalling the definition of a nonsignaling-preserving process tensor, from Eq.~\eqref{eq::W_NSP} in the main text. Formally, a nonsignaling-preserving process tensor $W$ is one that, for all sets of nonsignaling probability distributions $P^\text{NS}_{\Oap\Obp|\Iap\Ibp}$, 
\begin{equation}
    W * P^\text{NS}_{\Oap\Obp|\Iap\Ibp} = P^\text{NS}_{\Oa\Ob|\Ia\Ib} 
\end{equation}
where $P^\text{NS}_{\Oa\Ob|\Ia\Ib}$ is a set of nonsignaling probability distributions. In order to prove that this definition is equivalent to the characterization presented in Eqs.~\eqref{eq::W_NSP_A1}--\eqref{eq::W_NSP_B2} in the main text, we start be splitting it into two conditions, no signaling from Alice to Bob and from Bob to Alice. In the first case, this amounts to $W * P^{A\prec B}_{\Oap\Obp|\Iap\Ibp} = P^{A\prec B}_{\Oa\Ob|\Ia\Ib}$ for all $P^{A\prec B}_{\Oap\Obp|\Iap\Ibp}$ that satisfy Eq.~\eqref{eq::probNS_AprecB}, or, equivalently
\begin{equation}\label{eq::W_NSP_AprecB}
    \sum_{\Ob} W * P^{A\prec B}_{\Oap\Obp|\Iap\Ibp} = P_{\Oa|\Ia\Ib} 
\end{equation}
where $P_{\Oa|\Ia\Ib}$ is independent of $\Ib$ for all $P^{A\prec B}_{\Oap\Obp|\Iap\Ibp}$. The expression is analogous for nonsignaling from Bob to Alice.

In the reduce-and-replace representation, a set of probability distributions that is nonsignaling from Bob to Alice $P^{A\prec B}=P^{A\prec B}_{\Oap\Obp|\Iap\Ibp}$ can be expressed as 
\begin{align}
    P^{A\prec B} &\geq 0 \\
    r(P^{A\prec B}) &= d_\Iap d_\Ibp \label{eq::h3} \\
    _{\Oap\Obp}P^{A\prec B} &= \, _{\Oap\Obp\Iap\Ibp}P^{A\prec B} \label{eq::h1} \\
    _{\Obp}P^{A\prec B} &= \, _{\Obp\Ibp}P^{A\prec B}. \label{eq::h2}
\end{align}
Let us parameterize these sets of probability distributions similarly to how local operations were parametrized in Eq.~\eqref{eq::op_parametrization}, using the expression in Eq.~\eqref{eq::tensor_parametrization}. First, we recognize that Eqs.~\eqref{eq::h1} and~\eqref{eq::h2} can be written individual projectors $\widetilde{P}_3(P^{A\prec B}) = P^{A\prec B} - \, _{\Oap\Obp}P^{A\prec B} + \, _{\Oap\Obp\Iap\Ibp}P^{A\prec B}$ and $\widetilde{P}_4(P^{A\prec B}) =  P^{A\prec B} - \, _{\Obp}P^{A\prec B} + \, _{\Obp\Ibp}P^{A\prec B}$ that can be combined into a single projector $\widetilde{P}_\text{NS} \coloneqq \widetilde{P}_3 \circ \widetilde{P}_4$ such that 
\begin{equation}\label{eq::probNS_proj}
    \widetilde{P}_\text{NS}(P^{A\prec B}) = \, _{[1-\Obp+\Obp\Ibp-\Oap\Obp\Ibp+\Oap\Obp\Iap\Ibp]}P^{A\prec B}. 
\end{equation}

The projector in Eq.~\eqref{eq::probNS_proj} combine with the total reduction in Eq.~\eqref{eq::h3} allow us to parameterize any $P^{A\prec B}$ as
\begin{equation}\label{eq::f4}
    P^{A\prec B}_{\Oap\Obp|\Iap\Ibp} = \frac{\id}{d_\Oap d_\Obp} + \, _{[1-\Obp_\Obp\Ibp+\Oap\Obp\Ibp}\Xcal
\end{equation}
for some tensor $\Xcal$ such that $P^{A\prec B}\geq0$.

We are now equipped to prove the following characterization theorem.

\begin{theorem}[Characterization of bipartite process tensors that satisfy nonsignaling preservation]\label{thm::processtensorsNSP}
A bipartite process tensor $W$ is nonsignaling preserving iff it is of the form 
\begin{align}
    _{\Oa(\Oap\Obp\Iap\Ibp)}W &= \, _{\Oa\Ia(\Oap\Obp\Iap\Ibp)}W \label{eq::W_NSP_A1_app} \\
    _{\Ob(\Oap\Obp\Iap\Ibp)}W &= \, _{\Ob\Ib(\Oap\Obp\Iap\Ibp)}W \label{eq::W_NSP_B1_app} \\
    _{\Oa(1-\Oap+\Oap\Iap-\Obp\Oap\Iap)}W &= \, _{\Ia\Oa(1-\Oap+\Oap\Iap-\Obp\Oap\Iap)}W \label{eq::W_NSP_A2_app} \\
    _{\Ob(1-\Obp+\Obp\Ibp-\Oap\Obp\Ibp)}W &= \, _{\Ib\Ob(1-\Obp+\Obp\Ibp-\Oap\Obp\Ibp)}W. \label{eq::W_NSP_B2_app}
\end{align}
\end{theorem}

\begin{proof}
We begin analyzing the case of nonsignaling from Alice to Bob by substituting the parametrized $P^{A\prec B}_{\Oap\Obp|\Iap\Ibp}$ from Eq.~\eqref{eq::f4} into Eq.~\eqref{eq::W_NSP_AprecB}, to get
\begin{equation}
    \sum_{\Ob} W * P^{A\prec B}_{\Oap\Obp|\Iap\Ibp} = \sum_{\Ob} \sum_{\substack{\Oap\Obp\\\Iap\Ibp}} W \left(\frac{\id}{d_\Oap d_\Obp} + \, _{[1-\Obp\Obp\Ibp+\Oap\Obp\Ibp]}\Xcal\right),
\end{equation}
which must be independent of $\Ib$. Starting with the case where $\Xcal=0$, one gets that $\sum_{\Ob} \sum_{\substack{\Oap\Obp\\\Iap\Ibp}} W$ must be independent of $\Ib$, or, equivalently, that 
\begin{equation}
    _{\Ob(\Oap\Obp\Iap\Ibp)}W = \, _{\Ob\Ib(\Oap\Obp\Iap\Ibp)}W.
\end{equation}
For the case of $\Xcal\neq0$, we have that the term
\begin{equation}
    \sum_{\Ob} \sum_{\substack{\Oap\Obp\\\Iap\Ibp}} W\, _{[1-\Obp\Obp\Ibp+\Oap\Obp\Ibp]}\Xcal = \sum_{\substack{\Oap\Obp\\\Iap\Ibp}} \Xcal \, \sum_{\Ob} \, _{[1-\Obp\Obp\Ibp+\Oap\Obp\Ibp]} W
\end{equation}
must be independent of $\Ib$ for all $\Xcal$, or, equivalently, that
\begin{equation}
    _{\Ob(1-\Obp\Obp\Ibp+\Oap\Obp\Ibp)}W = \, _{\Ob\Ib(1-\Obp\Obp\Ibp+\Oap\Obp\Ibp)}W.
\end{equation}
From the restriction of nonsignaling from Bob to Alice, we arrive at analogous conditions.
\end{proof}

In fact, it can be checked that this characterization also holds for process tensors that are required to be \textit{completely nonsignaling preserving}. A completely nonsignaling-preserving transformation is one that maps every nonsignaling box to another nonsignaling box, even when acting only on part of said nonsignaling box. Let $P^\text{NS}= P^\text{NS}_{\Oap\Oat\Obp\Obt|\Iap\Iat\Ibp\Ibt}$ be a nonsignaling box, which satisfies 
\begin{align}
    \sum_{\Obp\Obt} P^\text{NS}_{\Oap\Oat\Obp\Obt|\Iap\Iat\Ibp\Ibt} &= P_{\Oap\Oat|\Iap\Iat} \label{eq::probNS_AprecB_double}  \\
    \sum_{\Oap\Oat} P^\text{NS}_{\Oap\Oat\Obp\Obt|\Iap\Iat\Ibp\Ibt} &= P^\text{NS}_{\Obp\Obt|\Ibp\Ibt}, \label{eq::probNS_BprecA_double}
\end{align}
where $P_{\Oap\Oat|\Iap\Iat}$ is independent of $\Ibp\Ibt$ and $P^\text{NS}_{\Obp\Obt|\Ibp\Ibt}$ is independent of $\Iap\Iat$. Then, $W = W_{\Iap\Ibp\Oa\Ob|\Ia\Ib\Oap\Obp}$ is completely nonsignaling preserving if, for all $P^\text{NS}_{\Oap\Oat\Obp\Obt|\Iap\Iat\Ibp\Ibt}$,
\begin{align}
(W_{\Iap\Ibp\Oa\Ob|\Ia\Ib\Oap\Obp}\otimes\id^{\Oat\Obt\Iat\Ibt}) * P^\text{NS}_{\Oap\Oat\Obp\Obt|\Iap\Iat\Ibp\Ibt} =  P^\text{NS}_{\Oa\Oat\Ob\Obt|\Ia\Iat\Ib\Ibt}
\end{align}
where $P^\text{NS}_{\Oa\Oat\Ob\Obt|\Ia\Iat\Ib\Ibt}$ is a nonsignaling box, which satisfies Eqs.~\eqref{eq::probNS_AprecB_double}--\eqref{eq::probNS_BprecA_double} with a relabeling between the primed and nonprimed variables.

The steps in the proof of Thm.~\ref{thm::processtensorsNSP} also hold in this case, leading to the same characterization and showing that the requirement of complete nonsignaling preservation does not impose any additional constraints as compared to simply nonsignaling preservation.
\\

\textbf{Validity of the $W_{\triangle}$ in Eq.~\eqref{eq::ansatzW2}.} Given this characterization, we can show that the process tensor in Eq.~\eqref{eq::ansatzW2} is a valid process tensor [respects Eqs.~\eqref{eq::W_positivity_app}--\eqref{eq::W_proj3_app}] and that it respects nonsignaling preservation [Eqs.~\eqref{eq::W_NSP_A1_app}--\eqref{eq::W_NSP_B2_app}]. We recall that the process tensor $W_{\triangle}$ in question is given by 
\begin{equation}\label{eq::ansatzW2_app}
    W_{\triangle} = \delta_{\Iap,\Ib}\,\delta_{\Oa,\phi}\,\delta_{\Ibp,\phi}\,\delta_{\Ob,\Oap}.
\end{equation}
We begin with the process tensor constraints. This tensor immediately satisfy the positivity constraint in Eq.~\eqref{eq::W_positivity_app}. Also notice that this tensor does not depend on the input variables $\Ia$ and $\Obp$, hence, $d_\Ia=d_\Obp=1$. Therefore, it also satisfies the total reduction condition in ~\eqref{eq::W_totalred_app}. Given the lack of dependence on $\Obp$, Eq.~\eqref{eq::W_proj1_app} reduces to 
\begin{equation}
    _{\mathbb{A}(1-\Obp+\Ob\Obp-\Ib\Ob\Obp)} W_{\triangle} =\, _{\Iap\Oa\Oap(\Ob-\Ib\Ob)}W_{\triangle} = 0,
\end{equation}
which is satisfied because once one applies the reduce-and-replaces over the variable $\Iap$, $W_{\triangle}$ no longer a depends on $\Ib$. Similarly, Eq.~\eqref{eq::W_proj2_app} reduces to 
\begin{equation}
     _{(1-\Oap+\Oa\Oap-\Ia\Oa\Oap)\mathbb{B}}W_{\triangle} =\, _{(1-\Oap)\Ibp\Ib\Ob}W_{\triangle} = 0,
\end{equation}
which is satisfied since the dependence on of $W_{\triangle}$ on $\Oap$ disappears once one applies a reduce-and-replace over the variable $\Ob$. Finally, Eq.~\eqref{eq::W_proj3_app} reduces to 
\begin{equation}
    _{(1-\Oap+\Oa\Oap-\Ia\Oa\Oap)(1-\Obp+\Ob\Obp-\Ib\Ob\Obp)}W_{\triangle} =\, _{(1-\Oap)(\Ob-\Ib\Ob)}W_{\triangle} =\, _{(1-\Oap)(1-\Ib)\Ob}W_{\triangle} = 0,
\end{equation}
which also holds from the independence of $_{\Ob}W_{\triangle}$ on $\Oap$. 

Now we check the NSP constraints. Equations~\eqref{eq::W_NSP_A1_app} and~\eqref{eq::W_NSP_A2_app} is trivially satisfied since $W_{\triangle}$ does not depend on $\Ia$. Equation~\eqref{eq::W_NSP_A1_app} is satisfies also because after the reduce-and-replace over variable $\Iap$, $W_{\triangle}$ no longer depends on $\Ib$. From the independence of $W_{\triangle}$ on $\Obp$, Eq.~\eqref{eq::W_NSP_B2_app} reduces to
\begin{equation}
    _{\Ob(\Ibp-\Oap\Ibp)}W = \, _{\Ib\Ob(\Ibp-\Oap\Ibp)}W,
\end{equation}
which is satisfied because $_{\Ob}W_{\triangle}$ is independent of $\Oap$. 

Hence, $W_{\triangle}$ is a nonsignaling-preserving process tensor.

%%%%%%%%%%%%%%
\subsection{Nonsignaling operations}\label{app::NSoperations}

In this section, we prove the characterization of nonsignaling operations provided in Eq.~\eqref{eq::op_nonsignaling} as a function of specific pre- and post-processing operations, and equivalently in Eqs.~\eqref{eq::opNS_proj1} and~\eqref{eq::opNS_proj2}, as a function of reduce-and-replace conditions. We recall here the definition of a set of deterministic nonsignaling operations, as a set of operations $\overline{T}^{\X}$ that, when acting locally on a box $P_{\O|\I}$, results in a box $P_{\Op|\Ip}=\overline{T}^{\X}*P_{\O|\I}$ which is independent of $\X$. A set of probabilistic nonsignaling operations $\overline{T}^{\A|\X}$ is defined as any nonnegative resolution of a set of deterministic nonsignaling operations, such that $\sum_A\overline{T}^{\A|\X}=\overline{T}^{\X}$. 

Explicitly, a set of probabilistic nonsignaling operations satisfy
\begin{equation}\label{eq:op_nonsignaling_def}
    \sum_{a} \overline{T}^{A|X}(a|x) * P_{\O|\I} = 
    \sum_{a} \overline{T}^{A|X}(a|x') * P_{\O|\I} \ \ \ \forall \ x,x'.
\end{equation}

\begin{theorem}[Characterization of nonsignaling operations]\label{thm::nonsignalingoperations}
Let $\overline{T}^{\X}\coloneqq \sum_{\A}\overline{T}^{\A|\X}$ be a set of deterministic nonsignaling operations. Then, $\overline{T}^{\X}$ satisfies either one of two cases:

Case 1 (trivial).--- No dependency at all on $X$:
\begin{align}
    \overline{T}^{\X} &= P_{\I|\Ip}\,P_{\Op|\Ip\O} \\
    &= \sum_\lambda \pi_\lambda\, D^\lambda_{\I|\Ip}\,D^\lambda_{\Op|\Ip\O}. \label{eq::op_nonsignaling_trivial}
\end{align}

Case 2 (nontrivial).--- $\Op$ does not depend on on $\O$ and $X$:
\begin{align}
    \overline{T}^{\X} &=P_{\I|\Ip \X}\,P_{\Op|\Ip} \\
    &= \sum_\lambda \pi_\lambda\, D^\lambda_{\I|\Ip X}\,D^\lambda_{\Op|\Ip}. \label{eq::op_nonsignaling_nontrivial}
\end{align}
\end{theorem}

\begin{proof}
The first observation is that it is enough to analyze the constraint in Eq.~\eqref{eq:op_nonsignaling_def} for deterministic operations, since both $\overline{T}^{\X}$ and $P$ can be written as convex decomposition of them. Using the general form of $\overline{T}^{\X}$ and $P$ in terms of deterministic operations, the condition in Eq.~\eqref{eq:op_nonsignaling_def} amounts to
\begin{equation}
    (\overline{T}^{\X}*P)(o'|i')=\sum_{io} \delta_{i,f_T(i',x)} \delta_{o',g_T(i',o,x)} \delta_{o, h_P(i)} \text{ is independent of } x,
\end{equation}
where $f_T$ and $g_T$ are the deterministic functions defining $\overline{T}^{\X}$ and $h_P$ the deterministic function defining $P$. We can further manipulate this expression as
\begin{equation}
    \sum_{io} \delta_{i,f_T(i',x)} \delta_{o',g_T(i',o,x)} \delta_{o, h_P(i)} = \sum_{o} \delta_{o',g_T(i',o,x)} \delta_{o, h_P(f_T(i',x))}= \delta_{o',g_T(i', h_P(f_T(i',x)),x)}, 
\end{equation}
in other words, we find that $o'$ is deterministically generated as
\begin{equation}
    o' = g_T(i', h_P(f_T(i',x)),x),
\end{equation}
and we want this expression to be independent of $x$ for each function $o=h_P(i)$. This straightforwardly provides the two cases: case $(1)$ both $f_T$ and $g_T$ are independent of $x$, or, case $(2)$ $g_T$ is independent  of both $x$ and $o$.
\end{proof}

%%%%%%%%%%%%%%
\subsection{Process tensors that respect no signaling without system exchange (NSWSE)}\label{app::NSWSE}

In this section, we prove the characterization of process tensors that satisfy NSWSE provided in Eqs.~\eqref{eq::W_NSWSE_A} and~\eqref{eq::W_NSWSE_B}. We recall the definition of process tensors that satisfy NSWSE in Eq.~\eqref{eq::W_NSWSE}, as a process tensor that, when under the action of nonsignaling local operations, can only generate nonsignaling correlations, i.e., a process tensor $W$ such that 
\begin{equation}\label{eq::W_NSWSE_app}
    P^\text{NS}_{\A\B|\X\Y} = W * (\overline{T}^{\A|\X} \otimes \overline{T}^{\B|\Y}), 
\end{equation}
for all sets of nonsignaling operations $\overline{T}^{\A|\X}$ and $\overline{T}^{\B|\Y}$, where $P^\text{NS}_{\A\B|\X\Y}$ are nonsignaling correlations. We may now prove the following characterization:

\begin{theorem}[Characterization of bipartite process tensors that satisfy no signaling without system exchange]\label{thm::processtensorsNSWSE}
A bipartite process tensor $W$ satisfies no signaling without system exchange iff it is of the form 
\begin{align}
    _{\Oa}W &= \, _{\Ia\Oa}W \label{eq::W_NSWSE_A_app} \\
    _{\Ob}W &= \, _{\Ib\Ob}W. \label{eq::W_NSWSE_B_app}
\end{align}
\end{theorem}

\begin{proof}
First notice that, when considering nonsignaling operations of (trivial) case 1 [Eq.~\eqref{eq::op_nonsignaling_trivial}], no new constraints are imposed on $W$ from the imposition of the condition in Eq.~\eqref{eq::W_NSWSE_app}. Hence, only nonsignaling operations of (nontrivial) case 2 [Eq.~\eqref{eq::op_nonsignaling_nontrivial}] must be considered.

We begin by splitting Eq.~\eqref{eq::W_NSWSE_app} into one that respects nonsignaling from Alice to Bob and another from Bob to Alice. Starting with nonsignaling from Alice to Bob, our question is: given any set of nonsignaling operations $\overline{T}^{X}$ acting on Alice's side, what is the largest set of process tensors $W$, such that
\begin{equation}
    \sum_A P_{\A\B|\X\Y} = \overline{T}^{\X} * T^{\B|\Y} * W = P_{\B|\X\Y}
\end{equation}
is independent of $X$ for all general operations $T^{\B|\Y}$ acting on Bob's side. 

We exploit the commutative property of the $*$ composition to rewrite 
\begin{equation}
    P_{\B|\X\Y} = T^{\B|\Y} * \overline{T}^{\X} * W 
\end{equation}
and then explicitly writing the $*$ composition between Alice's variables, yielding
\begin{equation}
    P_{\B|\X\Y} = T^{\B|\Y} * \sum_{\Ia\Iap\Oa\Oap} \overline{T}^X_{\Ia\Oap|\Iap\Oa} \, W_{\Iap\Oa\Ibp\Ob|\Ia\Oap\Ib\Obp}.
\end{equation}
Substituting in the definition of nonsignaling operation of case 2 [Eq.~\eqref{eq::op_nonsignaling_nontrivial}], one gets
\begin{align}
    P_{\B|\X\Y} &= T^{\B|\Y} * \sum_{\Ia\Iap\Oa\Oap} \sum_\lambda \pi_\lambda\, D^\lambda_{\Ia|\Iap \X}\,D^\lambda_{\Oap|\Iap} \, W_{\Iap\Oa\Ibp\Ob|\Ia\Oap\Ib\Obp} \\
    &= T^{\B|\Y} * \sum_\lambda \pi_\lambda \sum_{\Ia\Iap\Oa\Oap}  \delta_{\Ia,f_\lambda(\Iap,\X)}\,\delta_{\Oap,g_\lambda(\Iap)} \, W_{\Iap\Oa\Ibp\Ob|\Ia\Oap\Ib\Obp} \\
    &= T^{\B|\Y} * \sum_\lambda \pi_\lambda \sum_{\Iap\Oa}  W_{\Iap\Oa\Ibp\Ob|f_\lambda(\Iap,X)g_\lambda(\Iap)\Ib\Obp}.
\end{align}
Our initial condition then is equivalent to saying that
\begin{equation}\label{eq::proofattempt}
\begin{split}
    P_{\B|\X\Y} &= T^{\B|\Y} * \sum_\lambda \pi_\lambda \sum_{\Iap}  \sum_{\Oa}W_{\Iap\Oa\Ibp\Ob|f_\lambda(\Iap,X)g_\lambda(\Iap)\Ib\Obp} \\
    &= T^{\B|\Y} * \sum_\lambda \pi_\lambda \sum_{\Iap} {V}_{\Iap\Ibp\Ob|f_\lambda(\Iap,\X)g_\lambda(\Iap)\Ib\Obp} \ \ \ \text{is independent of $\X$} \ \forall \ T^{\B|\Y}, \pi_\lambda, f_\lambda, g_\lambda.
    \end{split}
\end{equation}
This is equivalent to saying that ${V}_{\Iap\Ibp\Ob|\Ia,\Oap,\Ib\Obp}$, defined as $V:=\sum_{\Oa} W$, is independent of $\Ia$, as we show in the following. It is clear that this condition is sufficient. Let us prove that it is necessary. Since $\pi_\lambda$ is arbitrary, we can pick a $\{0,1\}$-valued distribution, which implies that we can evaluate the above expression simply over all functions $f$ and $g$. Let us pick for $g$ a constant function that fixes the index $\Oap$, and, for the moment, let us fix also all other indices except $\Iap$, i.e., $\Ibp,\Ib,\Ob,\Obp$. This can be obtained by a proper choice of an operation $T^{\B|\Y}$; see discussion above the derivation of Eq.~\eqref{eq:sing_W_pos}. To simplify the notation we denote the corresponding $V$ as $V^F_{\Iap|\Ia}$, where $\Ia=f(\Iap,\X)$. 

Since $f$ is arbitrary, we have the condition that
\begin{equation}\label{eq:VF_indep}
    \sum_\Iap V^F_{\Iap|f_\X(\Iap)} \text{ is independent of } f_\X.
\end{equation}
This implies that $V^F_{\Iap|\Ia}=V^F_{\Iap}$ for all $\Iap$. In fact, consider two functions $f_1$  and $f_2$ that differ only on their value on $\Ia=0$, we have
\begin{equation}
    \sum_\Iap V^F_{\Iap|f_1(\Iap)} = \sum_\Iap V^F_{\Iap|f_1(\Iap)} \Rightarrow V^F_{0|f_1(0)} = V^F_{0|f_2(0)}.
\end{equation}
This can be repeated for any possible value of $\Ia$, giving Eq.~\eqref{eq:VF_indep}. Since the other indices, i.e., $\Oap$, $\Ibp,\Ib,\Ob,\Obp$ were chosen arbitrarily, and can be fixed by a choice of either $g$ or $T^{\B|\Y}$, as discussed above, we can conclude the necessity of the condition that $V$ is independent of $\Ia$. We thus have that a necessary and sufficient condition for $P_{\B|\X\Y}$ to be independent of $\X$ is that
\begin{equation}
    \sum_{\Oa} W \ \ \  \text{is independent of $\Ia$}.
\end{equation}
Equivalently, we may state that $P_{\B|\X\Y}$ is independent of $\X$ iff
\begin{equation}
    _{\Oa} W = \, _{\Oa\Ia} W.
\end{equation}

The proof of the condition $_{\Ob} W = \, _{\Ob\Ib} W$ is analogous, i.e., achieved by interchanging the roles of Alice and Bob.
\end{proof}

Boxworld processes are process tensors that satisfy the principle of no signaling without system exchange, and hence also satisfy nonsignaling preservation, a strictly weaker condition. Combining the conditions in Eqs.~\eqref{eq::W_positivity_app}--\eqref{eq::W_proj3_app} that define a valid process tensor with the condition of NSWSE in Eqs.~\eqref{eq::W_NSWSE_A_app}--\eqref{eq::W_NSWSE_A_app}, one gets that a tensor $W$ is a valid boxworld process iff
\begin{align}
    W &\geq 0  \label{eq::bwW_positivity_app} \\
    r(W) &= d_{\Ia}d_{\Oap}d_{\Ib}d_{\Obp} \label{eq::bwW_totalred_app} \\
    _{\mathbb{A}}W &= \, _{\mathbb{A}\Obp}W \label{eq::bwW_proj1_app} \\
    _{\mathbb{B}}W &= \, _{\Oap\mathbb{B}}W \label{eq::bwW_proj2_app} \\
    W &= \, _{\Oap}W + \, _{\Obp}W - \, _{\Oap\Obp}W \label{eq::bwW_proj3_app} \\
    _{\Oa}W &= \, _{\Ia\Oa}W \label{eq::bwW_proj4_app} \\
    _{\Ob}W &= \, _{\Ib\Ob}W, \label{eq::bwW_proj5_app}
\end{align}
as presented in Sec.~\ref{sec::methods} of the main text.

%%%%%%%%%%%%%%
\subsection{Boxworld processes and causal order}\label{app::boxworldprocesses}

In this section, we define subsets of boxworld processes with different properties that relate to their causal order. We recall the definition presented in the main text of causally ordered boxworld processes, and proceed to prove their characterization theorem. Equipped with the definition (and characterization) of causally ordered boxworld processes, we prove an alternative characterization of boxworld processes as elementwise nonnegative tensors that can be expressed as an affine combination of causally ordered boxworld processes. 
\\

\textbf{Characterization of causally ordered boxworld processes.} A bipartite boxworld process $W^{A\prec B}$ that is causally ordered from Alice to Bob is the most general boxworld process that satisfies 
\begin{equation}\label{eq::Wcausalorder}
     W^{A\prec B} * (T^{\A|\X} \otimes T^{\B|\Y} ) =  P^{A\prec B}_{\A\B|\X\Y},
\end{equation}
for all sets of local operations $T^{\A|\X}$ and $T^{\B|\Y}$, where $P^{A\prec B}_{\A\B|\X\Y}$ is a set of probability distributions that is nonsignaling from Bob to Alice (satisfies Eq.~\eqref{eq::probNS_AprecB}). A boxworld process $W^{B\prec A}$ that is causally ordered from Bob to Alice is defined in an analogous manner.

\begin{theorem}[Characterization of bipartite causally ordered boxworld processes]\label{thm::causordproc}
A bipartite boxworld process $W^{A\prec B}$ satisfies Eq.~\eqref{eq::Wcausalorder}, in the order from Alice to Bob, iff it satisfies
\begin{equation}\label{eq::W_AprecB}
    W^{A\prec B} = \, _{\Obp}W^{A\prec B},
\end{equation}
and analogously for $W^{B\prec A}$.
\end{theorem}

\begin{proof}
The question here is: given any sets of deterministic local transformations $T^{\Y}$ acting on Bob's side, what is the most general boxworld process $W^{A\prec B}$ such that, 
\begin{equation}
    \sum_{\B}P^{A\prec B}_{\A\B|\X\Y} = W^{A\prec B} * T^{\A|\X} * T^{\Y}  = P_{\A|\X\Y}
\end{equation}
is independent of $\Y$ for all general operations $T^{\A|\X}$ acting on Alice's side. 

Substituting in the parametrization of $T^{\Y}$ from Eq.~\eqref{eq::op_parametrization} as a function of the tensor parameter $\Ycal$, one has
\begin{align}
    \sum_{\B}P^{A\prec B}_{AB|XY} &= T^{\A|\X} * \sum_{\Ibp\Ib\Ob\Obp} T^{\Y} \, W^{A\prec B} \\
    &= T^{\A|\X} * \sum_{\Ibp\Ib\Ob\Obp} \left[\frac{\id^{\mathbb{B}}}{d_{\Ib}d_{\Obp}} + \, _{(1 -\Obp +\Ob\Obp - \Ib\Ob\Obp)}\Ycal\right] W^{A\prec B} \\
    &= T^{\A|\X} * \frac{\rd_{\mathbb{B}} (W^{A\prec B})}{d_{\Ib}d_{\Obp}} + 
    T^{\A|\X} * \sum_{\Ibp\Ib\Ob\Obp} \left[ _{(1 -\Obp +\Ob\Obp - \Ib\Ob\Obp)}\Ycal \ W^{A\prec B} \right] \\
    &= T^{\A|\X} *  \frac{\rd_{\mathbb{B}} (W^{A\prec B})}{d_{\Ib}d_{\Obp}} +  
    T^{\A|\X} * \sum_{\Ibp\Ib\Ob\Obp} \left[ \Ycal \ _{(1 -\Obp +\Ob\Obp - \Ib\Ob\Obp)}W^{A\prec B} \right], \label{eq::proof2}
\end{align}
where in the last line, the self-adjointness of the reduce-and-replace map was applied.

The dependence of Eq.~\eqref{eq::proof2} on $\Y$ appears only in the second term of the sum, hence, Eq.~\eqref{eq::proof2} will be independent of $\Y$ iff
\begin{equation}
    _{(1 -\Obp +\Ob\Obp - \Ib\Ob\Obp)}W^{A\prec B} = 0, 
\end{equation}
or, equivalently,
\begin{equation}
    W^{A\prec B} =\, _{(\Obp -\Ob\Obp + \Ib\Ob\Obp)}W^{A\prec B}.
\end{equation}

Since we require that, on top of satisfying the expression above concerning causal order, the tensor $W^{A\prec B}$ must also be a valid box world tensor, we may further simplify this expression by applying the condition of NSWSE from Eq.~\eqref{eq::W_NSWSE_B_app}, arriving at
\begin{equation}
    W^{A\prec B} =\, _{\Obp}W^{A\prec B}.
\end{equation}

The proof of Eq.~\eqref{eq::W_AprecB} is analogous.
\end{proof}

Combining the condition of causal order in Eq.~\eqref{eq::W_AprecB} with the conditions that define a valid boxworld process in Eqs.~\eqref{eq::bwW_positivity_app}--\eqref{eq::bwW_proj5_app}, we have that a tensor $W^{A\prec B}$ is a boxworld process that is causally ordered from Alice to Bob iff it satisfies 
\begin{align}
    W^{A\prec B} &\geq 0 \label{eq::W_AprecB_positivity}  \\
    r(W^{A\prec B}) &= d_{\Ia}d_{\Oap}d_{\Ib}d_{\Obp}  \label{eq::W_AprecB_totalred} \\
    _{\mathbb{B}}W^{A\prec B} &=\, _{\Oap\mathbb{B}}W^{A\prec B} \label{eq::W_AprecB_proj1} \\
    W^{A\prec B} &=\, _{\Obp}W^{A\prec B} \label{eq::W_AprecB_proj2} \\
    _{\Oa}W^{A\prec B} &=\, _{\Ia\Oa}W^{A\prec B} \label{eq::W_AprecB_proj3} \\
    _{\Ob}W^{A\prec B} &=\, _{\Ib\Ob}W^{A\prec B}, \label{eq::W_AprecB_proj4}
\end{align}
and equivalently for $W^{B\prec A}$.

A bipartite nonsignaling boxworld process $W^{A||B}$ is one is that satisfies the constraints of both causally ordered from Alice to Bob and from Bob to Alice, meaning they are the larger set of boxworld processes that can only lead to nonsignaling correlations. Hence, a tensor $W^{A||B}$ is a nonsignaling boxworld process iff it satisfies
\begin{align}
    W^{A||B} &\geq 0  \\
    r(W^{A||B}) &= d_{\Ia}d_{\Oap}d_{\Ib}d_{\Obp}  \\
    W^{A||B} &=\, _{\Oap\Obp}W^{A||B} \label{eq::W_nonsig}  \\
    _{\Oa}W^{A||B} &=\, _{\Ia\Oa}W^{A||B}  \\
    _{\Ob}W^{A||B} &=\, _{\Ib\Ob}W^{A||B}. 
\end{align}

A bipartite causally separable process $W^\text{sep}$ is a boxworld process that can be expressed as a convex combination of processes that are causally ordered from Alice to Bob and from Bob to Alice, i.e.,
\begin{equation}
    W^\text{sep} = q\,W^{A\prec B} + (1-q)W^{B\prec A}, \ \ \ \text{for some } q\in[0,1],
\end{equation}
and for some causally ordered boxworld processes $W^{A\prec B}$ and $W^{B\prec A}$. This definition is made in analogy to the definition of a bipartite causally separable process matrix~\cite{araujo2015witnessing}. Boxworld processes that are not causally separable are considered to display an indefinite causal order and may violate causal inequalities, as shown in detail in App.~\ref{app::inequalities}.

Notice that, from the conditions $W^{A\prec B} =\, _{\Obp}W^{A\prec B}$ (Eq.~\eqref{eq::W_AprecB_proj2}) for a boxworld process that is causally ordered from Alice to Bob, $W^{B\prec A} =\, _{\Oap}W^{B\prec A}$ for a boxworld process that is causally ordered from Bob to Alice, and
$W^{A||B} =\, _{\Oap\Obp}W^{A||B}$ (Eq.~\eqref{eq::W_nonsig}) for a nonsignaling boxworld process, one may rewrite the condition $W = \, _{\Oap}W + \, _{\Obp}W - \, _{\Oap\Obp}W$ in Eq.~\eqref{eq::bwW_proj3_app} for a valid boxworld process as
\begin{equation}\label{eq::bwW_proj3_app_rewritten}
    W = \, W^{A\prec B} + W^{B\prec A} - W^{A||B}. 
\end{equation}
\\

\textbf{Alternative characterization of boxworld processes.}
Equipped with the notion of causally ordered boxworld processes, we can now prove a second characterization theorem for boxworld processes, in terms of an affine combination of causally ordered boxworld processes. In order to do that, we will require the following lemma.

\begin{lemma}\label{lemma:positivity}
Let $P$ be an elementwise positive tensor and let $_{\alpha}P$ be the reduce-and-replace operation over the index $\alpha$ of the tensor $P$. Then,
\begin{equation}\label{eq:pos_red_rep}
    d_\alpha \, _{\alpha}P - P \geq 0,
\end{equation}
where $d_\alpha=|\alpha|$ and $\geq 0$ indicates elementwise nonnegativity.
\end{lemma}

\begin{proof}
The proof follows straightforwardly from the positivity of $P$. The reduce-and-replace substitutes the elements along one index of the tensor with their average (their sum divided by their cardinality $d$). By then multiplying the resulting tensor by $d$, the overall operation amounts to substituting each element along one index with the sum of all of them, which is greater than each individual element since they are all nonnegative. The elementwise positivity of Eq.~\eqref{eq:pos_red_rep} is then apparent. 
\end{proof}

We are now ready to prove the following theorem.

\begin{theorem}[Alternative characterization of boxworld processes]\label{thm:affinemodproc}
An elementwise nonnegative tensor $W\geq0$ is a boxworld process iff it can be expressed as an affine combination of causally ordered boxworld processes. That is, iff
\begin{equation}
    W = \lambda W^{A\prec B} + (1-\lambda) W^{B\prec A}, \ \ \ \text{for some } \lambda\in\mathbb{R},
\end{equation}
and for some causally ordered boxworld processes $W^{A\prec B}$ and $W^{B\prec A}$.
\end{theorem}

\begin{proof}
We start by showing the if part, i.e., by checking that any $W\geq0$ that can be written as $\lambda W^{A\prec B} + (1 - \lambda) W^{B\prec A}$ is a valid boxworld process. This can be done by direct insertion into the conditions defined by Eqs.~\eqref{eq::bwW_totalred_app}--\eqref{eq::bwW_proj5_app}. Since $W^{A\prec B}$ and $W^{B\prec A}$ are valid boxworld processes themselves, and therefore satisfy Eqs.~\eqref{eq::bwW_totalred_app}--\eqref{eq::bwW_proj5_app} individually, these conditions are also satisfied by an affine combination of them.

Now we show the only if part, i.e., that for every boxworld process $W$, there exists a $\lambda$, $W^{A\prec B}$, and $W^{B\prec A}$ such that $W = \lambda W^{A\prec B} + (1 - \lambda) W^{B\prec A}$. For that, take the ansatz
\begin{align}
    \lambda &\coloneqq d \\
    W^{A\prec B} &\coloneqq \frac{(d-1)\, _{\Oap\Obp}W+\, _{\Obp}W}{d} \\
     W^{B\prec A} &\coloneqq \frac{(d) \ _{\Oap\Obp}W-\, _{\Oap}W}{d-1}, 
\end{align}
where $d=|\Obp|>1$.

First, we show that $W^{A\prec B}$ and $W^{B\prec A}$ are valid causally ordered boxworld processes. 

The first step is to show that they are valid sets of probability distributions. Since $W$ is a valid set of probability distributions, the reduce-and-replace operation $_{X}\cdot$ preserves sets of probability distributions, and $d>1$, $W^{A\prec B}$ is a sum of nonnegative numbers, and is therefore elementwise nonnegative. For the nonnegativity of $W^{B\prec A}$, we invoke Lemma~\ref{lemma:positivity}.
One can also check that $r(W^{A\prec B})=r(W^{B\prec A})=r(W)$, which, since $W$ is a valid set of probability distributions, is equal to $d_{\Ia}d_{\Oap}d_{\Ib}d_{\Obp}$.

Now we check that $W^{A\prec B}$ satisfies Eqs.~\eqref{eq::W_AprecB_proj1}--\eqref{eq::W_AprecB_proj4}. Starting with Eq.~\eqref{eq::W_AprecB_proj2}, we have that
\begin{align}
    _{\Obp}W^{A\prec B} = \frac{(d-1)\, _{\Oap\Obp}W+\, _{\Obp}W}{d} = W^{A\prec B}.
\end{align}
Now for Eq.~\eqref{eq::W_AprecB_proj1}, we have that
\begin{align}
    _{\mathbb{B}}W^{A\prec B} 
    &= \frac{(d-1)\, _{\Oap\mathbb{B}}W+\, _{\mathbb{B}}W}{d} \label{eq::A1} \\
    &= \frac{(d-1)\, _{\Oap\mathbb{B}}W+\, _{\Oap\mathbb{B}}W}{d} \label{eq::A2} \\
    &= \, _{\Oap\mathbb{B}}W^{A\prec B}, 
\end{align}
where from Eq.~\eqref{eq::A1} to Eq.~\eqref{eq::A2} we used the validity of $W$, i.e., $_{\mathbb{B}}W= \, _{\Oap\mathbb{B}}W$. For Eq.~\eqref{eq::W_AprecB_proj3}, we have that
\begin{align}
    _{\Oa}W^{A\prec B} 
    &= \frac{(d-1)\, _{\Oa\Oap\Obp}W+\, _{\Oa\Obp}W}{d} \label{eq::A3} \\
    &= \frac{(d-1)\, _{\Ia\Oa\Oap\Obp}W+\, _{\Ia\Oa\Obp}W}{d} \label{eq::A4} \\
    &= \, _{\Ia\Oa}W^{A\prec B}
\end{align}
where from Eq.~\eqref{eq::A3} to Eq.~\eqref{eq::A4} we again used the validity of $W$, i.e., $_{\Oa}W= \, _{\Ia\Oa}W$. The equivalent holds for Eq.~\eqref{eq::W_AprecB_proj4}.

Next, we check that $W^{B\prec A}$ satisfies the analog of Eqs.~\eqref{eq::W_AprecB_proj1}--\eqref{eq::W_AprecB_proj4}, with a relabeling of Alice and Bob.
Starting with the relabeled Eq.~\eqref{eq::W_AprecB_proj2}, we have that
\begin{align}
    _{\Oap}W^{B\prec A} = \frac{(d)\ _{\Oap\Obp}W-\, _{\Oap}W}{d-1} = W^{A\prec B}.
\end{align}
Now for the relabeled Eq.~\eqref{eq::W_AprecB_proj1}, we have that
\begin{align}
    _{\mathbb{A}}W^{B\prec A} 
    &= \frac{(d)\ _{\mathbb{A}\Obp}W-\, _{\mathbb{A}}W}{d-1} \label{eq::B1} \\
    &= \frac{(d)\ _{\mathbb{A}\Obp}W-\, _{\mathbb{A}\Obp}W}{d-1} \label{eq::B2} \\
    &= \, _{\mathbb{A}\Obp}W^{B\prec A},
\end{align}
where from Eq.~\eqref{eq::B1} to Eq.~\eqref{eq::B2} we used the validity of $W$, i.e., that $_{\mathbb{A}}W= \, _{\mathbb{A}\Obp}W$. For Eq.~\eqref{eq::W_AprecB_proj3}, we have that 
\begin{align}
    _{\Oa}W^{B\prec A} 
    &= \frac{(d)\ _{\Oa\Oap\Obp}W-\, _{\Oa\Oap}W}{d-1} \label{eq::B3} \\
    &= \frac{(d)\ _{\Ia\Oa\Oap\Obp}W-\, _{\Ia\Oa\Oap}W}{d-1} \label{eq::B4} \\
    &= \, _{\Ia\Oa}W^{A\prec B}
\end{align}
where from Eq.~\eqref{eq::B3} to Eq.~\eqref{eq::B4} we again used the validity of $W$, i.e., $_{\Oa}W= \, _{\Ia\Oa}W$. The equivalent holds for Eq.~\eqref{eq::W_AprecB_proj4}.

Finally, we only need to show that
\begin{align}
     \lambda W^{A\prec B} + (1 - \lambda) W^{B\prec A}
     &= d\,\frac{(d-1)\, _{\Oap\Obp}W+\, _{\Obp}W}{d} + (1-d)\,\frac{(d)\ _{\Oap\Obp}W-\, _{\Oap}W}{d-1} \\
     &= (d-1)\, _{\Oap\Obp}W +\, _{\Obp}W - (d)\ _{\Oap\Obp}W +\, _{\Oap}W \\
     &= \, _{\Oap}W + \, _{\Obp}W - \, _{\Oap\Obp}W \\
     &= W,
\end{align}
which concludes the proof.
\end{proof}

%%%%%%%%%%%%%%%%%%%%%%%%%%%%%%
\subsection{Symmetry transformations of boxworld processes}\label{sec::Sym}
%%%%%%%%%%%%%%%%%%%%%%%%%%%%%%

In this section, we derive the symmetry transformation for boxworld processes, namely, the most general map that maps a valid $\W$ into a valid $\W$. The projector on valid boxworld processes can be written as a composition
\begin{equation}\label{eq:L_proj}
    L = L_A \circ L_B \circ L_{AB} \circ L_{I_A} \circ L_{I_B},
\end{equation}
where
\begin{equation}\label{eq:proj_list}
\begin{split}
    L_A(X) &:=  _{(1- \mathbb{A}+\mathbb{A}\Obp)}X, \quad L_B(X) :=  _{(1- \mathbb{B}+\mathbb{B}\Oap)}X,\\
    L_{I_A}(X) &:= _{(1- \Oa+\Ia\Oa)}X,\quad L_{I_B}(X) := _{(1- \Ob+\Ib\Ob)}X\\
    L_{AB}(X) &:=  _{\Oap}X + _{\Obp}X - _{\Oap\Obp}X,
\end{split}
\end{equation}

We want to characterize the most general local map $R_{\A\td{\A}}$ that maps valid boxworld processes into valid boxworld processes. According to the theory of Ref.~\cite{milz2024characterising} the transformation $S$ of $W$ obeys 
\begin{equation}\label{eq:L_T_comp}
    L \circ S \circ \td{L} = S \circ \td{L},
\end{equation}
where we defined $\td{L}$ as the equivalent of Eq.~\eqref{eq:L_proj} for the spaces $\td{\A},\td{\B}$, together with the normalization constraint
\begin{equation}\label{eq:norm_R}
    _{\mathbb{A}\mathbb{B}\td{\mathbb{A}}\td{\mathbb{B}}} (S) = d_{\td{A}\td{B}} \ \openone_{\A\td{\A}\B\td{\B}}.
\end{equation}
We later specialize this condition to the special case of a local transformation on Alice's side, namely,
$S= R_{\A\td{\A}}\otimes \openone_{\B\td{\B}}$.

With Eq.~\eqref{eq:L'_proj}, we can define the projector on the space of valid transformations as
\begin{equation}\label{eq:symm_cond}
    P[S] = S -  S \circ \td{L} + L \circ S \circ \td{L},
\end{equation}
one can easily verify that $P^2=P$. Notice that $L$ and $\td{L}$ commute, as they act on different indices and the reduce-and-replace operation is self-adjoint. As a consequence, applying $\td{L}$ to $W$ or to the transformation $S$ is equivalent, namely,
\begin{equation}
    S \circ \td{L} (W)= S(\td{L} (W))= S* (\td{L} * W) =  (S* \td{L}) * W 
\end{equation}
The idea is that the we reduce-and-replace the index of $S$ that act on the space of $W$.
Then, any transformation that sends valid $W$s to valid $W$s can be thought of as applied to the local operation $T$, hence transforming it into a new one, namely,
\begin{equation}\label{eq:symm_dual}
    S(W) * T_A * T_B = W * S^\dagger( T_A * T_B) = W * T'_A * T'_B.
\end{equation}
Intuitively, if two pairs of local operations $( T_A , T_B)$ and $( T'_A , T'_B)$, are connected by a symmetry transformation $S^\dagger$, which satisfies Eq.~\eqref{eq:symm_cond}, as in Eq.~\eqref{eq:symm_dual}, the optimal value of a certain expression (e.g., GYNI) over all $W$s is the same on  $( T_A , T_B)$ and $( T'_A , T'_B)$. We formalize better this aspect at the end of the section.

The expression in Eq.~\eqref{eq:symm_cond} can be further simplified considering only local operations on each party. This is no loss of generality, as local symmetry transformation can be composed with each other. 
We want to map the spaces $\td{\A},\td{\B}$ into the spaces $\A,\B$, via a local map which is trivial on the $\td{B}$ space. Since we are interested in symmetry transformations, $\td{\A},\td{\B}$ are just a copy of the spaces $\A,\B$. We write
\begin{equation}
    (R_{\A\td{\A}}\otimes \openone_{\B\td{\B}}) * \W_{\td{\A}\td{\B}} = \W'_{\A\B},
\end{equation}
where $\openone_{\B\td{\B}}$ is the identity map
\begin{equation}
    \openone_{\B\td{\B}} := [\tIbp=\Ibp][\tIb=\Ib][\tOb=\Ob][\tObp=\Obp].
\end{equation}
For instance, when considering local operations on Alice's side, we can use the fact that the local operation on $\td{B}$ is trivial, which can be written as $(R_{\A\td{\A}}\otimes \openone_{\B\td{\B}})=_{\td{\mathbb{B}}\mathbb{B}} (R_{\A\td{\A}}\otimes \openone_{\B\td{\B}})$. This implies that the projector $L$ in Eq.~\eqref{eq:symm_cond} can be substituted with 
\begin{equation}\label{eq:L'_proj}
    L'(X):= _{\Oap(1-\Oa+\Ia\Oa)}(X),
\end{equation}
and analogously for $\td{L}$. To prove this, it is enough to show that $L(X)=L'(X)$ whenever $X=_{\mathbb{B}}X$, e.g., by direct substitution in Eq.~\eqref{eq:proj_list} and then in Eq.~\eqref{eq:L_proj}. The same argument applies to $\td{L}$.

Finally, we can write the condition $P[S]=S$ as
\begin{align}\label{eq:R_pr_null}
    0 = _{\tOap(1-\tOa+\tIa\tOa)}(R_{\A\td{\A}}\otimes \openone_{\B\td{\B}}) - _{\Oap(1-\Oa+\Ia\Oa)\tOap(1-\tOa+\tIa\tOa)} (R_{\A\td{\A}}\otimes \openone_{\B\td{\B}}) 
\end{align}
Equation \eqref{eq:R_pr_null} together with the normalization condition $r[R_{\A\td{A}}]=d_{\A}$, see Eq.~\eqref{eq:norm_R}, provide a complete characterization of the symmetry transformations for $\W$. 

We notice that this implies that for any local operation $T_{\td{\A}}$, we have $R_{\A\td{A}} * T_{\td{\A}}=T'_{\A}$, where $T'_{\A}$ is a valid operation on the space $\A$. This transformation induces a pre-order, i.e., a reflexive and transitive homogeneous binary relation, given by $T\succ T'$ if $T'= R*T$. This idea can be exploited for the numerical optimizations discussed in the previous section. Indeed, instead of running the LP in Eq.~\eqref{eq:W_LP_T_fix} for all pairs of local operations $T^{\A|\X}$ and $T^{\B|\Y}$, one can use the local symmetries to fix one local operation for each party to be one of the extremal ones according to the pre-order $T\succ T'$ defined above. This would be the equivalent of using the unitary invariance of the set of quantum states to fix one local measurement per party in a Bell experiment to be in the computational basis. We leave this optimization for future research.

%%%%%%%%%%%%%%%%%%%%%%%%%%%%%%%%%%%%%%%%%%%%
\section{Boxworld correlations}\label{app::part2}
%%%%%%%%%%%%%%%%%%%%%%%%%%%%%%%%%%%%%%%%%%%%

%%%%%%%%%%%%%%
\subsection{No perfect two-way signaling in boxworld correlations}\label{app::no2waysig}

In this section we prove the result stated in the main text that boxworld correlations do not admit perfect two-way signaling. We begin by stating the formal definition of boxworld correlations.

A set of probability distributions $\overline{P}_{\A\B|\X\Y}$ are called boxworld correlations iff there exists a set of local operations $T^{\A|\X}$ and $T^{\B|\Y}$, and a boxworld process $W$, such that
\begin{equation}
    \overline{P}_{\A\B|\X\Y} = (T^{\A|\X} \otimes T^{\B|\Y}) * W.
\end{equation}

\begin{theorem}[Boxworld correlations do not admit perfect two-way signaling]\label{thm::no2waysig}
A set of perfect two-way signaling probability distributions $P^\text{2-sig}_{\A\B|\X\Y}=\delta_{\A,\Y}\,\delta_{\B,\X}$ cannot be achieved by performing local transformations on a boxworld process. That is, there does not exist any triplet $\{T^{\A|\X},T^{\B|\Y},W\}$, for any finite cardinality of their random variables, such that $(T^{\A|\X} \otimes T^{\B|\Y}) * W=\delta_{\A,\Y}\,\delta_{\B,\X}.$

Moreover, any sets of probability distributions $P_{\A\B|\X\Y}$ such that 
\begin{equation}
    \frac{1}{4} \sum_{\A,\B,\X,\Y} \delta_{\A,\Y}\,\delta_{\B,\X} P_{\A\B|\X\Y} > 1 - \frac{1}{2d},
\end{equation}
cannot be obtained by local transformations and a boxworld process of cardinality $d\coloneqq\min\{d_{\Oap},d_{\Obp}\}=\min\{|\Oap|,|\Obp|\}$.
\end{theorem}

\begin{proof}
As shown in Eq.~\eqref{eq::bwW_proj3_app_rewritten}, each valid boxworld process $W$ can be decomposed as
\begin{equation}
    W= W^{A\prec B} + W^{B\prec A}- W^{A||B}.
\end{equation}
Consequently, when contracting $W$ with any set of local transformations, the resulting correlations will satisfy the following:
\begin{equation}\label{eq::proof3}
    \overline{P}_{\A\B|\X\Y} = P_{\A\B|\X\Y}^{A\prec B} + P_{\A\B|\X\Y}^{B\prec A} - P_{\A\B|\X\Y}^\text{NS}.
\end{equation}

Now take the function $\text{GYNI}(P_{\A\B|\X\Y})\coloneqq\frac{1}{4}\sum_{\A,\B,\X,\Y} \delta_{\A,\Y}\,\delta_{\B,\X} P_{\A\B|\X\Y}$, which returns the score of any given set of probability distributions $P_{\A\B|\X\Y}$ in ``guess your neighbor's input'' (GYNI) game (see Sec.~\ref{app::inequalities} for more details). Evaluating this function on both sides of Eq.~\eqref{eq::proof3}, one gets a valid expression for the GYNI score of any correlations that comes from a boxworld process. Namely,
\begin{equation} 
    \text{GYNI}(\overline{P}_{\A\B|\X\Y}) = \text{GYNI}(P_{\A\B|\X\Y}^{A\prec B}) + \text{GYNI}(P_{\A\B|\X\Y}^{B\prec A}) - \text{GYNI}(P_{\A\B|\X\Y}^\text{NS}).
\end{equation}

From the bounds of the GYNI inequality (again see Sec.~\ref{app::inequalities} for more details), it is known that a one-way nonsignaling set of probability distributions such as $P_{\A\B|\X\Y}^{A\prec B})$ or $P_{\A\B|\X\Y}^{B\prec A})$ can never achieve a higher score than $\frac{1}{2}$. 
Now, since $W^{B\prec A}$ and $W^{B||A}$ are related by a single partial reduction operation, i.e., $W^{B||A} = \, _{\Obp}W^{B\prec A}$, and since, from Lemma~\ref{lemma:positivity} we have that $d_{\Obp} \, (_{\Obp}W^{B\prec A}) - W^{B\prec A} = d_{\Obp} \, W^{B||A} - W^{B\prec A} \geq 0$, it follows from linearity that 
\begin{equation}
    \text{GYNI}(P_{\A\B|\X\Y}^\text{NS})\geq \frac{1}{d_{\Obp}} \,\text{GYNI}(P_{\A\B|\X\Y}^{B\prec A}),
\end{equation} 
However, notice that the same argument follows if the starting point of $W^{B||A} = \, _{\Oap}W^{A\prec B}$, which would yield 
\begin{equation}
    \text{GYNI}(P_{\A\B|\X\Y}^\text{NS})\geq \frac{1}{d_{\Oap}} \,\text{GYNI}(P_{\A\B|\X\Y}^{A\prec B}).
\end{equation} 
Hence, one may pick the smallest value $d\coloneqq\min\{d_{\Oap},d_{\Obp}\}$, and putting all these bounds together, one gets that
\begin{align}
    \text{GYNI}(\overline{P}_{\A\B|\X\Y}) &\leq \frac{1}{2} + \frac{1}{2} - \frac{1}{d}\frac{1}{2} \\
    \implies \frac{1}{4}\sum_{\A,\B,\X,\Y} \delta_{\A,\Y}\,\delta_{\B,\X} \overline{P}_{\A\B|\X\Y} &\leq 1 - \frac{1}{2d}
\end{align}
is satisfied by all boxworld correlations $\overline{P}_{\A\B|\X\Y}$.

A consequence of this is that no boxworld process with finite $d$ can achieve the score of $1$ in the GYNI game, which implies that it cannot achieve perfect two-way signaling. 
\end{proof}

%%%%%%%%%%%%%%
\subsection{The set of boxworld correlations is a polytope}\label{app::polytope}

In this section, we prove that the set of all $\overline{P}_{\A\B|\X\Y}$ is a polytope. In order to do so, we start by realizing that the set of local operations and the set of boxworld processes individually form polytopes themselves.

\begin{fact}
    The set of local transformations $T^{\A|\X}$ [tensors that satisfy Eq.~\eqref{eq::op_probabilistic}] is a polytope.
\end{fact}

\begin{fact}
    The set of boxworld processes $W$ [tensors that satisfy Eqs.~\eqref{eq::bwW_positivity_app}--\eqref{eq::bwW_proj5_app}] is a polytope.
\end{fact}

The first fact is straightforward to see from the definition of local operations. Since all operations can be expressed as a convex combination of deterministic probability distributions, and deterministic probability distributions are extremal points in the set of all probability distributions, the set of local operations is a polytope formed by the convex hull of these deterministic probability distributions. This fact can also be checked from the point of view of the facets of this polytope:
Take, for simplicity, a two-outcome local operation $T^{\A|\X}=\{T^{\A|\X}(0|x),T^{\A|\X}(1|x)\}$. The general case is analogous. From its definition it follows that 
\begin{align}
    T^{\X}(x) &:= T^{\A|\X}(0|x) + T^{\A|\X}(1|x),\\
    T^{\A|\X}(0|x) &\geq 0,\\
    T^{\A|\X}(1|x) &= T^{\X}(x) - T^{\A|\X}(0|x) \geq 0,\\
    {}_{\Ia\Oap} T^{\X}(x) &= \frac{1}{(d_{\Aoo} d_{\Aot} )} \id, \label{eq:boundpoly}\\
    {}_{\Oap(1-\Oa)} T^{\X}(x) &= 0,
\end{align}
for all $x$. These equations consist of elementwise positivity and linear constraints, and the set is bounded as Eq.~\eqref{eq:boundpoly} fixes the total normalization of the tensor. Consequently, we conclude that this set is a polytope.

As for the second fact, it follows from the trivial observation that the conditions in Eqs.~\eqref{eq::bwW_positivity_app}--\eqref{eq::bwW_proj5_app} form a polytope: we have again only linear constraints, elementwise positivity, and the total normalization constraint in Eq.~\eqref{eq::bwW_totalred_app} guarantees that the set is bounded.

Furthermore, it is interesting to remark that both sets of tensors, the one of local operations and of boxworld processes, form a polytope for any fixed ``dimension'', that is, for any fixed value for the cardinality of their random variables.

To then prove that the set of boxworld correlations $\overline{P}_{\A\B|\X\Y}$ is also a polytope, we use the following lemma.

\begin{lemma}\label{lemma:polytope}
Given two polytopes $\PP_1\in \mathbb{R}^{n_1}, \PP_2\in \mathbb{R}^{n_2}$ the following sets
\begin{align}
\PP^\times := \{ (v,w)\in \mathbb{R}^{n_1+n_2} | v\in \PP_1, w\in \PP_2\},\\
\PP^\otimes := \{ v\otimes w \in  \mathbb{R}^{n_1 n_2}| v\in \PP_1, w\in \PP_2\},
\end{align}
are polytopes.
\end{lemma}

\begin{proof} 
By assumption, each $v \in \PP_1$ can be decomposed as $v=\sum_{i=1}^{k_1} \lambda_i v_i$, with $\lambda_i \geq 0, \sum_{i=1}^{k_1}\lambda_i = 1$, where $\{ v_1, \ldots, v_{k_1}\}$ are the extremal points of $\PP_1$. At the same time, it can expressed as $A_1 v\leq b_1$, where the inequality is intended component-wise, for some $m_1\times n_1$ matrix and some vector $b_1\in \mathbb{R}^{m_1}$.  Similarly for $w \in \PP_2$, we write $w=\sum_{i=1}^{k_2} \mu_i w_i$ and $A_2 w \leq b_2$. 

For the set $\PP^\times$, we can write
\begin{equation}
\PP^\times = \{ (v,w) \in \mathbb{R}^{n_1+n_2} | A_1 v\leq b_1, \ A_2 w \leq b_2 \} = \{ z \in \mathbb{R}^{n_1+n_2}| A z \leq b \},
\end{equation}
where we defined $A$ as the block-diagonal matrix $A := A_1 \oplus A_2$ and $b:= (b_1,b_2)$. It is clear that the set is again a polyhedral set, 
and since $v$ and $w$ are bounded, also $(v,w)$ is bounded. Hence, $\PP^\times$ is a polytope. From this characterization, one sees that 
$\PP^\times$ is generated by a set of points of the form $(v_i, w_j)$, where  $\{v_i\}_i$ and $\{w_j\}$ are the extremal points of the original 
polytopes $\PP_1$ and $\PP_2$. This can be shown as follows. Each extremal point of the polytope is defined by a unique set of linear equalities 
that have a unique solution. In other words, each vertex $v_i$ is associated with a submatrix $A_1^{(i)}$ of $A_1$, selected by taking only some of
the rows, and the corresponding vector $b_1^{(i)}$ such that $A_1^{(i)}v_i=b_1^{(i)}$ and $\dim{\rm Ker}(A_1^{(i)})=0$. There is only a finite number of ways of choosing the submatrix $A_1^{(i)}$ such that it satisfies the above conditions, and each way corresponds to a vertex $v_i$. The same argument 
applies to $A_2$ and $\{w_j\}_j$. As a consequences, since the spaces for $\PP_1$ and $\PP_2$ are in direct sum, the only way of choosing a 
submatrix $A^{(ij)}$ of $A$ that satisfy the analogous constraints as above (unique solution), and thus define a vertex, is to combine the matrices
$A_1^{(i)}$ and $A_2^{(j)}$. This corresponds to a vertex of the form $(v_i, w_j)$.

To prove that $\PP^\otimes$ is a polytope, we notice that a generic element $v\otimes w \in \PP^{\otimes}$ can be written as
\begin{equation}
v \otimes w = \left(  \sum_{i=1}^{k_1} \lambda_i v_i \right) \otimes \left( \sum_{j=1}^{k_2} \mu_j w_j \right)= \sum_{ij} \lambda_i \mu_j v_i \otimes w_i.
\end{equation}
The coefficients $\gamma_{ij}:=\lambda_i \mu_j$ satisfy $\gamma_{ij}\geq 0$ for all $i,j$, and $\sum_{ij} \gamma_{ij}=1$. We can, then, interpret $\gamma_{ij}$ as the coefficient of the convex combination in terms of the points $v_i\otimes w_j$, which are, thus, the extremal points of the polytope $\PP^\otimes$. 
\end{proof}

We are now equipped to prove the main result.

\begin{theorem}\label{thm::polytope}
The set of boxworld correlations $\overline{P}_{\A\B|\X\Y}$ is a polytope.
\end{theorem}

\begin{proof}
To keep the notation lighter, we will prove the case of two inputs and two outputs, i.e., $\A=\B=\X=\Y=\{0,1\}$. The general case can be proven similarly. Moreover, the same argument can be extended to the multipartite case.

First, let us notice that $\overline{P}_{\A\B|\X\Y}$ can be written as
\begin{equation}
\overline{P}_{\A\B|\X\Y} = (T^{\A|\X} \otimes T^{\B|\Y}) * W = (T^{\A|\X} \otimes T^{\B|\Y} \otimes  W') * \Gamma,
\end{equation}
where $\Gamma$ is a tensor appropriately chosen to perform the right contraction between $T^{\A|\X} \otimes T^{\B|\Y}$ and $W'$, which is a boxworld process identical to $W$ but defined on a ``copy'' of its random variables. The existence of $\Gamma$ is guaranteed by the universal property%
\footnote{The universal property of the tensor product states that every bilinear map can be transformed into a linear map on the tensor product of the input spaces. In this case we have a bilinear map from the pair of $(T^{\A|\X} \otimes T^{\B|\Y} , W)$ into the reals, which can be transformed into a linear map from the tensor product, i.e., $T^{\A|\X} \otimes T^{\B|\Y} \otimes W$ into the reals.}
of the operation of tensor composition, which we defined as a tensor product of tensors in Eq.~\eqref{eq::notation_tensorprod}.

We can construct a vector of tensors by taking their Cartesian product. Namely, $T^{\X}(x) \coloneqq (T^{\A|\X}(0|x), T^{\A|\X}(1|x))$ and $T^{\Y}(y) \coloneqq (T^{\B|\Y}(0|y), T^{\B|\Y}(1|y))$, then $T_{A} \coloneqq (T^{\X}(0), T^{\X}(1))$, and $T_{B} \coloneqq (T^{\Y}(0), T^{\Y}(1))$. 
Since $T^{\X}(x)$ belongs to a polytope for each $x$, and similarly for $T^{\Y}(y)$, we have by Lemma~\ref{lemma:polytope} that $T_{A}$, $T_{B}$ also belong to a polytope and the same holds for $T_{A} \otimes T_{B}$ and $T_{A} \otimes T_{B} \otimes W'$.

To conclude, it is sufficient to extend the operator $\Gamma$ to $\td{\Gamma}$, given by $k=|\A|\cdot|\B|\cdot|\X|\cdot|\Y|$ copies of it, one for each possible value of $(a,b,x,y)$.  We then have 
\begin{equation}\label{eq:p_poly}
    \overline{P}_{\A\B|\X\Y}(a,b|x,y) = (T_{A} \otimes T_{B} \otimes  W')_{a,b,x,y} * \td\Gamma.
\end{equation}
In other words, the action of $ *\, \td \Gamma$ consists in the contraction of the tensor $\Gamma$ for each entry $T^{\A|\X}\otimes T^{\B|\Y} \otimes W'$ of the tensor  $T_{A} \otimes T_{B} \otimes  W'$.
Since the contraction with $\td \Gamma$ constitutes a linear operation, we have that the image of the polytope is again a polytope, which concludes the proof. 
\end{proof}

%%%%%%%%%%%%%%
\subsection{All causal correlations are achieved by boxworld processes}\label{app::causalcorrelations}

Causal correlations, in a bipartite scenario, are those that do not contain any genuine form of two-way signaling. These are the sets of probability distributions that are either: 
\begin{itemize}
    \item Nonsignaling: $P^\text{NS}_{\A\B|\X\Y}$, called nonsignaling correlations, respecting both Eqs.~\eqref{eq::probNS_AprecB} and~\eqref{eq::probNS_BprecA};
    \item One-way signaling: $P^{A\prec B}_{\A\B|\X\Y}$, respecting Eq.~\eqref{eq::probNS_AprecB}, or $P^{B\prec A}_{\A\B|\X\Y}$, respecting Eq.~\eqref{eq::probNS_BprecA}, called causally ordered correlations;
    \item Mixtures of one-way signaling: $P^\text{causal}_{\A\B|\X\Y}\coloneqq q\,P^{A\prec B}_{\A\B|\X\Y}+(1-q)P^{B\prec A}_{\A\B|\X\Y}$, respecting Eq.~\eqref{eq::probcausal}, called causal correlations. 
\end{itemize}

Bipartite correlations that are not causal display genuine two-way signaling. In the process matrix formalism, all causal correlations can be achieved by causally separable process matrices~\cite{oreshkov2012quantum,bavaresco2019semi}. We show that the same holds for boxworld processes. 

\begin{theorem}[Realization of causal correlations in boxworld]\label{thm::causalcorrelations}
All causally ordered correlations are boxworld correlations that can be achieved by causally order boxworld processes of the same order. That is, for every $P^{A\prec B}_{\A\B|\X\Y}$, there exist sets of local operations $T^{A|X}$ and $T^{B|Y}$, and a causally ordered boxworld process $W^{A\prec B}$, such that 
\begin{equation}
    P^{A\prec B}_{\A\B|\X\Y} = (T^{\A|\X}\otimes T^{\B|\Y}) *W^{A\prec B},
\end{equation}
and analogously for causally order correlations of the type $P^{B\prec A}_{\A\B|\X\Y}$. Similarly, all nonsignaling correlations $P^\text{NS}_{\A\B|\X\Y}$ can be generated from local operations and nonsignaling boxworld processes $W^{A||B}$, and causal correlations $P^\text{causal}_{\A\B|\X\Y}$ from causally separable boxworld processes $W^\text{sep}.$
\end{theorem}

\begin{proof}
We start with causally ordered correlations $P^{A\prec B}_{\A\B|\X\Y}$, which respect Eq.~\eqref{eq::probNS_AprecB}.
Such probability distributions can always be decomposed~\cite{hoffmann2018structure} as
\begin{equation}\label{eq:P_seq_fact}
    P^{A\prec B}_{\A\B|\X\Y} =P_{\A|\X}P_{\B|\A\X\Y},
\end{equation}
where $P_{\A|\X}\coloneqq\sum_\B P^{A\prec B}_{\A\B|\X\Y}$ and $P_{\B|\A\X\Y}\coloneqq P^{A\prec B}_{\A\B|\X\Y}/P_{\A|\X}$. 
Furthermore, we take $P_{\A|\X}$ in its decomposition in terms of deterministic distributions, i.e., $P_{\A|\X} =  \sum_\lambda \pi_\lambda\,\delta_{\A,f_\lambda(\X)}$, for some functions $\{f_\lambda\}_\lambda$. It is helpful at this point to notice that Eq.~\eqref{eq:P_seq_fact} can be rewritten as $P_{A|X}P_{\B|\A\X\Y}=\sum_\lambda \pi_\lambda\,\delta_{\A,f_\lambda(\X)} P_{\B|\A\X\Y} = \sum_\lambda \pi_\lambda\,\delta_{\A,f_\lambda(\X)} P_{\B|f_\lambda(\X) \X\Y}$, since the only terms surviving the summation are those for which $f_\lambda(\X)=A$. 

Take the following local transformations $T^{\A|\X\,(A\prec B)}$ and $T^{\B|\Y\,(A\prec B)}$, constructed from $P_{\A|\X}$ and $P_{\B|\A\X\Y}$ as
\begin{align}
    T^{\A|\X\,(A\prec B)} &= \sum_\lambda \pi_\lambda\, \delta_{\Ia,\phi}\,\delta_{\Oap,(X,f_\lambda(\X))}\,\delta_{\A,f_\lambda(\X)} \label{eq::TAX_AbeforeB} \\
    T^{\B|\Y\,(A\prec B)} &= \delta_{\Ib,\phi}\,\delta_{\Obp,\phi}\, P_{\B|\Ibp\Y}, \label{eq::TBY_AbeforeB}
\end{align}
where $\phi$ is a constant, and where the variables $\Ibp$ and $\Oap$ encode a pair of variables with the same cardinality as the pair $(A,X)$, i.e., $|\Oap|=|\Ibp|=|\A|\cdot|\X|$, and the distribution $P_{\B|\Ibp\Y}$ is defined from the proper relabeling $P_{\B|\Ibp\Y}=P_{\B|\A\X\Y}$. It is straightforward to check that these are valid transformations, by comparison with Eq.~\eqref{eq::op_probabilistic}.
Now take the following causally ordered boxworld process
\begin{equation}
    W^{A\prec B} = \delta_{\Iap,\phi}\,\delta_{\Oa,\phi}\,\delta_{\Ibp,\Oap}\,\delta_{\Ob,\phi}, \label{eq::W_AbeforeB}
\end{equation}
which is essentially an identity channel from $\Oap$ to $\Ibp$. By direct insertion into Eqs.~\eqref{eq::W_AprecB_positivity}--\eqref{eq::W_AprecB_proj4}, it can be checked that this is a valid causally ordered boxworld process.

It is then straightforward to verify that
\begin{align}
    T^{\A|\X\,(A\prec B)} * T^{\B|\Y\,(B\prec A)} * W^{A\prec B} &= \sum_{\substack{\Iap\Ia\Oa\Oap\\\Ibp\Ib\Ob\Obp}} \sum_\lambda \pi_\lambda\, \delta_{\Ia,\phi}\,\delta_{\Oap,(X,f_\lambda(\X))}\,\delta_{\A,f_\lambda(\X)}\, \delta_{\Ib,\phi}\,\delta_{\Obp,\phi}\,\nonumber \\
    & \hspace{2.8cm} \times P_{\B|\Ibp\Y} \,\delta_{\Iap,\phi}\,\delta_{\Oa,\phi}\,\delta_{\Ibp,\Oap}\,\delta_{\Ob,\phi}  \\ 
    &= \sum_{\Oap\Ibp} \sum_\lambda \pi_\lambda\,  \delta_{\Oap,(X,f_\lambda(\X))}\,\delta_{\A,f_\lambda(\X)}\,P_{\B|\Ibp\Y} \,\delta_{\Ibp,\Oap} \\
    %&= \sum_{\Ibp} \sum_\lambda \pi_\lambda\, \delta_{\A,f_\lambda(\X)}\,P_{\B|\Ibp\Y} \,\delta_{\Ibp,(X,f_\lambda(\X))} \\
    &= \sum_\lambda \pi_\lambda\, \delta_{\A,f_\lambda(\X)}\,P_{\B|\X f_\lambda(\X)\Y}=P_{\A|\X}P_{\B|\A\X\Y}=P^{A\prec B}_{\A\B|\X\Y}. 
\end{align}

The proof for $P^{B\prec A}_{\A\B|\X\Y}$ is analogous. 

Very similarly, we can prove that all nonsignaling correlations can be generated by performing local operations on a nonsignaling boxworld process. 
Given the nonsignaling correlations $P^\text{NS}_{\A\B|\X\Y}$, which satisfy Eqs.~\eqref{eq::probNS_AprecB} and~\eqref{eq::probNS_BprecA},
one can construct the following nonsignaling boxworld process
\begin{equation}
    W^{A||B} = \delta_{\Iap,\phi}\,\delta_{\Ibp,\phi}\,P^\text{NS}_{\Oa\Ob|\Ia\Ib},
\end{equation}
where, following a relabeling of the variables $\A\mapsto\Oa$, $\B\mapsto\Ob$, $\X\mapsto\Ia$, and $\Y\mapsto\Ib$, one has that $P^\text{NS}_{\Oa\Ob|\Ia\Ib}=P^\text{NS}_{\A\B|\X\Y}$. Notice that this boxworld process is very similar to the one in Eq.~\eqref{eq::ansatzW1}, except that here the set of probability distributions $P^\text{NS}_{\Oa\Ob|\Ia\Ib}$ is nonsignaling, ensuring that $W^{A||B}$ is a valid nonsignaling boxworld process and satisfies the NSWSE principle.

Then, take the following local operations
\begin{align}
    T^{\A|\X\,(A||B)} &= \delta_{\Ia,\X}\,\delta_{\Oap,\phi}\,\delta_{\A|\Oa} \\
    T^{\B|\Y\,(A||B)} &= \delta_{\Ib,\Y}\,\delta_{\Obp,\phi}\,\delta_{\B|\Ob}.
\end{align}

It is straightforward to check that 
\begin{align}
    T^{\A|\X\,(A||B)} * T^{\B|\Y\,(A||B)} * W^{A||B} &= \sum_{\substack{\Iap\Ia\Oa\Oap\\\Ibp\Ib\Ob\Obp}} \delta_{\Ia,\X}\,\delta_{\Oap,\phi}\,\delta_{\A,\Oa}\,\delta_{\Ib,\Y}\,\delta_{\Obp,\phi}\,\delta_{\B,\Ob}\,\delta_{\Iap,\phi}\,\delta_{\Ibp,\phi}\,P^\text{NS}_{\Oa\Ob|\Ia\Ib}\\
    &=\sum_{\substack{\Ia\Oa\\\Ib\Ob}} \delta_{\Ia,\X}\,\delta_{\A,\Oa}\,\delta_{\Ib,\Y}\,\delta_{\B,\Ob}\,P^\text{NS}_{\Oa\Ob|\Ia\Ib} \\
    &=P^\text{NS}_{\A\B|\X\Y}
\end{align}
This construction amounts to Alice and Bob simply reading the nonsignaling correlations $P^\text{NS}_{\A\B|\X\Y}$ that are prepared by the nonsignaling boxworld process itself.

For causal correlations $P^\text{causal}_{\A\B|\X\Y}$, which can be decomposed as
\begin{equation}
    P^\text{causal}_{\A\B|\X\Y} = q\,P_{\A|\X}\,P_{\B|\A\X\Y} + (1-q)P_{\B|\Y}\,P_{\A|\B\X\Y},
\end{equation}
a similar proof can be achieved by extending the transformations and processes used in the proofs of $P^{A\prec B}_{\A\B|\X\Y}$ and $P^{B\prec A}_{\A\B|\X\Y}$ by an extra `flag' random variable $\alpha$ for Alice and $\beta$ for Bob, according to
\begin{align}
    T^{A|X} &= T^{\A|\X\,(A\prec B)}\delta_{\alpha,0} + T^{\A|\X\,(B\prec A)}\delta_{\alpha,1} \\
    T^{B|Y} &= T^{\B|\Y\,(A\prec B)}\delta_{\beta,0} + T^{\B|\Y\,(B\prec A)}\delta_{\beta,1} \\
    W^\text{sep} &= q\, W^{A\prec B}\delta_{\alpha,0}\delta_{\beta,0} + (1-q)W^{B\prec A}\delta_{\alpha,1}\delta_{\beta,1}.
\end{align}

\end{proof}

%%%%%%%%%%%%%%
\subsection{Violating causal inequalities with boxworld correlations}\label{app::inequalities}

In this section we show that boxworld correlations can violate causal inequalities. We start by focusing on two renowned inequalities: the ``guess your neighbor's input'' (GYNI) inequality~\cite{almeida2010guess} and the ``lazy guess your neighbor's input'' (LGYNI) inequality~\cite{branciard2016simplest}. These inequalities correspond to the nontrivial facets of the causal polytope in the scenario where $|\A|=|\B|=|\X|=|\Y|=2$, the correlations scenario with two parties, and two inputs and two outputs per party. We also analyze the ``Oreshkov-Costa-Brukner'' (OCB) inequality~\cite{oreshkov2012quantum}, which is an inequality in the scenario $|\A|=|\B|=|\X|=|\Y|=|\Y'|=2$, that is, in which Bob has two bits as input, hereon denoted as $\Y$ and $\Y'$, for a total of 4 different input values. The OCB inequality does not correspond to a facet of the causal polytope, but nevertheless it has its causal bound violated by process matrix correlations, just as the GYNI and LGYNI inequalities.

We recall the definition of these inequalities.
The GYNI inequality is defined as 
\begin{equation}\label{eq::gyni_app}
    \text{GYNI}(P_{\A\B|\X\Y})\coloneqq\frac{1}{4}\sum_{\A,\B,\X,\Y} \delta_{\A,\Y}\,\delta_{\B,\X} \, P_{\A\B|\X\Y} \leq \frac{1}{2},
\end{equation}
and the LGYNI inequality is defined as 
\begin{equation}\label{eq::lgyni_app}
    \text{LGYNI}(P_{\A\B|\X\Y})\coloneqq\frac{1}{4}\sum_{\A,\B,\X,\Y} \delta_{\X(\A\oplus\Y),0}\,\delta_{\Y(\B\oplus\X),0}\, P_{\A\B|\X\Y} \leq \frac{3}{4}.
\end{equation}
Finally, the OCB inequality is defined as
\begin{equation}\label{eq::ocb_app}
    \text{OCB}(P_{\A\B|\X\Y\Yp})\coloneqq\frac{1}{8}\sum_{\A,\B,\X,\Y, \Yp} \delta_{(\Yp\oplus 1)(\A\oplus\Y),0}\,\delta_{\Yp(\B\oplus\X),0}\, P_{\A\B|\X\Y\Yp} \leq \frac{3}{4}.
\end{equation}
All these inequalities are satisfied by all causal correlations and can be violated by some correlations generated by process matrices~\cite{oreshkov2012quantum,branciard2016simplest}. Here we show that all three inequalities are also violated by boxworld correlations.

We have found that boxworld correlations can achieve the following bounds for these expressions:
\begin{align}
    \max_{\mathcal{P}_\text{BW}}\, \text{GYNI}(\overline{P}_{\A\B|\X\Y}) &\geq \frac{3}{4} = 0.75 \\
    \max_{\mathcal{P}_\text{BW}}\, \text{LGYNI}(\overline{P}_{\A\B|\X\Y}) &\geq \frac{11}{12} \approx 0.9167 \\
    \max_{\mathcal{P}_\text{BW}}\, \text{OCB}(\overline{P}_{\A\B|\X\Y\Yp}) &= 1,
\end{align}
where $\mathcal{P}_\text{BW}$ is the polytope of boxworld correlations.

We now present the local operations and boxworld processes that generate the boxworld correlations that can achieve such violation.

In all cases, the spaces~$\Iap,\Ibp$ are irrelevant and can be taken to be trivial. For the GYNI inequality, we present a violation of~$2/3$ where all remaining spaces are binary, and a violation of~$3/4$ where the spaces~$\Oap,\Ob,\Obp$ are augmented to~$\{0,1,2\}$. For the violation of the LGYNI inequality (value $11/12$), and for the violation of the OCB inequality (maximum violation), all spaces are binary. After the presentation of the local operations and boxworld processes used to show these results, we give an extended outline for how we found these processes.

%%%
\subsubsection{GYNI}\label{subapp::gyni}
%%%

The value of~$2/3$ for GYNI is obtained by a process that, conditioned on~$\Oap,\Obp$ produces a~$(\tfrac{1}{3},\tfrac{2}{3})$ mixture of a constant box~$P_{\Oa\Ob|\Ia\Ib}=P_{\Oa\Ob}$ and a PR box~\cite{tsirelson1993some,rastall1985locality,popescu1994quantum}, given by 
\begin{align}
    W_\text{GYNI}
    &=
    \frac{1}{3}\
    \delta_{\Oa,\Obp}\ \delta_{\Ob,\Oap}
    +
    \frac{2}{3}
    \delta_{(\Oa\oplus\Ob\oplus\Oap\oplus\Obp\oplus 1), ((\Ia\oplus\Oap\oplus 1)(\Ib\oplus\Obp\oplus1))}/2, \label{eq::gyni2/3_W}\\
\end{align}
and the following local operations:
\begin{align}
    T^{A|X}_\text{GYNI} &= \delta_{\Ia,\X}\ \delta_{\Oap,\X}\ \delta_{\A,\Oa} \label{eq::gyni2/3_TAX} \\
    T^{B|Y}_\text{GYNI} &= \delta_{\Ib,\Y}\ \delta_{\Obp,\Y}\ \delta_{\B,\Ob}. \label{eq::gyni2/3_TBY}
\end{align}

Combined, they lead to the correlations
\begin{align}
 \overline{P}_{\A\B|\X\Y}^\text{\ (GYNI=2/3)}
  =\
  \frac{1}{3}\
  \delta_{\A,\Y}\ \delta_{\B,\X}
  +
  \frac{2}{3}\ 
  \left(
    \frac{\delta_{\A,\Y}\ \delta_{\B,\X}}{2}
    +
    \frac{\delta_{\A\oplus 1,\Y}\ \delta_{\B\oplus 1,\X}}{2}
  \right)
  \,.
\end{align}
    
For the same inequality, the value of~$3/4$ is obtained using the following boxworld process and local operations. Here, due to the complexity introduced by the ternary spaces, we represent the boxworld process by a probability table for each~$\Oap,\Obp$ (see Table~\ref{table:GYNIprocess}).
    
\begin{table}[h!]
    \centering
    \begin{tabular}{ccc}
        \begin{tabular}{|c|c|c|c|} \hline
          \makebox[2em]{$1/4$} & \makebox[2em]{$1/2$} & \makebox[2em]{$0$} & \makebox[2em]{$1/2$} \\ \hline
          \makebox[2em]{$1/4$} & \makebox[2em]{$0$} & \makebox[2em]{$1/2$} & \makebox[2em]{$0$} \\ \hline
          \makebox[2em]{$0$} & \makebox[2em]{$0$} & \makebox[2em]{$0$} & \makebox[2em]{$0$} \\ \hline
          \makebox[2em]{$1/4$} & \makebox[2em]{$1/2$} & \makebox[2em]{$1/2$} & \makebox[2em]{$1/2$} \\ \hline
          \makebox[2em]{$1/4$} & \makebox[2em]{$0$} & \makebox[2em]{$0$} & \makebox[2em]{$0$} \\ \hline
          \makebox[2em]{$0$} & \makebox[2em]{$0$} & \makebox[2em]{$0$} & \makebox[2em]{$0$} \\ \hline
        \end{tabular} &\hspace{3em}
        \begin{tabular}{|c|c|c|c|} \hline
          \makebox[2em]{$1/4$} & \makebox[2em]{$1/2$} & \makebox[2em]{$1/2$} & \makebox[2em]{$1$} \\ \hline
          \makebox[2em]{$1/4$} & \makebox[2em]{$0$} & \makebox[2em]{$1/2$} & \makebox[2em]{$0$} \\ \hline
          \makebox[2em]{$0$} & \makebox[2em]{$0$} & \makebox[2em]{$0$} & \makebox[2em]{$0$} \\ \hline
          \makebox[2em]{$1/4$} & \makebox[2em]{$1/2$} & \makebox[2em]{$0$} & \makebox[2em]{$0$} \\ \hline
          \makebox[2em]{$1/4$} & \makebox[2em]{$0$} & \makebox[2em]{$0$} & \makebox[2em]{$0$} \\ \hline
          \makebox[2em]{$0$} & \makebox[2em]{$0$} & \makebox[2em]{$0$} & \makebox[2em]{$0$} \\ \hline
        \end{tabular} &\hspace{3em}
        \begin{tabular}{|c|c|c|c|} \hline
          \makebox[2em]{$1/4$} & \makebox[2em]{$1/2$} & \makebox[2em]{$1/2$} & \makebox[2em]{$1/2$} \\ \hline
          \makebox[2em]{$1/4$} & \makebox[2em]{$0$} & \makebox[2em]{$0$} & \makebox[2em]{$0$} \\ \hline
          \makebox[2em]{$0$} & \makebox[2em]{$0$} & \makebox[2em]{$0$} & \makebox[2em]{$0$} \\ \hline
          \makebox[2em]{$1/4$} & \makebox[2em]{$1/2$} & \makebox[2em]{$0$} & \makebox[2em]{$1/2$} \\ \hline
          \makebox[2em]{$1/4$} & \makebox[2em]{$0$} & \makebox[2em]{$1/2$} & \makebox[2em]{$0$} \\ \hline
          \makebox[2em]{$0$} & \makebox[2em]{$0$} & \makebox[2em]{$0$} & \makebox[2em]{$0$} \\ \hline
        \end{tabular}
        \\ \\
        \begin{tabular}{|c|c|c|c|} \hline
          \makebox[2em]{$1/4$} & \makebox[2em]{$0$} & \makebox[2em]{$0$} & \makebox[2em]{$0$} \\ \hline
          \makebox[2em]{$1/4$} & \makebox[2em]{$0$} & \makebox[2em]{$1/2$} & \makebox[2em]{$0$} \\ \hline
          \makebox[2em]{$0$} & \makebox[2em]{$1/2$} & \makebox[2em]{$0$} & \makebox[2em]{$1/2$} \\ \hline
          \makebox[2em]{$1/4$} & \makebox[2em]{$1/2$} & \makebox[2em]{$1/2$} & \makebox[2em]{$1/2$} \\ \hline
          \makebox[2em]{$1/4$} & \makebox[2em]{$0$} & \makebox[2em]{$0$} & \makebox[2em]{$0$} \\ \hline
          \makebox[2em]{$0$} & \makebox[2em]{$0$} & \makebox[2em]{$0$} & \makebox[2em]{$0$} \\ \hline
        \end{tabular} &\hspace{3em}
        \begin{tabular}{|c|c|c|c|} \hline
          \makebox[2em]{$1/4$} & \makebox[2em]{$0$} & \makebox[2em]{$1/2$} & \makebox[2em]{$1/2$} \\ \hline
          \makebox[2em]{$1/4$} & \makebox[2em]{$0$} & \makebox[2em]{$1/2$} & \makebox[2em]{$0$} \\ \hline
          \makebox[2em]{$0$} & \makebox[2em]{$1/2$} & \makebox[2em]{$0$} & \makebox[2em]{$1/2$} \\ \hline
          \makebox[2em]{$1/4$} & \makebox[2em]{$1/2$} & \makebox[2em]{$0$} & \makebox[2em]{$0$} \\ \hline
          \makebox[2em]{$1/4$} & \makebox[2em]{$0$} & \makebox[2em]{$0$} & \makebox[2em]{$0$} \\ \hline
          \makebox[2em]{$0$} & \makebox[2em]{$0$} & \makebox[2em]{$0$} & \makebox[2em]{$0$} \\ \hline
        \end{tabular} &\hspace{3em}
        \begin{tabular}{|c|c|c|c|} \hline
          \makebox[2em]{$1/4$} & \makebox[2em]{$0$} & \makebox[2em]{$1/2$} & \makebox[2em]{$0$} \\ \hline
          \makebox[2em]{$1/4$} & \makebox[2em]{$0$} & \makebox[2em]{$0$} & \makebox[2em]{$0$} \\ \hline
          \makebox[2em]{$0$} & \makebox[2em]{$1/2$} & \makebox[2em]{$0$} & \makebox[2em]{$1/2$} \\ \hline
          \makebox[2em]{$1/4$} & \makebox[2em]{$1/2$} & \makebox[2em]{$0$} & \makebox[2em]{$1/2$} \\ \hline
          \makebox[2em]{$1/4$} & \makebox[2em]{$0$} & \makebox[2em]{$1/2$} & \makebox[2em]{$0$} \\ \hline
          \makebox[2em]{$0$} & \makebox[2em]{$0$} & \makebox[2em]{$0$} & \makebox[2em]{$0$} \\ \hline
        \end{tabular}
        \\ \\
        \begin{tabular}{|c|c|c|c|} \hline
          \makebox[2em]{$1/4$} & \makebox[2em]{$1/2$} & \makebox[2em]{$0$} & \makebox[2em]{$0$} \\ \hline
          \makebox[2em]{$1/4$} & \makebox[2em]{$0$} & \makebox[2em]{$1/2$} & \makebox[2em]{$0$} \\ \hline
          \makebox[2em]{$0$} & \makebox[2em]{$0$} & \makebox[2em]{$0$} & \makebox[2em]{$1/2$} \\ \hline
          \makebox[2em]{$1/4$} & \makebox[2em]{$0$} & \makebox[2em]{$1/2$} & \makebox[2em]{$1/2$} \\ \hline
          \makebox[2em]{$1/4$} & \makebox[2em]{$0$} & \makebox[2em]{$0$} & \makebox[2em]{$0$} \\ \hline
          \makebox[2em]{$0$} & \makebox[2em]{$1/2$} & \makebox[2em]{$0$} & \makebox[2em]{$0$} \\ \hline
        \end{tabular} &\hspace{3em}
        \begin{tabular}{|c|c|c|c|} \hline
          \makebox[2em]{$1/4$} & \makebox[2em]{$1/2$} & \makebox[2em]{$1/2$} & \makebox[2em]{$1/2$} \\ \hline
          \makebox[2em]{$1/4$} & \makebox[2em]{$0$} & \makebox[2em]{$1/2$} & \makebox[2em]{$0$} \\ \hline
          \makebox[2em]{$0$} & \makebox[2em]{$0$} & \makebox[2em]{$0$} & \makebox[2em]{$1/2$} \\ \hline
          \makebox[2em]{$1/4$} & \makebox[2em]{$0$} & \makebox[2em]{$0$} & \makebox[2em]{$0$} \\ \hline
          \makebox[2em]{$1/4$} & \makebox[2em]{$0$} & \makebox[2em]{$0$} & \makebox[2em]{$0$} \\ \hline
          \makebox[2em]{$0$} & \makebox[2em]{$1/2$} & \makebox[2em]{$0$} & \makebox[2em]{$0$} \\ \hline
        \end{tabular} &\hspace{3em}
        \begin{tabular}{|c|c|c|c|} \hline
          \makebox[2em]{$1/4$} & \makebox[2em]{$1/2$} & \makebox[2em]{$1/2$} & \makebox[2em]{$0$} \\ \hline
          \makebox[2em]{$1/4$} & \makebox[2em]{$0$} & \makebox[2em]{$0$} & \makebox[2em]{$0$} \\ \hline
          \makebox[2em]{$0$} & \makebox[2em]{$0$} & \makebox[2em]{$0$} & \makebox[2em]{$1/2$} \\ \hline
          \makebox[2em]{$1/4$} & \makebox[2em]{$0$} & \makebox[2em]{$0$} & \makebox[2em]{$1/2$} \\ \hline
          \makebox[2em]{$1/4$} & \makebox[2em]{$0$} & \makebox[2em]{$1/2$} & \makebox[2em]{$0$} \\ \hline
          \makebox[2em]{$0$} & \makebox[2em]{$1/2$} & \makebox[2em]{$0$} & \makebox[2em]{$0$} \\ \hline
        \end{tabular}
        \\ \\
    \end{tabular}
    \caption{Boxworld process for obtaining the value~$3/4$ in GYNI. The table in row~$r$ and column~$c$ is for~$\Oap=r,\Obp=c$.
      Within each table, the columns represent the inputs~$(\Ia,\Ib)$, and the rows the outputs~$(\Oa,\Ob)$.}
    \label{table:GYNIprocess}
\end{table}
    
The corresponding local operations are
\begin{equation}\label{eq::GYNI3/4_Ts}
    \begin{tabular}{c|c}
        $\X\Oa\ $ & $\ \A\Oap\Ia$ \\ \hline
        $00$ & $110$ \\
        $01$ & $120$ \\
        $10$ & $101$ \\
        $11$ & $021$
    \end{tabular}
    \text{for Alice, and }
    \qquad\qquad
    \begin{tabular}{c|c}
        $ \Y\Ob\ $ & $\ \B\Obp\Ib $ \\ \hline
        $00$ & $100$ \\
        $01$ & $120$ \\
        $02$ & $000$ \\
        $10$ & $111$ \\
        $11$ & $111$ \\
        $12$ & $001$
    \end{tabular}
    \text{for Bob}.
\end{equation}

That boxworld process in Table~\ref{table:GYNIprocess} combined with these local operations produce the following correlations:
\begin{equation}
    \overline{P}_{\A\B|\X\Y}^\text{\ (GYNI=3/4)} = \delta_{\A,\X(\Y\oplus 1)\oplus 1}\delta_{\B,(\X\oplus 1)\Y\oplus 1}
    \,.
\end{equation}
Note that only for the case~$\X=\Y=0$ the values of~$\A$ and~$\B$ do not correspond to~$\Y$ and~$\X$.

%%%
\subsubsection{LGYNI}\label{subapp::lgyni}
%%%

For a violation of the LGYNI inequality, we can also express the local operations and boxworld process analytically. The boxworld process is a $(\tfrac{1}{3},\tfrac{2}{3})$ mixture of a box with a LGYNI functionality, with a box that behaves similar to a PR box:
\begin{equation}
    W_\text{LGYNI} = \frac{1}{3}P_\text{LGYNI}^\alpha + \frac{2}{3}P_\text{LGYNI}^\text{PR} \label{eq::LGYNI11/12_W}
    \,
\end{equation}
with
\begin{align}
    P_\text{LGYNI}^\alpha &= \delta_{\Oa,\Ia\Obp}\delta_{\Ob,\Ib\Oap}\\
    P_\text{LGYNI}^\text{PR} &=
    \begin{cases}
      \frac{1}{2} \delta_{\Oa\oplus\Ob,\Ia\Ib} & \Oap=0 \wedge \Obp = 0\\
      \frac{1}{2} \delta_{\Oa\oplus\Ob,\Ia(\Ib\oplus 1)} & \Oap=0 \wedge \Obp = 1\\
      \frac{1}{2} \delta_{\Oa\oplus\Ob,\Ia\oplus 1} & \Oap=1 \wedge \Obp = 0\\
      \frac{1}{2} \delta_{\Oa\oplus\Ob,\Ia\Ib\oplus 1} & \Oap=1 \wedge \Obp = 1
    \end{cases}
    \,.
\end{align}
The local operations are
\begin{align}
    T^{\A|\X}_\text{LGYNI} &= \delta_{\Ia,\X}\ \delta_{\Oap,\Oa}\ \delta_{\A,\X\Oa} \label{eq::LGYNI11/12_TAX}\\
    T^{\B|\Y}_\text{LGYNI} &= \delta_{\Ib,\Y}\ \delta_{\Obp,\Ob}\ \delta_{\B,\Y\Ob} \label{eq::LGYNI11/12_TBY}
    \,.
\end{align}
Together, they yield the correlations
\begin{equation}
    \overline{P}_{\A\B|\X\Y}^\text{\ (LGYNI=11/12)} = \frac{1}{3}\delta_{\A,0}\delta_{\B,0} + \frac{2}{3}\delta_{\A,\X\Y}\delta_{\B,\X\Y}
    \,.
\end{equation}
For the LGYNI inequality, this means that with probability~$1/3$ the value of~$3/4$ is obtained, and with probability~$2/3$ the algebraic maximum of~$1$ is obtained; in total we get the value~$11/12$.

%%%
\subsubsection{OCB}\label{subapp::ocb}
%%%

For the OCB inequality, the process is just a relabeling of the PR box for fixed~$\Oap,\Obp$:
\begin{equation}\label{eq::OCB1_W}
    W_\text{OCB} = \frac{1}{2} \delta_{\Oa\oplus\Ob\oplus\Obp, (\Ia \oplus \Oap \oplus \Obp)\Ib}
    \,,
\end{equation}
and the local transformations have a particular simple form:
\begin{align}
    T^{\A|\X}_\text{OCB} &= \delta_{\Ia,0}\ \delta_{\Oap,\Oa\oplus \X}\ \delta_{\A,\Oa} \label{eq::OCB1_TAX}  \\
    T^{\B|\Y\Yp}_\text{OCB} &= \delta_{\Ib,\Yp}\ \delta_{\Obp,\Ob \oplus \Y}\ \delta_{\B,\Ob}
    \,. \label{eq::OCB1_TBY}
\end{align}

The resulting correlations from this boxworld process and local operations is given in the form of Table~\ref{table:OCB_prob}.

\begin{table}[h!]
    \centering
    {\renewcommand{\arraystretch}{1.25}
    \begin{tabular}{|c|c|c|c|c|c|c|c|c|}
    \hline
    $(a,b) \diagdown (x,y,y')$& $\ (0,0,0)\ $ & $\ (0,0,1)\ $ & $\ (0,1,0)\ $ & $\ (0,1,1)\ $ & $\ (1,0,0)\ $ & $\ (1,0,1)\ $ & $\ (1,1,0)\ $ & $\ (1,1,1)\ $ \\
    \hline
    (0,0) & 1/2     & 1/2     & 0    & 1/2     & 1/2     & 0    & 0    & 0    \\
    (0,1) & 1/2     & 0    & 0    & 0    & 1/2     & 1/2     & 0    & 1/2     \\
    (1,0) & 0    & 1/2     & 1/2     & 1/2     & 0    & 0    & 1/2     & 0    \\
    (1,1) & 0    & 0    & 1/2     & 0    & 0    & 1/2     & 1/2     & 1/2     \\
    \hline
    \end{tabular}
}
\caption{Values for $\overline{P}_{\A\B|\X\Y\Yp}^\text{\ (OCB=1)}$ for the distribution saturating the OCB inequality. The first column and row indicate the values of $(a,b)$ and $(x,y,y')$, respectively. The entries of the table are the value of the corresponding probability.}
\label{table:OCB_prob}
\end{table}

In the next section, we discuss these correlations in more detail.

%%%
\subsubsection{Extended discussion}\label{subapp::extdiscussion}
%%%

In principle, one could exploit the polytope characterization of the set of boxworld correlations to find the maximum violation of these inequalities by boxworld correlations in any scenario with fixed local dimensions, i.e., a fixed cardinality of the boxworld process' variables. Let $\mathcal{P}_{T}$ be the polytope of all pairs of sets of local operations $T^{\A|\X}\otimes T^{\B|\Y}\in\mathcal{P}_{T}$, let $\mathcal{P}_{W}$ be the polytope of all boxworld processes $W\in\mathcal{P}_{W}$, and finally let $\mathcal{P}_{\rm BW}$ be the polytope of all boxworld correlations $\overline{P}_{\A\B|\X\Y}\in\mathcal{P}_{\rm BW}$. The maximization of the expressions in Eqs.~\eqref{eq::gyni_app},~\eqref{eq::lgyni_app}, and~\eqref{eq::ocb_app} corresponds to a maximization of a linear functional on the polytope $\mathcal{P}_{\rm BW}$. As such, we know that the maximum is achieved on the extremal points of the polytope and it is, in principle, sufficient to verify the value on all the extremal points. In the example of the GYNI inequality, this is equivalent to checking all extremal points in the polytopes $\mathcal{P}_{T}$ and $W\in\mathcal{P}_{W}$,
\begin{equation}
\begin{split}
    \max_{\overline{P}\in \mathcal{P}_{\rm BW}} \text{GYNI}(\overline{P}) &= \max \left\lbrace \frac{1}{4}\sum_{\A,\B,\X,\Y} \delta_{\A,\Y}\,\delta_{\B,\X}\, \overline{P}_{\A\B|\X\Y} \ \scalebox{2}{|}  \ \overline{P}_{\A\B|\X\Y}\in\mathcal{P}_{\rm BW} \right\rbrace\\
    &=\max \left\lbrace \frac{1}{4}\sum_{\A,\B,\X,\Y} \delta_{\A,\Y}\,\delta_{\B,\X}\, T^{\A|\X}\otimes T^{\B|\Y} * W\ \scalebox{2}{|}  \ T^{\A|\X},T^{\B|\Y}  \in \mathcal{P}_{T},  W \in \mathcal{P}_W\right\rbrace\\
    &=\max \left\lbrace \frac{1}{4}\sum_{\A,\B,\X,\Y} \delta_{\A,\Y}\,\delta_{\B,\X}\, T^{\A|\X}\otimes T^{\B|\Y} * W\ \scalebox{2}{|}  \ T^{\A|\X},T^{\B|\Y}, W \text{ extremal points}\right\rbrace.
\end{split}
\end{equation}

However, our characterization of the polytope of $\mathcal{P}_W$ of boxworld processes $W$ is given in terms of its facets, rather than its extremal points, and the numerical conversion between the two representations, even in the all-bit scenario, is too costly. 

An alternative approach could be to use the hyperplane characterization of $\mathcal{P}_W$. 
More precisely, for each extreme set of extreme local operations $T^{\A|\X}, T^{\B|\Y}$ we solve the linear program (LP)
\begin{equation}\label{eq:W_LP_T_fix}
\begin{split}
    \textbf{\text{given}} \ \ \  & \{T^{\A|\X}, T^{\B|\Y}\} \text{ extremal} \\
    \textbf{\text{max}} \ \ \  &\frac{1}{4}\sum_{\A,\B,\X,\Y} \delta_{\A,\Y}\,\delta_{\B,\X} \, T^{\A|\X}\otimes T^{\B|\Y} * \W\,\\
    \textbf{\text{subject to}}  \ \ \  &\W\geq 0,\\
    & \W \text{ satisfying Eqs.~\eqref{eq::bwW_positivity_app}--\eqref{eq::bwW_proj5_app}.}
\end{split}
\end{equation}

One can estimate the number of LPs needed to compute the maximum of GYNI for the all-bit scenario to be $1024^4 \approx 10^{12}$, where $1024$ is the number of extreme local operations $\{T^{\A|\X}\}_{\A}$. This is clearly too computationally expensive. One can still get some intuition about the optimal values for the all-bit case, by performing the optimization under the assumption that Alice and Bob perform exactly the same operations. Notice that, even if the problem is symmetric under the exchange $\A\leftrightarrow \B$, it is not guaranteed that the optimal solution is achieved by a symmetric deterministic operation for Alice and Bob. Therefore, we have no guarantee that this is indeed the optimal solution for the all-bit scenario. However, these assumptions reduce the number of LPs to approximately $10^6$, which we can compute.
\\

\textbf{GYNI.} The result obtained for the maximum violation of the GYNI inequality in the all-bit scenario, under the assumption of identical operations for Alice and Bob, is exactly~$2/3$, and is attained by the local operations and boxworld process described in Eqs.~\eqref{eq::gyni2/3_W}--\eqref{eq::gyni2/3_TBY}. The action of this process [Eq.~\eqref{eq::gyni2/3_W}] can be understood in the following way. If Alice and Bob provide~$\Oap$,~$\Obp$ to the process, it responds with a boxworld system on the~$\Oa,\Ob,\Ia,\Ib$ spaces. This system is in a ($\tfrac{1}{3},\tfrac{2}{3})$ mixture of the following two states. With probability~$1/3$, the system returned is in the state that independently of~$\Ia,\Ib$, returns~$\Oa=\Obp$ and~$\Ob=\Oap$. Thus, the first term in Eq.~\eqref{eq::gyni2/3_W} describes a constant that yields, if we take into account the local strategies, the desired output for Alice and Bob.

In more detail: Alice provides~$\X$ to the space~$\Oap$ and with probability~$1/3$ Bob receives Alice's input~$\X$ on his~$\Ob$ space. Likewise, Alice receives Bob's input~$\Y$ on her~$\Oa$ space with probability~$1/3$. The second term in Eq.~\eqref{eq::gyni2/3_W} describes a variant of the PR box~\cite{tsirelson1993some,rastall1985locality,popescu1994quantum}. In its most simple form, the PR box provides an input-output behavior where the parity of the outputs equals the product of the inputs. Since for each input pair two possible output pairs that satisfy this PR condition exist, the PR box responds equally to each with half probability. In our case, the same happens: Depending on~$\Oap$ and~$\Obp$, the state of the system on the~$\Ia,\Ib$ and~$\Oa,\Ob$ spaces is, with probability~$2/3$, a parameterized PR box.

It turns out, however, that the value for GYNI can be improved by going to the all-trit case, i.e., to $d=3$. In that case, via a {\it seesaw optimization} (which alternates the optimization of $W$ against the transformations $T^{\A|\X}$ and $T^{\B|\Y}$, with each step in the optimization being a linear program) we obtained a value of~$3/4$ attained by the local operations in Eq.~\eqref{eq::GYNI3/4_Ts} and boxworld process in Table~\ref{table:GYNIprocess}. Extending the heuristic seesaw optimization method to the case $d=4$ did not find a better solution, however, this does not guarantee that a higher violation does not exist in $d\geq4$.
\\

\textbf{LGYNI.} In the case of the LGYNI inequality, the value of~$11/12$ attained by the local operations and boxworld process described in Eqs.~\eqref{eq::LGYNI11/12_W}--~\eqref{eq::LGYNI11/12_TBY} was also found as the maximum violation in the all-bit scenario, under the assumption of identical operations for Alice and Bob. We were not able to improve the bound $11/12$ for the LGYNI by going to higher dimensions using a seesaw optimization. Note that for process matrix correlations, the maximum value for the LGYNI inequality is attained by a process matrix of local system dimension $d=2$, that is, the all-qubit scenario~\cite{branciard2016simplest,liu2024tsirelson}. For boxworld correlations, we did not find in our search any improvement by going to higher dimension as well.
\\

\textbf{OCB.} Finally, we also applied a seesaw optimization to find the boxworld violation of the OCB inequality that reached the algebraic maximum of~$1$, already in the all-bit scenario. Notice that winning the OCB game with probability one does not necessarily imply the presence of perfect two-way signaling. In fact, even if $P(\A=\Y|\Yp = 0)=P(\B=\X|\Yp=1)=1$, the winning distribution is not deterministic, as can be seen from the correlations $\overline{P}_{\A\B|\X\Y\Yp}^\text{\ (OCB=1)}$ in Table~\ref{table:OCB_prob}. In particular, such a distribution does not allow for perfect two-way signaling, as the marginals over $Y'$ gives $P(\A=\Y)=P(\B=\X)=3/4$. Moreover, it provides a score of  $7/16= 0.4375$ in the GYNI game, i.e., even below the causal bound.
\\

\textbf{Comparison to process matrix correlations.} We can compare the violations of these causal inequalities obtained by boxworld correlations to the violations that are obtained by process matrices. In the work of Ref.~\cite{liu2024tsirelson}, upper bounds for the violation of these three inequalities by correlations that can be generated by process matrices of any dimension have been derived.

Let $\mathcal{S}_\text{QW}$ be the set of all process matrix correlations. It has been shown in Ref.~\cite{liu2024tsirelson} that the maximum obtained by process matrices for the expressions GYNI, LGYNI, and OCB, is given by
\begin{align}
    \max_{\mathcal{S}_\text{QW}}\, \text{GYNI}(P_{\A\B|\X\Y}) &\leq 0.7592 \\ 
    \max_{\mathcal{S}_\text{QW}}\, \text{LGYNI}(P_{\A\B|\X\Y}) &\approx 0.8194 \\ 
    \max_{\mathcal{S}_\text{QW}}\, \text{OCB}(P_{\A\B|\X\Y\Yp}) &= \frac{2+\sqrt{2}}{4}.
\end{align}
It is not known whether the bound for GYNI is tight; on the other hand, both for the numerical value for LGYNI (written approximately) and for the value for OCB (which is exact) \textit{are} known to be tight. On comparing these results with the violation of these inequalities that we find in boxworld, we find that 
\begin{enumerate}
    \item Boxworld correlations attain a higher violation of the LGYNI inequality than any process matrix correlations.
    \item Boxworld correlations attain a higher violation of the OCB inequality than any process matrix correlations, reaching the algebraic maximum.
    \item It is unclear whether boxworld correlations attain a higher violation of the GYNI inequality than any process matrix correlations, since the maximum value for this expression found in boxworld is not greater than the upper bound for process matrices, which in turn may not be tight.
\end{enumerate}

Based on the fact that both for the LGYNI and the OCB inequalities, boxworld correlations reach a higher violation than process matrix correlations, it seems likely that this is also the case for the GYNI inequality, and that the bound on Ref.~\cite{liu2024tsirelson} might not be tight. It could also be that in higher-dimensions, the GYNI violation attained in boxworld might increase beyond the $3/4$, however our numerical calculations show no evidence of this.

Given these results we expect that the polytope of boxworld correlations is a nontrivial outer approximation of the set of process matrix correlations.

\end{document}